\documentclass[preprint, 10pt]{elsarticle}
\usepackage[utf8]{inputenc}

\usepackage{graphicx}
\usepackage{amssymb}

\usepackage{url}            
\usepackage{booktabs}       
\usepackage{amsthm}
\usepackage{amsmath}\usepackage{amsfonts}       
\usepackage{amsfonts}       
\usepackage{mathabx} 
\usepackage{multirow}
\usepackage{geometry}
\usepackage{caption}

\usepackage[labelformat=simple]{subcaption}

\usepackage{afterpage}

\usepackage{tikz-cd}
\usepackage{rotating}
\usepackage[Symbol]{upgreek} 

\usepackage[unicode = true,
            colorlinks = true,
            linkcolor = black,
            urlcolor  = blue,
            citecolor = black,
            anchorcolor = black]{hyperref}

\usepackage{algorithm}      
\usepackage{algorithmic}    

\theoremstyle{plain}
\newtheorem{proposition}{Proposition}[section]
\newtheorem{theorem}[proposition]{Theorem}
\newtheorem{corollary}[proposition]{Corollary}
\newtheorem{lemma}[proposition]{Lemma}

\theoremstyle{definition}
\newtheorem{definition}[proposition]{Definition}

\theoremstyle{remark}
\newtheorem*{remark}{Remark}

\newenvironment{proofsketch}{%
  \proof}{\endproof} 
  
\newcommand{\din}{d_{\infty}}
\newcommand{\linf}{\ell_{\infty}}

\newcommand{\diamInfty}[1]{\textrm{diam}_{\infty}({#1})}

\newcommand{\linkY}{\bar{\mathcal{Y}}}
\newcommand{\openY}{\mathring{{\mathcal{Y}}}}

\newcommand{\KC}[1]{K_{#1}^{\check{C}}}
\newcommand{\KAF}[1]{K_{#1}^{AF}}
\newcommand{\LC}[1]{L_{#1}^{\check{C}}}

\newcommand{\tauC}{\check{\tau}}

\newcommand{\yC}{\check{y}}
\newcommand{\YC }{\check{\mathcal{Y}}}

\newcommand{\bisector}[1]{\textrm{bis}_{#1}}

\newcommand{\cechri}{K_{r_i}^{\check{C}}}

\newcommand{\flagri}{K_{r_i}^{AF}}

\newcommand{\cechr}{K_r^{\check{C}}}
\newcommand{\alphar}{K_r^A}
\newcommand{\ripsr}{K_r^{VR}}
\newcommand{\minir}{K_r^M}
\newcommand{\flagr}{K_r^{AF}}

\newcommand{\cechre}{K_{r + \varepsilon}^{\check{C}}}

\newcommand{\flagre}{K_{r + \varepsilon}^{AF}}

\newcommand{\Falpha}{K_{\mathcal{R}}^A}
\newcommand{\Fcech}{K_{\mathcal{R}}^{\check{C}}}
\newcommand{\Frips}{K_{\mathcal{R}}^{VR}}
\newcommand{\Fflag}{K_{\mathcal{R}}^{AF}}
\newcommand{\Fmini}{K_{\mathcal{R}}^M}

\newcommand{\modulek}{\text{M}_k(K_{\mathcal{R}})}

\newcommand{\dgmk}{\text{Dgm}_k}

\newcommand{\cl}{\text{Cl}}
\newcommand{\st}{\text{St}}

\newcommand{\A}{A_{e}}
\newcommand{\cB}[2]{\overline{B_{#1}(#2)}}

\newcommand{\minipq}{\text{Mini}_{pq}}

\newcommand{\nrv}{\textrm{Nrv}}
\newcommand{\R}{\mathbb{R}}
\newcommand{\Z}[1]{\mathcal{Z}_{#1}}

\newcommand{\boundMiniEdges}{O\big(2^{d-1} n\ln^{d-1}(n)\big)}

\newcommand{\minitwodim}{\textrm{Mini}_{p_{xy}q_{xy}}}

\newcommand{\finalcomplex}{\cl(\st(e)) \setminus \st(e)}

\newcommand{\nrvOpenY}{\nrv\left( \{\cB{\bar{r}}{y}\}_{y \in \openY} \right)}
\newcommand{\openX}{\mathring{\mathcal{X}}}

\begin{document}

\begin{frontmatter}

\title{\texorpdfstring{Persistent Homology in $\linf$ Metric}
    {Persistent Homology in l∞ Metric}}
\tnotetext[t1]{This work was partially funded by the School of Mathematical Sciences at Queen Mary University of London and by the SSHRC-NFRF and DSTL/Turing Institute grant DS-015.}

\author{Gabriele Beltramo}\corref{cor1}
\ead{g.beltramo@qmul.ac.uk}

\author{Primoz Skraba}
\ead{p.skraba@qmul.ac.uk}

\cortext[cor1]{Corresponding author}

\address{School of Mathematical Sciences\\
Queen Mary University of London, London, E1 4NS}

\journal{COMP GEOM-THEOR APPL}

\begin{abstract}
Proximity complexes and filtrations are central constructions in topological data analysis. 
Built using distance functions, or more generally metrics, they are often used to infer connectivity information from point clouds. 
Here we investigate proximity complexes and filtrations built over the Chebyshev metric, also known as the maximum metric or $\ell_{\infty}$ metric, rather than the classical Euclidean metric.  
Somewhat surprisingly, the $\linf$ case has not been investigated thoroughly.
In this paper, we examine a number of classical complexes under this metric, including the \v{C}ech, Vietoris-Rips, and Alpha complexes.
We define two new families of flag complexes, which we call the Alpha flag and Minibox complexes, and prove their equivalence to \v{C}ech complexes in homological degrees zero and one.
Moreover, we provide algorithms for finding Minibox edges of two, three, and higher-dimensional points.
Finally, we present computational experiments on random points, which shows that Minibox filtrations can often be used to speed up persistent homology computations in homological degrees zero and one by reducing the number of simplices in the filtration. 
\end{abstract}


\end{frontmatter}


\section{Introduction}
\label{sec:intro}
Topological data analysis (TDA) has been the subject of intense research over the last decade~\cite{carlsson2009topology, chazal2017introduction, edelsbrunner2008persistent}. 
Persistent (co)homology is by far the most studied and popular algebraic invariant considered in TDA. This is an invariant which is assigned to a filtration -- an increasing sequence of spaces. A common filtration arises from the sub-level sets of the distance to a finite sample of a space under consideration. Most commonly, the finite sample is on or near a manifold embedded in Euclidean space, $\R^d$. 
In the standard Euclidean setting, the \v{C}ech and the Alpha filtrations~\cite{bauer2014morse, chazal2014persistencestability, edelsbrunner1994three} directly capture the topology of the corresponding sub-level sets. 
Relatedly, the Vietoris-Rips filtration~\cite{vietoris1927hoheren} provides an approximation to this topology. 
In particular, the corresponding filtrations in Euclidean space may be related via a sandwiching argument~\cite{ghrist2008barcodes}.

In this paper, we study the \v{C}ech persistent homology of a finite set of points $S$ in a $\linf$ metric space.
Given $n$ points in
a $d$-dimensional space, the \v{C}ech filtration has $\Theta(n^{d+1})$ simplices. 
In the Euclidean setting, the number of simplices to be considered can be reduced by using Alpha filtrations, restricting simplices to those of the Delaunay triangulation of $S$. Furthermore, the Alpha filtration is known to carry the same topological information as  the  \v{C}ech filtration (via \emph{homotopy equivalence}).
On the other hand, $\linf$-Voronoi regions are generally not convex, which invalidates the standard proof used to show the equivalence of Alpha and \v{C}ech filtrations.
Moreover, $\linf$-Voronoi regions and their dual Delaunay triangulations have been studied primarily from a geometric standpoint and/or in low dimension \cite{shute1991planesweep, boissonnat1998voronoi}.
To overcome some of these limitations, we define two novel families of complexes: Alpha flag complexes and Minibox complexes.
These are both flag complexes defined on a subset of the edges of \v{C}ech complexes.
Our contributions can be summarised as follows:
\begin{itemize}
    \item Under genericity assumptions, we prove that Alpha complexes are equivalent to \v{C}ech complexes for two-dimensional point sets in $\linf$ metric, i.e. filtrations built from these complexes produce the same persistence diagrams. Moreover, we give a counterexample to this equivalence for points sets in three-dimensions.
    \item For arbitrary dimension, we prove the equivalence of Alpha flag and Minibox complexes with \v{C}ech complexes of point sets in $\linf$ metric in homological degrees zero and one.
    \item We study algorithms for finding edges contained in Minibox complexes.
    We recall known results on direct dominance and rectangular visibility for two-dimensional point sets.
    In three dimensions, we describe two novel algorithms for finding Minibox edges taking $O(k\log^2(n))$ and $O(n^2 \log(n))$ time respectively, where $k$ is the number of edges to be found.
    Finally, using orthogonal range queries, we achieve a running time of $O(n^2\log^{d-1}(n))$ for point sets in $\R^d$.
    For $d \geq 4$, this improves over a brute force approach, but does not improve over the algorithms given for the three-dimensional case.
    \item We show that for randomly sampled points in $\R^d$ the expected number of Minibox edges is bounded by $\boundMiniEdges$.
    This is an improvement over the quadratic number of edges contained in \v{C}ech complexes, and results in smaller filtrations. 
    Moreover, Minibox edges can be found independently of higher-dimensional Minibox simplices.
    By comparison, Delaunay triangulations and hence Alpha complexes built over random points using the Euclidean metric are known to be $O(n)$ (linear in the number of vertices )~\cite{dwyer1991higher,schneider2008stochastic},  showing that in this setting, Minibox complexes (using the $\ell_\infty$ metric) are only larger by a polylogarithmic factor.
    %
    \item We provide experimental evidence for speed ups in computation of persistence diagrams by means of Minibox filtrations in homological degrees zero and one.
\end{itemize}
While there is a much smaller body of work on complexes in the $\ell_\infty$ metric, as opposed to the $\ell_2$ metric, there are several relevant related works. 
In particular, approximations of $\linf$-Vietoris-Rips filtrations are studied in \cite{kerber2}.
Moreover, the equivalence  of the different complexes in zero and one homology is related to the results of \cite{kerber1}. 
In this work offset filtrations of convex objects in two and three-dimensional space are considered. 
As in our case, an equivalence of filtrations is proven in homological degrees zero and one by restricting offsets with Voronoi regions.
While this result holds for general convex objects, Minibox filtrations can be used to reduce the size of $\linf$-\v{C}ech filtration in dimensions higher than three.
We also note that our approach is similar in spirit to the preprocessing step via collapses of \cite{boissonnat2020edge-collapses}, but works directly on the geometry of the given finite point set $S$.

\paragraph{Outline} 
We introduce background information in Section \ref{sec:pre}, where we also define Alpha flag and Minibox complexes.
Then, we study Alpha complexes and their properties in the $\linf$ setting in Section \ref{sec:alpha}.
In Section \ref{sec:alpha-flag} and \ref{sec:minibox}, we prove the equivalence of Alpha flag and Minibox complexes with \v{C}ech complexes in homological degrees zero and one.
Next, algorithms for finding Minibox edges, and results on worst-case and expected number of Minibox edges, are given in Section \ref{sec:algorithms}.
Finally, in Section \ref{sec:experiments} we present the results of computational experiments using Alpha flag and Minibox complexes.
Proof details of various technical results, as well as a summary of the notation, can be found in the Appendix.

\section{Preliminaries}
\label{sec:pre}
We first introduce the relevant definitions used in later sections, then define the two new families of complexes studied in this paper.
We point the reader to \cite{hatcher2002algebraic, edelsbrunner2010computational, oudot2015persistence} for a more detailed introduction to homology and persistent homology.

\paragraph{Simplicial complexes} In this work we limit ourselves to simplicial complexes built on a finite set of points in $\mathbb{R}^d$.
We denote the simplicial complex by $K$.
We say that $K$ is a \emph{flag complex} if it is the clique complex of its $1$-skeleton, i.e. it contains a simplex $\sigma$ if and only if it contains all the one-dimensional faces of $\sigma$.
We now introduce several constructions we will use. Let $\tau$ denote a simplex in $K$.
\begin{itemize}
    \item The \emph{nerve} of a finite collection of closed sets $\{A_i\}_{i\in I}$ in $\R^d$ is the abstract simplicial complex $\nrv(\{A_i\}_{i\in I}) = \big\{ \sigma \subseteq I \ \vert \ \bigcap_{i\in \sigma} A_i \neq \emptyset \big\}$.
    \item The \emph{star} of $\tau$ in $K$ is the subset of simplices of $K$ defined by $\st(\tau) = \{\sigma \in K \ \vert \ \tau \leq \sigma \}.$
    \item The \emph{closed star} $\cl(\st(\tau))$ of $\tau$ in $K$ is the smallest subcomplex of $K$ containing $\st(\tau)$.
\end{itemize}

\paragraph{Balls and boxes in \texorpdfstring{$\linf$}{l∞} metric}
Given $p,q \in \R^d$, the \emph{$\ell_{\infty}$ distance}, also known as maximum distance or Chebyshev distance, is defined by
\begin{equation*}
    \label{eq:distance-infinity}
	\din(p, q) = 
	    \max_{1 \leq i \leq d} 
	    \{ \vert p_i - q_i \vert \}.
\end{equation*}
The $\linf$-\emph{diameter} of a finite set of points $\sigma$ is $\diamInfty{\sigma} = \max_{p, q \in \sigma} \din(p, q)$.
Given a point $p\in (\mathbb{R}^d, \din)$ and $r \geq 0$, the \emph{open ball} of radius $r$ and center $p$ is $B_r(p) = \{ x \in \mathbb{R}^d \ \vert \ \din(x,p) < r \}$.
We denote the \emph{closed ball} of radius $r$ and center $p$ by $\overline{B_r(p)}$, its \emph{boundary} by  $\partial \overline{B_r(p)}$.
We have $\varepsilon(\overline{B_r(p)}) = \overline{B_{r+\varepsilon}(p)}$, where  $\varepsilon(A) = \{x \in \R^d \ \vert \ \din(x, A) \leq \varepsilon\}$ is the \emph{$\varepsilon$-thickening} of $A \subseteq \R^d$.
Alternatively, an open ball $B_r(p)$ consists of the points $x$ such that $p_i - r < x_i < p_i + r$ for $1 \leq i \leq d$.
Thus, $B_r(p) = \prod_{i=1}^d I_i^p$, where $I_i^p$ are intervals of the form $(p_i - r, p_i +r)$ for all $i = 1, \ldots, d$, is the interior of an axis-parallel hypercube centered at $p$ with sides of length $2r$.
In general, we call any such Cartesian product of open (closed) intervals, a $d$-dimensional \emph{open (closed) box}.
In case $l$ of the $d$ intervals defining a closed box are degenerate, i.e. their endpoints coincide, we obtain a $(d-l)$-dimensional closed box in $\R^d$.
To conclude, we recall two properties of boxes, which we often refer to in the rest of the paper. 
%
\begin{proposition}
\label{prop:intersection-boxes}
Let $\mathcal{B}$ be a finite collection of either open or closed boxes in $\R^d$.
\begin{enumerate}
    \item[\emph{(i)}] The intersection of the boxes in $\mathcal{B}$ is equal to the Cartesian product of the intersections of intervals defining these boxes, i.e. this intersection is either empty or a box.
    \item[\emph{(ii)}] The intersection of any subset of boxes in $\mathcal{B}$ is non-empty if and only if all the pairwise intersections of these boxes are non-empty.
\end{enumerate}
\end{proposition}
\begin{proof}
Both \emph{(i)} and {\emph{(ii)}} follow from the facts that Cartesian products and intersections commute, and that the intersection of a finite number of intervals is either empty or an interval.
\end{proof}
\paragraph{Voronoi diagrams and Delaunay triangulations} These constructions have been extensively studied in computational geometry \cite{de2010computational}, primarily for Euclidean space. We refer the reader to \cite{aurenhammer2013voronoi} for a reference on general Voronoi diagrams and Delaunay triangulations.
\begin{definition}
\label{def:voronoi}
Let $S$ be a finite set of points in $(\mathbb{R}^d, \din)$.
The $\ell_{\infty}$-\emph{Voronoi region} of a point $p\in S$ is
$$ V_p = \left\{ x \in \mathbb{R}^d \ \vert \ 
                 \din(p,x) \leq \din(q,x), 
                 \ \forall q \in S \right\}.$$
The \emph{bisector} of a subset $\sigma \subseteq S$ is 
$\bisector{\sigma}
= \left\{x \in \R^d \ \vert \
        \din(p, x) = \din(q, x) 
        \textrm{ for } p,q \in \sigma \right\} =$ 
$\bigcap_{p \in \sigma} V_p$.
The set of $\ell_{\infty}$-Voronoi regions $\{V_p\}_{p\in S}$ is the $\ell_{\infty}$-\emph{Voronoi diagram} of $S$.
\end{definition}
\begin{definition}
\label{def:delaunay-complex}
The $\linf$-\emph{Delaunay complex} of a finite set of points $S$ in $(\R^d, \din)$ is the simplicial complex 
$$K^D = \left\{ \sigma \subseteq S \ \vert \ 
              \bisector{\sigma} 
              \neq
              \emptyset \right\}.$$
\end{definition}
If $d+2$ points lie on the boundary of a closed ball, see Figure \ref{fig:deg1}, then $K^D$ contains a $(d+1)$-dimensional simplex even if the points set $S$ is embedded in dimension $d$.
Moreover, given two points $p, q\in S$ lying on an axis-parallel hyperplane, their bisector may be degenerate. For instance, this is the case for the points in Figure \ref{fig:deg2}, where $V_p \cap V_q$ is the union of a line segment and two cones.
We define a concept of general position to avoid such cases for points in $\R^2$.
This is necessary for proving the equivalence of Alpha and \v{C}ech complexes of $S \subseteq \R^2$, and makes the geometric realization of $K^D$ well-defined.
%
\begin{figure}[tb]
    \centering
        \begin{subfigure}[b]{2in}
            \centering
            \includegraphics[width=2in]{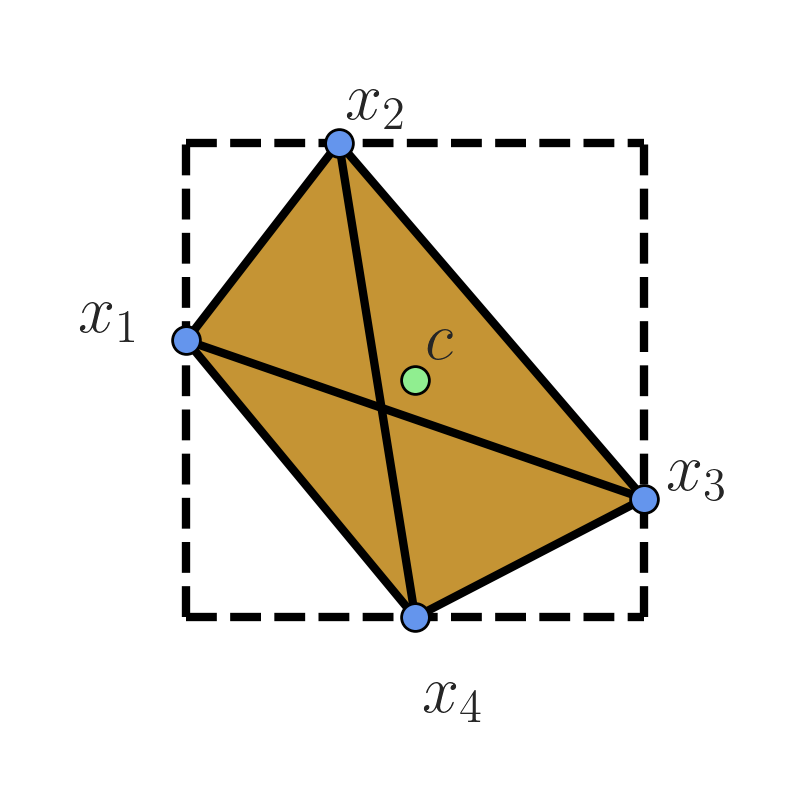}
            \caption{}
            \label{fig:deg1}
        \end{subfigure}
        \begin{subfigure}[b]{2in}
            \centering
            \includegraphics[width=2in]{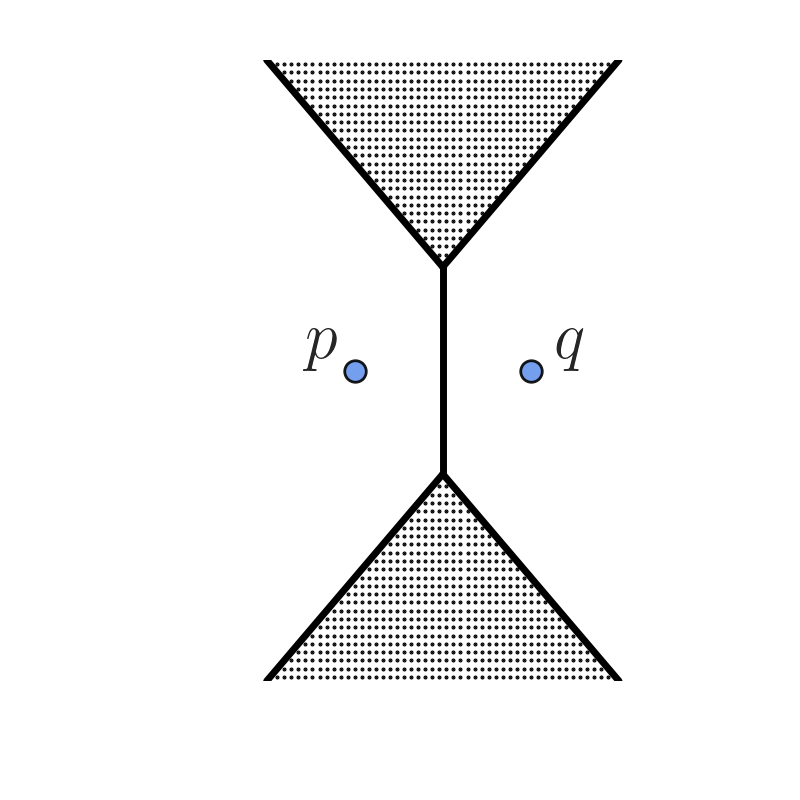}
            \caption{}
            \label{fig:deg2}
        \end{subfigure}
        \begin{subfigure}[b]{2in}
            \centering
            \includegraphics[width=2in]{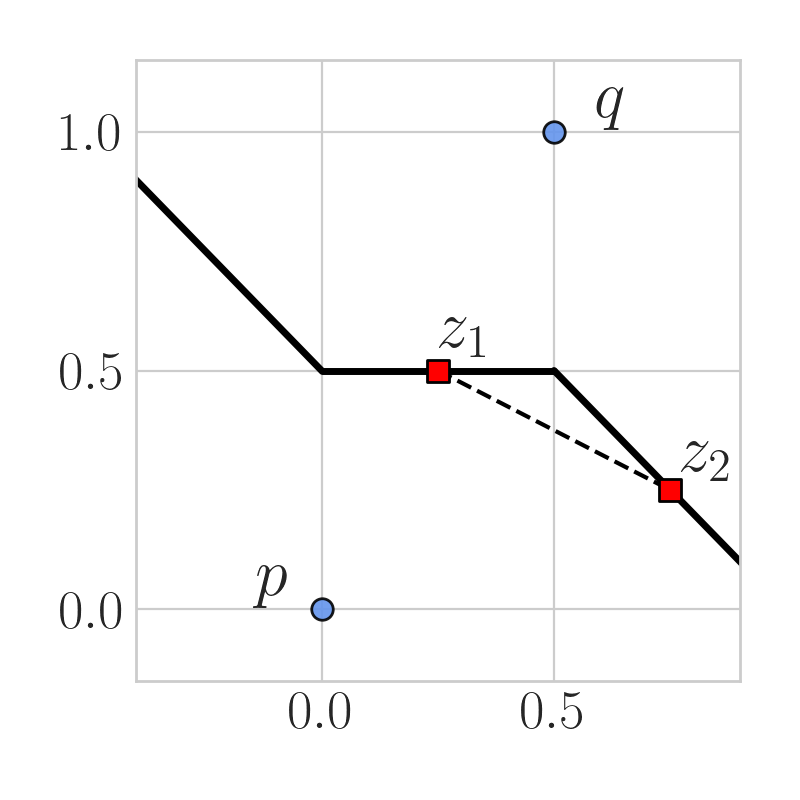}
            \caption{}
            \label{fig:non-conv}
        \end{subfigure}
    \caption{\textbf{(a)} Four points in $\mathbb{R}^2$ whose $\linf$-Delaunay complex is three-dimensional. \textbf{(b)} Degenerate intersection of $\linf$-Voronoi regions. \textbf{(c)} $\linf$-Voronoi regions are not convex.}
    \label{fig:voronoi-delaunay-and-degenerate-case}
\end{figure}
%
\begin{definition}
\label{def:general-position}
Let $S$ be a finite set of points in $(\R^2, \din)$.
We say that $S$ is in \emph{general position} if no four points of $S$ to lie on the boundary of a square, and no two points share a coordinate.
\end{definition}
Corollary $3.18$ of \cite{criado2019tropical} ensures that if $S \subseteq (\R^2, \din)$ is in general position, then $\bisector{\{p_1, p_2, p_3\}} = V_{p_1} \cap V_{p_2} \cap V_{p_3}$ is either empty or a point for any distinct $p_1, p_2, p_3 \in S$.\footnote{Our general position corresponds to the weak general position of \cite{criado2019tropical}, where strong general position is also defined and implies weak general position.
We use the former of the two concepts, because it is sufficient to obtain the property of intersections of $\linf$-Voronoi regions of points in $\R^2$ used in this paper.}
\begin{definition}
\label{def:delaunay-triangulation}
The $\linf$-\emph{Delaunay triangulation} of a finite set of points $S$ in general position in $(\mathbb{R}^2, \din)$ is the geometric realisation of the $\linf$-Delaunay complex $K^D$ of $S$, which is the set of convex hulls of simplices of $K^D$.
\end{definition}
Finally, we note that $\linf$-Voronoi regions are not generally convex.
To see this, consider $p=\left(0, 0\right)$ and $q=\left(\frac{1}{2}, 1\right)$ as in Figure \ref{fig:non-conv}.
These are such that $z_1=\left(\frac{1}{4}, \frac{1}{2}\right), z_2=\left(\frac{3}{4}, \frac{1}{4}\right) \in V_p, V_q$, but the middle point on the line segment from $z_1$ to $z_2$ is $\frac{z_1 + z_2}{2} = \left(\frac{1}{2}, \frac{3}{8}\right)$ which belongs to $V_p$ only.

\paragraph{$\linf$-Delaunay edges} 
By definition of $\linf$-Delaunay complex $K^D$, pairs of points $\{p, q\} \subseteq S$ are an edge of $K^D$ if and only if $V_p \cap V_q$ is non-empty.
We can further characterize $\linf$-Delaunay edges making use of the concept of witness points.
The proof of this characterization is given in Appendix \ref{sec-app:delaunay-edges}.
\begin{definition}
\label{def:witness-points}
Let $S$ be a finite set of points in $(\mathbb{R}^d, \din)$.
A \emph{witness} point of $\sigma \subseteq S$ is a point $z$ such that $z \in \bisector{\sigma} = \bigcap_{p \in \sigma} V_p$ and $\din(z,p) = \frac{\diamInfty{\sigma}}{2}$ for each $p \in \sigma$.
We write $\Z{\sigma}$ for the \emph{set of witness points} of $\sigma$.
\end{definition}
\begin{proposition}
\label{prop:delaunay-edge}
Let $S$ be a finite set of points in $(\R^d, \din)$.
Given a subset $e = \{p, q\} \subseteq S$, we define $\A^{\bar{r}} 
= 
\partial \cB{\bar{r}}{p} \cap \partial \cB{\bar{r}}{q}$, where $\bar{r} = \frac{\din(p,q)}{2}$.
We have that $\A^{\bar{r}} = \cB{\bar{r}}{p} \cap \cB{\bar{r}}{q}$ is a non-empty degenerate closed box.
Moreover, the set of witness points of $e$ is
$\Z{e} = \A^{\bar{r}} \setminus \big(\bigcup_{y\in S\setminus e} B_{\bar{r}}(y)\big)$, and $e$ belongs to the $\linf$-Delaunay complex of $S$ if and only if $\Z{e}$ is non-empty.
\end{proposition}

\paragraph{Persistent homology}
A \emph{filtration} of a simplicial complex $K$ parameterized by $\mathcal{R}$ is a nested sequence of subcomplexes 
$K_{\mathcal{R}} = \{ K_{r_1} \subseteq K_{r_2} \subseteq \ldots \subseteq K_{r_m} \}$, where $\mathcal{R} = \{r_i\}_{i=1}^m$ a finite set of monotonically increasing real values. 
We list three types of complexes used to define filtrations on a finite set of points $S \subseteq (\R^d, \din)$.
\begin{itemize}
    \item The \emph{Vietoris-Rips complex} with radius $r$ of $S$ is $\ripsr = \big\{\sigma \subseteq S \ \vert \ \diamInfty{\sigma} \leq 2r \big\}$.
    \item The \emph{\v{C}ech complex} with radius $r$ of $S$ is $\cechr = \big\{\sigma \subseteq S \ \vert \ \bigcap_{p \in \sigma} \overline{B_r(p)} \neq \emptyset \big\}$.
    \item The \emph{Alpha complex} with radius $r$ of $S$ is $\alphar = \big\{\sigma \subseteq S \ \vert \ \bigcap_{p \in \sigma} \big(\overline{B_r(p)} \cap V_p\big) \neq \emptyset \big\}$.
\end{itemize}
For each of the complexes above, we have $K_{r_1}^{\bullet} \subseteq K_{r_2}^{\bullet}$ if $r_1 < r_2$. 
So, given a monotonically increasing set of real values $\mathcal{R}$, we have the filtration $K_{\mathcal{R}}^{\bullet} = \{ K_{r_1}^{\bullet} \subseteq K_{r_2}^{\bullet} \subseteq \ldots \subseteq K_{r_m}^{\bullet} \}$.
\begin{proposition}
\label{prop:equality-cech-and-rips}
Let $S$ be a finite set of points in $(\R^d, \din)$. The \v{C}ech and Vietoris-Rips complexes of $S$ coincide, i.e. $\cechr = \ripsr$ for any $r \in \R$. 
\end{proposition}
\begin{proof}
Follows from the definitions of \v{C}ech and Vietoris-Rips complexes and Proposition \ref{prop:intersection-boxes}(ii).
\end{proof}
\begin{corollary}
\label{cor:cech-is-flag}
The \v{C}ech complexes of $S \subseteq (\R^d, \din)$ are flag complexes.
The smallest radius such that $\sigma \in \cechr$ is $\bar{r} = \frac{\diamInfty{\sigma}}{2}$ for each $\sigma \subseteq S$.
\end{corollary}

Given a filtration $K_{\mathcal{R}}$, we obtain the $k$-th \emph{persistence module} $\modulek = \{ H_k(K_{r_1}; \mathbb{F}) \rightarrow H_k(K_{r_2}; \mathbb{F}) \rightarrow \cdots \rightarrow H_k(K_{r_m}; \mathbb{F}) \}$ by applying the $k$-th homology functor $H_k(-;\mathbb{F})$, with coefficients in a field $\mathbb{F}$, to its elements.
This admits a unique decomposition, as shown in \cite{zomorodian2005computing}, which is in bijection with a set of intervals of the form $[r_i, r_j)$ and $[r_i, +\infty)$.
Mapping these intervals into the points $(r_i, r_j)$ and $(r_i, +\infty)$, we obtain the $k$-th \emph{persistence diagram} $\dgmk(K_{\mathcal{R}})$ of the filtration $K_{\mathcal{R}}$.
This is a multi-set of points in the extended plane $\overline{\R}^2$, where $\overline{\mathbb{R}} = \mathbb{R}\cup \{ + \infty \}$.
%
%
Importantly, the Stability Theorem of \cite{cohen2007stability} implies that the persistence diagrams of filtrations of \v{C}ech complexes of $S \subseteq (\R^d, \din)$ are infinitesimally perturbed if $S$ is infinitesimally perturbed.

In practice, the persistent homology algorithm, first described in \cite{edelsbrunner2002topo-simplification}, takes a filtration, and outputs its persistence diagrams up to a fixed homological degree.
A substantial amount of work has been done on the computational complexity of computing persistent homology, with a large number of results \cite{ chen2011persistent, de2011dualities, wagner2012efficient, mischaikow2013morse, bauer2014clear, bauer2014distributed} which have greatly sped up computations in practice \cite{otter2017roadmap}. 
The standard algorithm has a complexity of $O(m^3)$, which can be reduced to $O(m^\omega)$ where $m$ is the number of simplices in the input filtrations and $\omega$ is the matrix multiplication exponent \cite{milosavljevic2011zigzag}.
However, it has been observed that the majority of computation time is spent constructing the filtration. Thus, smaller complexes generally result in faster computation. 
For instance, in the case of \v{C}ech filtrations, we have to consider $\Uptheta(n^{k+2})$ simplices in order to compute their $k$-th persistence diagram.
In Euclidean metric, \v{C}ech persistent homology can be computed using Alpha filtrations, which greatly reduces the number of simplices to be considered.

\paragraph{New complexes}
We propose the use of the following families of complexes as an alternative to Alpha complexes for the computation of \v{C}ech persistence diagrams of a finite set of points $S \subseteq (\R^d, \din)$.
\begin{itemize}
    \item The \emph{Alpha flag complex} with radius $r$ of $S$ is 
    $$\flagr 
    = 
    \big\{ 
      \sigma \subseteq S 
      \ \vert \ 
      \diamInfty{\sigma} \leq 2r 
      \textrm{ and } 
      \{p,q\} \in K^D 
      \textrm{ for each } 
      \ p, q \in \sigma   
    \big\},$$
    where $K^D$ is the $\linf$-Delaunay complex of $S$.
    \item The \emph{Minibox complex} with radius $r$ of $S$ is 
    $$
    \minir 
    = 
    \left\{ 
    \sigma \subseteq S 
    \ \vert \ 
    \diamInfty{\sigma} \leq 2r
    \textrm{ and }
    \minipq \cap S \neq \emptyset
    \textrm{ for each }
    p, q \in \sigma
    \right\}, 
    $$
    where $\minipq = \prod_{i=1}^d (\min \{p_i, q_i\}, \max\{p_i, q_i\})$ is the interior of the minimal bounding box of $p$ and $q$.
\end{itemize}
\begin{remark}
In Section \ref{sec:alpha}, we show that $\alphar$ is not in general a flag complex, implying $\alphar \neq \flagr$. 
\end{remark}

It should be noted that both Alpha flag and Minibox complexes are flag complexes. 
Thus, we only need to determine their edges in order to build them.
In particular, we need to find the $\linf$-Delaunay edges of $S$ for $\flagr$.
On the other hand, we have to find the all the pairs of points $\{p, q\} \subseteq S$ satisfying $\minipq \cap S \neq \emptyset$ for $\minir$.

In the remainder of the paper, we prove the equivalence of the above complexes with \v{C}ech complexes in homological degrees zero and one.
Moreover, we provide efficient algorithms for the computation of the edges contained in Minibox complexes.
Finally, in Section \ref{sec:experiments}, we study how the reduced size Minibox filtrations helps with persistence computations in practice.

\section{Alpha Complexes}
\label{sec:alpha}
Given a finite set of points $S$ in a Euclidean space, it is known that Alpha and \v{C}ech complexes are equivalent, i.e. produce the same persistence diagrams.
This is a consequence of the Nerve Theorem, see \cite[Section 3.4]{edelsbrunner2010computational}.
In this section, we study the properties of Alpha complexes of points in a $\linf$ metric space.
In particular, we show their equivalence with \v{C}ech complexes for two-dimensional points.
On the other hand, we give a counterexample to this equivalence for points in thee-dimensions.

\paragraph{Alpha complexes in $\R^2$}
\begin{figure}[tb]
    \centering
        \begin{subfigure}[b]{2in}
            \centering
            \includegraphics[width=2in]{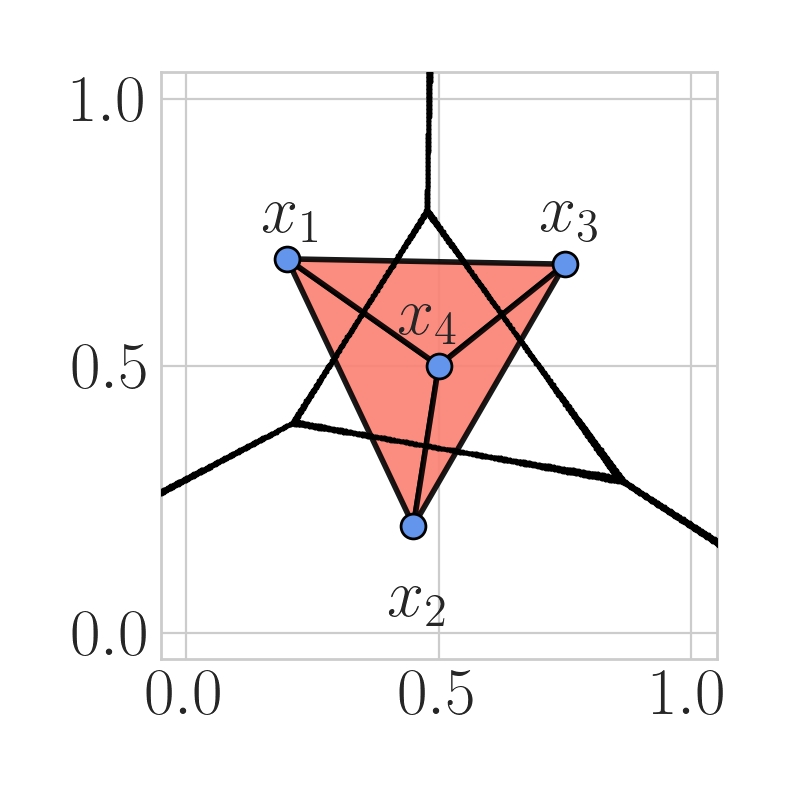}
            \caption{}
            \label{fig:voronoi-euclidean}
        \end{subfigure}
        \qquad
        \begin{subfigure}[b]{2in}
            \centering
            \includegraphics[width=2in]{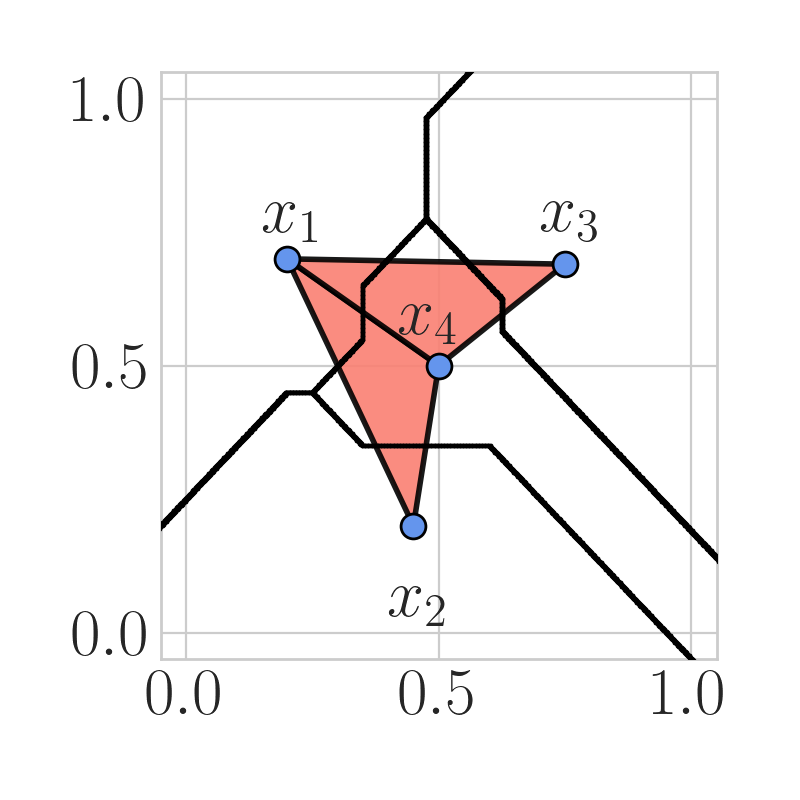}
            \caption{}
            \label{fig:voronoi-inf}
        \end{subfigure}
    \caption{Voronoi diagrams and Delaunay triangulations of four points in $\mathbb{R}^2$, with Euclidean and $\linf$ metric in \textbf{(a)} and \textbf{(b)} respectively. 
    Note that the triangle $\{x_1, x_2, x_3\}$ is missing from the Euclidean Delaunay triangulation in \textbf{(a)}.}
    \label{fig:example-voronoi}
\end{figure}
Assuming general position (Definition \ref{def:general-position}), we are able to show the following equivalence of complexes.
Here we provide a proof sketch highlighting the main ideas of the proof. The full details of the proof can be found in Appendix \ref{sec-app:alpha-cech-R2}.
\begin{theorem}
\label{thm:equiv-2D}
Let $S$ be a finite set of points in $(\R^2, \din)$ in general position.
The Alpha and \v{C}ech filtrations of $S$ are equivalent, i.e. produce the same persistence diagrams.
\end{theorem}
\begin{proofsketch}
We apply the Nerve Theorem of \cite{handbook-comb} to the Alpha complex $\alphar$ for any $r \in \R$, which is the nerve of $\{\cB{r}{r} \cap V_p\}_{p \in S}$.
Given $e = \{p, q\}$ and $\bar{r} = \frac{\din(p,q)}{2}$ as in Proposition \ref{prop:delaunay-edge}, this is possible because:
\begin{itemize}
    \item Each $\cB{r}{p} \cap V_p$ is star-like, and so contractible;
    \item The intersections $\big(\cB{r}{p} \cap V_p\big) \cap \big(\cB{r}{q} \cap V_q\big)$ are either a line segment contained in $\A^{\bar{r}}= \cB{\bar{r}}{p} \cap \cB{\bar{r}}{q}$ or retract on such a line segment;
    \item Intersections of $3$ or more $\cB{r}{p} \cap V_p$ are either empty or consist of a single point by the general position assumption. 
\end{itemize}
Moreover, the Nerve Theorem applies to the collection $\{\cB{r}{p}\}_{p \in S}$ by convexity, and $\bigcup_{p\in S} \big(\cB{r}{p} \cap V_p\big) = \bigcup_{p \in S} \cB{r}{p}$.
So $\alphar \simeq \cechr$ for any $r \in \R$, and the result follows by applying the Persistence Equivalence Theorem of \cite[Section 7.2]{edelsbrunner2010computational}.
\end{proofsketch}
This is similar to the results of \cite{kerber1}, which proves that the nerve of offsets of convex shapes is equivalent, in homological degrees zero and one, to the union of the shapes in two and three dimensions.
Furthermore, the above theorem implies that the degree two homology of Alpha complexes of $S \subseteq (\R^2, \din)$ is trivial, because it equals the one of the two-dimensional sets $\bigcup_{p \in S} \cB{r}{p}$.

We conclude our discussion of the properties of Alpha complexes of point sets in $(\R^2, \din)$ with the following result, proven in Appendix \ref{sec-app:alpha-cech-R2}, which shows that they are flag complexes on $\linf$-Delaunay edges.
The same does not hold for Euclidean Alpha complexes, which are built on Euclidean Delaunay triangulations, which are not flag in general. 
See Figure \ref{fig:example-voronoi} for an example.
%
\begin{proposition}
\label{prop:delaunay-alpha-flag-in-2D}
Let $S$ be a finite set of points in general position in $(\mathbb{R}^2, \din)$ and $r \geq 0$.
Both the $\linf$-Delaunay complex $K^D$ and the Alpha complex $\alphar$ of $S$ are flag complexes.
Moreover, $e = \{p, q\} \in K^D$ belongs to $\flagr$ if and only if $\frac{\din(p, q)}{2} \leq r$.
\end{proposition}
\begin{remark}
Note that the $\linf$-Delaunay edges of $S \subseteq (\R^d, \din)$ can be found with the $O(n \log(n))$ plane-sweep algorithm of \cite{shute1991planesweep}, and used to build Alpha filtrations.
\end{remark}
\begin{figure}[tb]
    \centering
    \begin{subfigure}[b]{2in}
        \centering
        \includegraphics[width=2in]{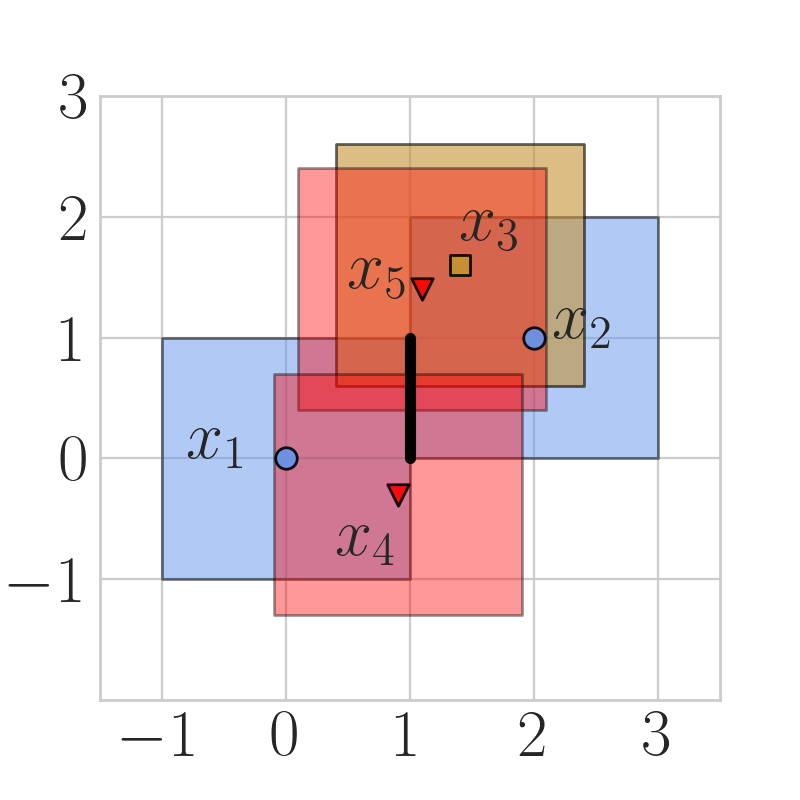}
        \caption{Projection along $x$ and $y$ axes.}
        \label{fig:proj1}
    \end{subfigure}
    \qquad
    \begin{subfigure}[b]{2in}
        \centering
        \includegraphics[width=2in]{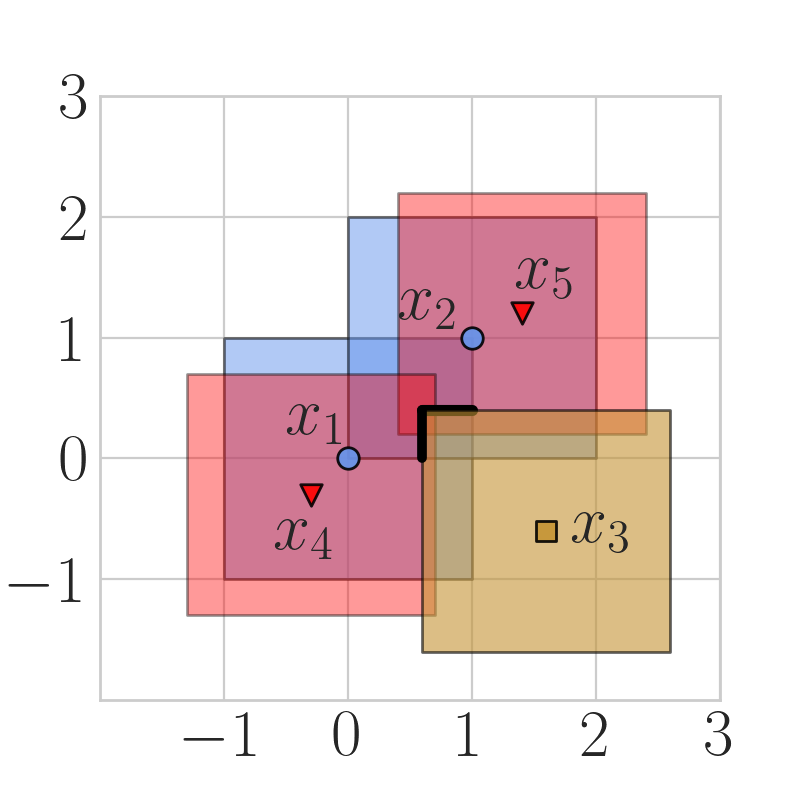}
        \caption{Projection along $y$ and $z$ axes.}
        \label{fig:proj2}
    \end{subfigure}
    \caption{Five points in $\mathbb{R}^3$ such that their $\linf$-Delaunay complex $K^D$ is not flag.}
    \label{fig:counterexample-3D}
\end{figure}
%
%
\paragraph{Counterexample: $\alphar$ not flag in $(\R^3, \din)$}
It can be shown that in general the $\linf$-Delaunay complexes of three-dimensional points do not contain all the cliques on their edges.
Thus, $\alphar \subseteq \flagr \subseteq \cechr$ for any $r \in \R$, where $\flagr$ is the Alpha flag complex defined in Section \ref{sec:pre}.

Let $S = \{x_i\}_{i=1}^5$ be the set of five points in $(\R^3, \din)$ such that $x_1=(0,0,0)$, $x_2=(2,1,1)$, $x_3=(1.4, 1.6, -0.6)$, $x_4=(0.9, -0.3, -0.3)$, and $x_5=(1.1, 1.4, 1.2)$.
One can check that the $\linf$-Delaunay complex $K^D$ of $S$ is not a flag complex.

First, we have that $(1,0,1)$, $(0.8,0.8,0.0)$, and $(1.5, 1.5, 0.2)$ are witness points of the edges $\{x_1, x_2\}$, $\{x_1, x_3\}$, and $\{x_2, x_3\}$ respectively.
Thus, $K^D$ contains these edges by Proposition \ref{prop:delaunay-edge}.

On the other hand $\tau = \{x_1, x_2, x_3\}$ is not a triangle in $K^D$.
This follows from the fact that $A_{\tau}^{1} = \partial \cB{1}{x_1} \cap \partial \cB{1}{x_2} \cap \partial \cB{1}{x_3}$ is formed by the two line segments, plotted as thickened lines in Figure \ref{fig:counterexample-3D}, with endpoints $(1, 0.6, 0)$, $(1, 0.6, 0.4)$ and $(1, 0.6, 0.4)$, $(1, 1, 0.4)$, which are covered by $B_1(x_4) \cup B_1(x_5)$.
The $\varepsilon$-thickenings of these line segments contain $A_{\tau}^{1+\varepsilon}$ for any $\varepsilon \geq 0$, by the properties of $\varepsilon$-thickenings described in Appendix \ref{sec-app:delaunay-edges}.
In turn, the $\varepsilon$-thickenings of the two line segments are contained in $\varepsilon(B_1(x_4) \cup B_1(x_4)) = B_{1+\varepsilon}(x_4) \cup B_{1+\varepsilon}(x_5)$.
This implies that
there does not exist a point $z \in V_{x_1} \cap V_{x_2} \cap V_{x_3}$, as this would require $A_{\tau}^{1+\varepsilon} \setminus \big( B_{1+\varepsilon}(x_4) \cup B_{1+\varepsilon}(x_5) \big)$ to be non-empty for some $\varepsilon \geq 0$.

%
\paragraph{Counterexample: $\alphar$ and $\cechr$ not equivalent in $(\R^3, \din)$}
We conclude this section by providing a counterexample to the equivalence of Alpha and \v{C}ech complexes of three-dimensional points.
In particular, we give a configuration of eight points $S = \{ x_i \}_{i=1}^8 \subseteq \R^3$ such that their $\linf$-Delaunay complex contains the four faces of the tetrahedron $\{x_1, x_2, x_3, x_4\}$, but not the tetrahedron itself.
This way the Alpha complexes of $S$ never contain $\{x_1, x_2, x_3, x_4\}$ as a simplex, but they contain its the four faces for a big enough radius parameter.
Moreover, the $\linf$-Delaunay complex of $S$ also does not contain other tetrahedra that could possibly fill in the two-dimensional void created by the faces of $\{x_1, x_2, x_3, x_4\}$.
We list the coordinates of the points giving a counterexample in Table \ref{table:counterexample}, and plot them in Figure \ref{fig:alpha-counterexamle} by projecting along two of the three coordinate axes.
These were found by randomly sampling many sets of eight points in $\R^3$, and testing whether their Alpha and \v{C}ech persistence diagrams were equal.
The existence of such a counterexample can be thought of as a consequence of the non-convexity of $\linf$-Voronoi regions, even if one may have hoped the nerve of general Voronoi regions to be well behaved enough to prevent this from happening.

\begin{figure}[tb]
    \centering
    \begin{subfigure}[b]{2in}
        \centering
        \includegraphics[width=2in]{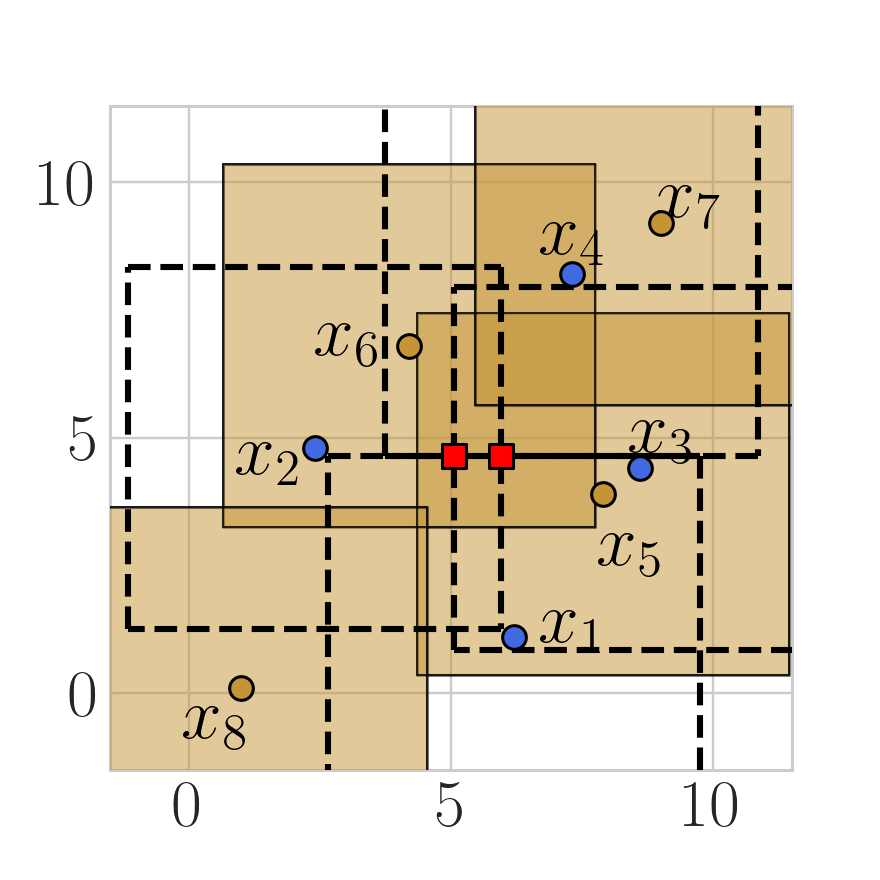}
        \caption{Projection along $x$ and $y$ axes.}
        \label{fig:alpha-count-1}
    \end{subfigure}
    \quad
    \begin{subfigure}[b]{2in}
        \centering
        \includegraphics[width=2in]{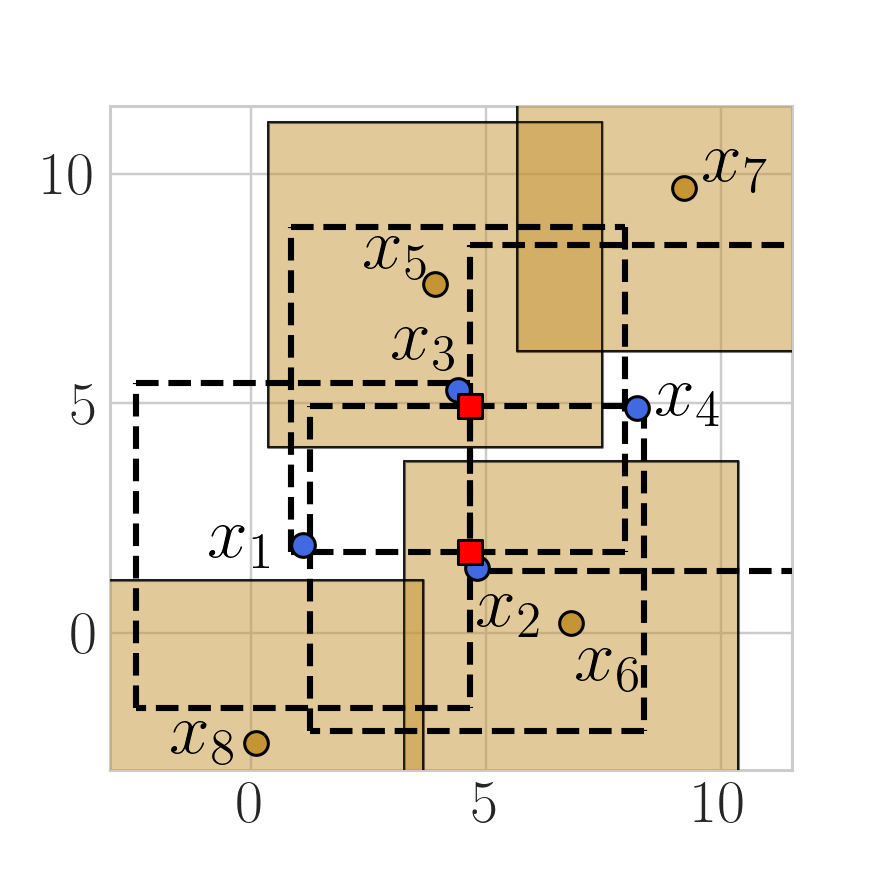}
        \caption{Projection along $y$ and $z$ axes.}
        \label{fig:alpha-count-2}
    \end{subfigure}
    \caption{Counterexample to the equivalence of Alpha and \v{C}ech persistent homology in $\ell_{\infty}$ metric. The two circumcenters of the tetrahedron $\{x_1, x_2, x_3, x_4\}$ are the red square markers. The boundaries of cubes centered in the vertices of $\{x_1, x_2, x_3, x_4\}$ are shown as dashed lines.}
    \label{fig:alpha-counterexamle}
\end{figure}
%
\begin{table}[tb]
 \caption{Coordinates of points $S \subseteq (\R^3, \din)$ giving a counterexample to the equivalence of Alpha and \v{C}ech filtrations.}
 \label{table:counterexample}
  \centering
  \begin{tabular}{l | c | c | c |}
        \cline{2-4}
        & x  &  y  & z \\
    \hline
    \multicolumn{1}{|l|}{$x_1$} 
        &  6.2 & 1.1 &  1.9  \\
    \hline
    \multicolumn{1}{|l|}{$x_2$}  
        &  2.4 & 4.8 & 1.4  \\
    \hline
    \multicolumn{1}{|l|}{$x_3$} 
        &  8.6 & 4.4 & 5.3  \\
    \hline
    \multicolumn{1}{|l|}{$x_4$}  
      &  7.3 & 8.2 & 4.9  \\
    \hline
    \multicolumn{1}{|l|}{$x_5$}  
        &  7.9 & 3.9 & 7.6  \\
    \hline
    \multicolumn{1}{|l|}{$x_6$}  
        &  4.2 & 6.8 & 0.2  \\
    \hline
    \multicolumn{1}{|l|}{$x_7$}  
        &  9.0 &  9.2 & 9.7  \\
    \hline
    \multicolumn{1}{|l|}{$x_8$}  
        & 1.0  & 0.1 & -2.4 \\
    \hline
    \end{tabular}
\end{table}
One can check that there are six tetrahedra belonging to the $\linf$-Delaunay complex $K^D$ of $S$: $\{x_1, x_2, x_3, x_5\}$, $\{x_1, x_2, x_3, x_6\}$,
$\{x_1, x_2, x_4, x_5\}$, $\{x_1, x_3, x_4, x_6\}$,
$\{x_2, x_3, x_4, x_5\}$, and $\{x_2, x_3, x_4, x_6\}$.
This can be done by finding the circumcenters of any four given points, and checking that the circumspheres of these (which in this case are cubes) do not contain any of the other points.
It is important to note that in $\linf$ metric four three-dimensional points might have two distinct circumcenters. 
For instance this is the case for $\{x_1, x_2, x_3, x_4\}$, the circumcenters of which are represented as red square markers in Figures \ref{fig:alpha-count-1} and \ref{fig:alpha-count-2}, having coordinates $w_1 = (5.95, 4.65, 1.75)$ and $w_2 = (5.05, 4.65, 4.95)$\footnote{Note that $w_1$ and $w_2$ are two distinct points, of which the bisector of $\{x_1, x_2, x_3, x_4\}$ is a subset. Thus, it is possible to have a non-contractible bisector, which invalidates the equivalence proof strategy making use of the Nerve Theorem.}.
On the other hand, in Euclidean metric four affinely independent three-dimensional points have exactly one circumcenter.
Moreover, $w_1$ and $w_2$ are not witnesses of $\{x_1, x_2, x_3, x_4\}$, because they are closer to $x_5$ and $x_6$ than to the vertices of this tetrahedron.
Thus, $\{x_1, x_2, x_3, x_4 \} \not\in K^D$.
Regarding the faces of $\{x_1, x_2, x_3, x_4\}$, we have that:
\begin{itemize}
    \item $(5.5, 4.2, 3.9)$ is a witness of $\{x_1, x_2, x_3\}$ at distance $3.1$ from $x_1$, $x_2$, and $x_3$.
    \item $(4.05, 4.65, 4.95)$ is a witness of $\{x_1, x_2, x_4\}$ at distance $3.55$ from $x_1$, $x_2$, and $x_4$.
    \item $(8.75, 4.65, 1.75)$ is a witness of $\{x_1, x_3, x_4\}$ at distance $3.55$ from $x_1$, $x_3$, and $x_4$.
    \item $(5.5, 5.1, 3.9)$ is a witness point of $\{x_2, x_3, x_4\}$ at distance $3.1$ from $x_2$, $x_3$, and $x_4$.
\end{itemize}
The tetrahedra belonging to the $\linf$-Delaunay complex of $S$ (listed in the above discussion) do not create a boundary for the degree-two homology class created by adding $\{x_1, x_2, x_3\}$, $\{x_1, x_2, x_4\}$, $\{x_1, x_3, x_4\}$, and $\{x_2, x_3, x_4\}$ into $\alphar$, for $r > 0$ big enough.
Thus, the persistence diagram in homological degree two of the Alpha filtration of $S$ has a point at infinity, i.e. an homology class that never dies.
On the other hand, the persistence diagrams in homological degree two of the \v{C}ech filtration of $S$ cannot have such a point, because \v{C}ech complexes have trivial homology for a big enough radius.

\section{Equivalence of Alpha Flag Complexes}
\label{sec:alpha-flag}
We prove that Alpha flag and \v{C}ech complexes of a finite set of points $S \subseteq (\R^d, \din)$ produce the same persistence diagrams in homological degrees zero and one.
Thus, the complexes $\flagr$ can be used to limit the size of filtrations used in the computation of \v{C}ech persistence diagrams.
However, Alpha flag complexes, similarly to Alpha complexes, cannot be used for the computation of persistence in homological degree two or higher.
One advantage over Alpha complexes is that it is only necessary to find the $\linf$-Delaunay edges on $S$, rather than the full $\linf$-Delaunay complex.

In the proof of the theorem presented in this section, we make use of the following two results.
We again provide proof sketches for readability, with full proofs in Appendix \ref{sec-app:proof-main}. 
We omit referencing the field $\mathbb{F}$ when referring to the homology of complexes.
Finally, a \emph{non-Delaunay} edge is a pair of points $\{p, q\} \subseteq S$ if it does not belong to the $\linf$-Delaunay complex of $S$.
Importantly, we do not make use of any general position assumption on $S$.

Our proof is based on the idea of decomposing \v{C}ech complexes into filtrations adding a single edge, and the cliques it forms, at each filtration step.
To deal with the problem of multiple edges possibly having the same length $2\bar{r}$ we introduce the following definitions.
%
\begin{definition}
\label{def:edge-by-edge-filtration}
Let $S$ be a finite set of points in $(\R^d, \din)$.
A \emph{single edge-length range} of \v{C}ech complexes of $S$ is an open interval $(r, r + \varepsilon) \subseteq \R$ such that all the edges not in $\cechr$ and contained in $\cechre$ have the same length $2\bar{r}$.
Given a single edge-length range $(r, r+\varepsilon)$, the  \emph{\v{C}ech edge-by-edge filtration} of $S$ on this range is 
$$\cechr 
=
\KC{0}
\subseteq 
\KC{1}
\subseteq 
\ldots 
\subseteq
\KC{n_i}
=
\cechre,$$
where $\KC{i}$ contains exactly one edge not in $\KC{i-1}$, together with the cliques containing this edge, for each $1 \leq i \leq n_i$.
The corresponding \emph{Alpha flag edge-by-edge filtration} of $S$ on the same range is
$$\flagr
=
\KAF{0}
\subseteq 
\KAF{1}
\subseteq 
\ldots 
\subseteq
\KAF{n_i}
=
\flagre,$$
where $\KAF{i} = \KC{i} \cap \flagre$ for each $1 \leq i \leq n_i$.
\end{definition}
%
\begin{lemma}
\label{lemma:proof-main-1-add-one-edge}
Let $(r, r+\varepsilon)$ be a single edge-length range of \v{C}ech complexes of $S \subseteq (\R^d, \din)$, and $\{\KAF{i}\}_{i=0}^{n_i}$, $\{\KC{i}\}_{i=0}^{n_i}$ the Alpha flag and \v{C}ech edge-by-edge filtrations on this range. 
If going from $\KAF{i-1}$ to $\KAF{i}$ a $\linf$-Delaunay edge is the only simplex added in $\KAF{i}$, then this is also the only simplex added going from $\KC{i-1}$ to $\KC{i}$.
\end{lemma}
\begin{proofsketch}
Let $e=\{p, q\}$ be the $\linf$-Delaunay edge added in $\KAF{i}$. 
We define $\bar{r} = \frac{\din(p, q)}{2}$, so that $r < \bar{r} < r + \varepsilon$, and $\openY = \{ y \in S \ \vert \  \din(y, p) < 2 \bar{r} \textrm{ and  } \din(y, q) < 2 \bar{r} \}$.

It is possible to show that $\openY$ has to be empty, otherwise at least a triangle would be added in $\KAF{i}$ together with $e$.
Finally, a contradiction is obtained with $e$ adding a higher-dimensional simplex in $\KC{i}$, because this would require $\openY$ to be non-empty or the same simplex to be added in $\KAF{i}$.
\end{proofsketch}
%
\begin{lemma}
\label{lemma:removing-non-delaunay-edge}
Let $(r, r+\varepsilon)$ be a single edge-length range of \v{C}ech complexes of $S \subseteq (\R^d, \din)$, and $\{\KC{i}\}_{i=0}^{n_i}$ the \v{C}ech edge-by-edge filtration on this range. 
If the edge $e = \{p, q\}$ added going from $\KC{i-1}$ to $\KC{i}$ is non-Delaunay for $1 \leq i \leq n_i$, then $H_k(\KC{i} \setminus \st(e)) = H_k(\KC{i-1})$ and $H_k(\KC{i})$ are isomorphic for $k=0,1$.
\end{lemma}
\begin{proofsketch}
Let $A = \cl(\st(e)) \subseteq \KC{i}$ and $B=\KC{i} \setminus \st(e)$, so that $A \cap B = \cl(\st(e)) \setminus \st(e)$.
Note that $\KC{i} \setminus \st(e) = \KC{i-1}$ by definition of \v{C}ech edge-by-edge filtration.
Applying the reduced Mayer-Vietoris sequence with these $A$ and $B$, we obtain
\begin{align*}
	\cdots \rightarrow 
	\tilde{H}_k(\cl(\st(e)) \setminus \st(e)) 
	\rightarrow 
	\tilde{H}_k(\KC{i} \setminus \st(e)) & 
	\rightarrow
	\tilde{H}_k(\KC{i}) 
	\rightarrow
	\tilde{H}_{k-1}(\cl(\st(e)) \setminus \st(e))
	\rightarrow \cdots
\end{align*}
Thus, it is sufficient to show that $\tilde{H}_k(\cl(\st(e)) \setminus \st(e))$ is trivial for $k=0,1$.

It is possible to define a complex $K_0$ on $\openY = \{y \in S \ \vert \ \din(y, p) < 2\bar{r} \textrm{ and } \din(y, q) < 2\bar{r} \}$, where $\bar{r} = \frac{\din(p,q)}{2}$, such that $K_0 \subseteq 
\cl(\st(e)) \setminus \st(e) 
\subseteq 
\KC{i}$.
The proof follows by showing that $K_0$ has trivial homology, and proving that $\cl(\st(e)) \setminus \st(e)$ has the same homology of $K_0$ in degrees zero and one.
\end{proofsketch}
%
\begin{theorem}
\label{thm:main}
Let $S$ be a finite set of points in $(\mathbb{R}^d, \din)$. 
Given a single edge-length range $(r, r + \varepsilon)$ of \v{C}ech complexes of $S$, if $H_k(\flagr) \rightarrow H_k(\cechr)$ is an isomorphism, then $H_k(\flagre) \rightarrow H_k(\cechre)$ is also an isomorphism for $k=0,1$.
\end{theorem}
\begin{proof}
Given the Alpha flag and \v{C}ech edge-by-edge filtrations of $S$ on $(r, r + \varepsilon)$, we prove that if $H_k(\KAF{i-1}) \rightarrow H_k(\KC{i-1})$ is an isomorphism, then $H_k(\KAF{i}) \rightarrow H_k(\KC{i})$ is also an isomorphism for $k=0,1$ and each $1 \leq i \leq n_i$.
The result follows by chaining these isomorphisms.
We write $e=\{p, q\}$ for the only edge in $\KC{i}$ not in $\KC{i-1}$, and define $\bar{r} = \frac{\din(p,q)}{2}$.
The case in which $e$ is a $\linf$-Delaunay edge and the case in which it is not are treated separately.
The idea is to show that in the former case $H_k(\KAF{i})$ and $H_k(\KC{i})$ change in the same way for $k=0,1$, while in the latter case there are no changes in homological degrees zero and one.

\medskip 

\noindent
\textsc{CASE 1:} $e$ is $\linf$-Delaunay 
\\
We further subdivide this case in two subcases.
\begin{enumerate}
    \item[1.] The edge $e$ is the only simplex added going from $\KAF{i-1}$ to $\KAF{i}$.
    \item[2.] The edge $e$ and other simplices, which are cliques containing $e$, are added going from $\KAF{i-1}$ to $\KAF{i}$.
\end{enumerate}

\noindent
\textsc{Subcase 1.1.} By Lemma \ref{lemma:proof-main-1-add-one-edge} $e$ is the only simplex added going from $\KC{i-1}$ to $\KC{i}$.
Then, $e$ either deletes the same connected component or creates the same degree-one homology class in $\KAF{i}$ and $\KC{i}$, because $\KAF{i-1}$ and $\KC{i-1}$ have the same vertices and $H_k(\KAF{i-1}) \rightarrow H_k(\KC{i-1})$ is an isomorphism for $k=0,1$.
Thus, $H_k(\KAF{i}) \rightarrow H_k(\KC{i})$ induced by the inclusion $\KAF{i} \subseteq \KC{i}$ is also an isomorphism for $k=0,1$.

\smallskip 

\noindent
\textsc{Subcase 1.2.} Apart from $e=\{p, q\}$, the complexes $\KAF{i}$ and $\KC{i}$ contain at least a triangle $\tau = \{p, q, y\}$ not in $\KAF{i-1}$ and $\KC{i-1}$.
Thus, $e$ cannot either delete connected components or create degree-one homology classes in both. 
This follows because $p$ and $q$ are connected in $\KAF{i-1}$ and $\KC{i-1}$ via the edges $\{p, y\}$, $\{q, y\}$, and any new $1$-cycle would need to contain $e$, but would also be homologous to a $1$-cycle containing $\{p, y\}$ and  $\{y, q\}$ in place of $e$.

We show that the same degree-one homology classes are deleted in $\KAF{i}$ and $\KC{i}$.
For this, we further refine the inclusion $\KC{i-1} \subseteq \KC{i}$ into the subfiltration
$$
\KC{i-1} =
\LC{0} \subseteq \LC{1} \subseteq
\ldots \subseteq \LC{n_j} 
= \KC{i},
$$
where $\LC{1}$ is equal to $\KC{i-1}$ union the simplices added going from $\KAF{i-1}$ to $\KAF{i}$, and $\LC{j}$ is equal to $\LC{j-1}$ union a triangle $\tauC_j$ and the higher-dimensional simplices containing this triangle for each $2 \leq j \leq n_j$.
Note that $\tauC_j$ is a non-Delaunay triangle with at least one non-Delaunay edge, i.e. $\tauC \not\in \KAF{i}$.
In particular, $\LC{j}$ is defined by choosing $\tauC_j = \{p, q, \yC_j\}$ among the triangles containing $e$ and such that 
$\cB{\bar{r}}{\yC_j} \cap \A^{\bar{r}}$ 
intersects 
$\bigcup_{\yC \in \YC} B_{\bar{r}}(\yC)$, 
where $\YC = \{ \yC \in S \ \vert \ \{p, q, \yC\} \in \LC{j-1}\}$.
We prove that it is possible to define $\LC{j}$ this way by showing that one such $\tauC_j$ exists for each $2 \leq j \leq n_j$.

Suppose $\cB{\bar{r}}{\yC_j} \cap \A^{\bar{r}}$ 
and
$\bigcup_{\yC \in \YC} B_{\bar{r}}(\yC)$ do not intersect for any of the triangles $\{\tauC_j\}$ still to be added in $\LC{n_j}$.
It follows that also the union of the closed balls centered in the vertices $\yC_j$ of the triangles $\{\tauC_j\}$ does not intersect $\bigcup_{\yC \in \YC} B_{\bar{r}}(\yC)$.
Thus, the boundary of this union of closed balls intersects $\A^{\bar{r}}$ in a point $z$, because the latter is not covered by any union of open balls of radius $\bar{r}$ by Proposition \ref{prop:delaunay-edge}.
But $z$ is a witness of some $\tauC_j$ still to be added, because the boundary of a union of balls is a subset of the union of their boundaries, making $\tauC_j$ into a $\linf$-Delaunay simplex.
This contradicts the fact that all $\linf$-Delaunay triangles are added going from $\LC{0}$ to $\LC{1}$.

To conclude, we have that going from $\KAF{i-1}$ to $\KAF{i}$ and from $\KC{i-1} = \LC{0}$ to $\LC{1}$ the same degree-one homology classes are deleted, because the same set of triangles containing only $\linf$-Delaunay edges is added in $\KAF{i}$ and $\LC{1}$.
Moreover, we can show that no degree-one homology class is deleted at each step $\LC{j-1} \subseteq \LC{j}$ for $2 \leq j \leq n_j$.
By definition of $\LC{j}$, $\cB{\bar{r}}{\yC_j}$ has to intersect at least one $B_{\bar{r}}(\yC')$ such that $\yC' \in \YC = \{ \yC \in S \ \vert \ \{p, q, \yC\} \in \LC{j-1} \}$.
So $\din(\yC_j, \yC') < 2\bar{r}$, because $\cB{\bar{r}}{\yC_j} \cap B_{\bar{r}}(\yC') \neq \emptyset$, and it follows $\{\yC_j, \yC'\} \in \KC{i-1} \subseteq \LC{j-1}$.
Besides, the edges $\{p, \yC'\}$, $\{q, \yC'\}$ are in $\LC{j-1}$, because $\yC' \in \check{\mathcal{Y}}$, and the edges $\{p, \yC_j\}$, $\{q, \yC_j\}$ are in $\LC{j-1}$, otherwise $\tauC$ could not be added in $\LC{j}$.
Thus, we have $\{p, \yC_j, \yC'\}$, $\{q, \yC_j, \yC'\} \in \LC{j-1}$, because all their edges are in $\LC{j-1}$, and $\{p, q, \yC'\} \in \LC{j-1}$ by definition of $\YC$.
So, adding the triangle $\tauC_j$ in $\LC{j}$ also adds the tetrahedron $\{p, q, \yC', \yC_j\}$ in $\LC{j}$.
Thus, any $2$-cycle containing $\tauC_j$ is homologous to one containing the other three faces of $\{p, q, \yC', \yC_j\}$ in its place, and we conclude that $\tauC_j$ cannot delete any degree-one homology class in $\LC{j}$.

\medskip 

\noindent
\textsc{CASE 2:} $e$ is non-Delaunay 
\\
We have $\KAF{i-1} = \KAF{i}$, because $e$ is not added to Alpha flag complexes.
Moreover, $H_k(\KC{i-1}) \rightarrow H_k(\KC{i})$ is an isomorphism for $k=0, 1$, by Lemma \ref{lemma:removing-non-delaunay-edge}. 
So, homology in degrees zero and one remains unchanged at $\KAF{i-1} \subseteq \KAF{i}$ and $\KC{i-1} \subseteq \KC{i}$.
\end{proof}
%
\begin{corollary}
\label{cor:alpha-flag-and-cech-have-same-persistence}
Let $S$ be a finite set of points in $(\mathbb{R}^d, \din)$.
Given a finite set of monotonically increasing real-values $\mathcal{R} = \{r_i\}_{i=1}^m$, the Alpha flag $\Fflag$ and \v{C}ech filtrations $\Fcech$ of $S$ have the same persistence diagrams in homological degrees zero and one.
\end{corollary}
\begin{proof}
Let 
$K_{r_1}^{AF} 
\subseteq 
K_{r_2}^{AF} 
\subseteq \ldots \subseteq 
K_{r_m}^{AF}$
and    
$K_{r_1}^{\check{C}} 
\subseteq 
K_{r_2}^{\check{C}} 
\subseteq \ldots \subseteq 
K_{r_m}^{\check{C}}$ 
the Alpha flag and \v{C}ech filtrations on $\mathcal{R}$.
We have that 
$H_k(\flagri) \rightarrow H_k(\cechri)$ 
is an isomorphism for each 
$1 \leq i \leq m$ and $k=0, 1$ by chaining the isomorphisms obtained by applying Theorem \ref{thm:main} to each single edge-length range $(r, r+\varepsilon) \subseteq [0, r_i]$. 
The proof follows by the Persistence Equivalence Theorem of \cite[Section~7.2]{edelsbrunner2010computational}
\end{proof}
\begin{remark}
The above result extends to generic dimension $d$ the equivalence of persistence diagrams in homological degrees zero and one proven in \cite{kerber1} for two and three-dimensional convex objects.
\end{remark}

\section{Equivalence of Minibox Complexes}
\label{sec:minibox}
Given the results of the previous section, we  prove that Minibox complexes can also be used to compute \v{C}ech persistence diagrams in homological degrees zero and one.
We show that $\flagr \subseteq \minir$, and use this inclusion to derive isomorphisms $H_k(\flagr) \rightarrow H_k(\minir)$ for $k=0,1$ for any $r \in \R$. 

A disadvantage of Minibox complexes, compared to Alpha flag complexes, is that they contain more edges, and so produce larger filtrations, leading to slower computation of persistence diagrams.
However, we provide efficient algorithms for finding Minibox edges in the next section, which
 are used for the computation of persistence diagrams in Section \ref{sec:experiments}.
On the other hand, efficient algorithms for finding Alpha flag edges, i.e. $\linf$-Delaunay edges, are only available for points in $\R^2$ \cite{shute1991planesweep}.
We also show that the expected number of edges in Minibox complexes is $\boundMiniEdges$, when considering $n$ randomly sampled points $S \subseteq (\R^d, \din)$.
This is substantially smaller than the quadratic number of edges in \v{C}ech complexes. 
By comparison, the number of edges in the Minibox complex is  within a polylogarithmic factor of the number of edges in an Alpha complex (using the $\ell_2$ metric). 
The expected size of the Alpha complex using the $\ell_\infty$ metric is not known, but we conjecture that it is linear in the number of vertices, just as in the standard $\ell_2$ case. 

Recall that $\minipq = \prod_{i=1}^d \big( \min\{p_i, q_i\}, \max\{p_i, q_i\} \big)$, used in the definition of Minibox complexes, is the interior of the minimal bounding box of any $p, q \in S$.
In the following, we refer to $\minipq$ as the minibox of $p$ and $q$, and to any pair $\{p, q\}$ such that $\minipq \cap S = \emptyset$ as a Minibox edge.
%
\begin{proposition}
\label{prop:minibox-property}
Let $S$ be a finite set of points in $(\mathbb{R}^d, \din)$, $e=\{p, q\}$ a pair of points of $S$, and $\minipq$ the minibox of $p$ and $q$.
If there exists $y\in S$ such that $y \in \minipq$, then $e$ is not an edge of the $\linf$-Delaunay complex of $S$.
\end{proposition}
\begin{proof}
Given $\bar{r} = \frac{\din(p,q)}{2}$, we have $\A^{\bar{r}} = \cB{\bar{r}}{p} \cap \overline{B_{\bar{r}}(q)}$ by Proposition \ref{prop:delaunay-edge}. 
Equivalently $\A^{\bar{r}} = \prod_{i=1}^d [b_i - \bar{r}, a_i + \bar{r}]$, where $a_i = \min\{p_i, q_i\}$ and $b_i = \max\{p_i, q_i\}$ for each $1 \leq i \leq d$.
Then, given $y\in \minipq$, it follows that $a_i < y_i < b_i$ for each $1 \leq i \leq d$, implying 
$y_i-\bar{r} < b_i-\bar{r}$
and
$a_i+\bar{r} < y_i+\bar{r}$.
Thus, $[b_i-\bar{r}, a_i+\bar{r}] \subset (y_i-\bar{r}, y_i+\bar{r})$ for each $1 \leq i \leq d$, and $\A^{\bar{r}} \subset B_{\bar{r}}(y)$.
The result follows applying Proposition \ref{prop:delaunay-edge}.
\end{proof}
\begin{remark}
The above result implies that any $\linf$-Delaunay edge is also a Minibox edges, and so $\flagr \subseteq \minir$ for any $r \in \R$.
\end{remark}
%
\begin{theorem}
\label{thm:alpha-flag-and-minibox}
Let $S$ be a finite set of points in $(\mathbb{R}^d, \din)$.
Given the Alpha flag $\flagr$ and Minibox $\minir$ complexes with radius $r \in \R$, then $H_k(\flagr)$ and $H_k(\minir)$ are isomorphic in homological degrees zero and one.
\end{theorem}
\begin{proof}
We have 
$\flagr \subseteq \minir \subseteq \cechr$,
and we know that $H_k(\flagr) \rightarrow H_k(\cechr)$ is an isomorphism for $k=0, 1$ from the discussion in the proof of Corollary \ref{cor:alpha-flag-and-cech-have-same-persistence}.
Thus, we have the following commutative diagrams, implying that $H_k(\flagr) \rightarrow H_k(\minir)$ is injective for $k = 0, 1$ and any $r\in \mathbb{R}$.
\begin{equation*}
    \centering
    \begin{tikzcd}
        \flagr \ar[rd, hook] \ar[rr, hook] &   & \cechr\\
            & \minir \ar[ur, hook] &
    \end{tikzcd}
    \quad \Longrightarrow \quad
    \begin{tikzcd}
        H_k(\flagr) \ar[rd, hook] \ar[rr, "\cong"] &   & H_k(\cechr)\\
            & H_k(\minir) \ar[ur, two heads] &
    \end{tikzcd}
\end{equation*}
To conclude our proof, it is sufficient to show the surjectivity of $H_k(\flagr) \rightarrow H_k(\minir)$ for $k=0, 1$.

First, note that $\minir$ contains more edges than $\flagr$, and that they have the same vertices.
Thus, a connected component in $\minir$ corresponds to one or more connected components in $\flagr$, and $H_0(\flagr) \rightarrow H_0(\minir)$ induced by the inclusion has to be surjective.

To prove the surjectivity of $H_1(\flagr) \rightarrow H_1(\minir)$, we show that for any $[\gamma] \in H_1(\minir)$ there exists a $1$-cycle representing $[\gamma]$ containing only $\linf$-Delaunay edges of length less than or equal to $2r$.

Let $\gamma$ be a $1$-cycle in $\minir$ representing $[\gamma] \in H_1(\minir)$, and $e=\{p, q\}$ a non-Delaunay edge in $\gamma$ of maximum length.
We have $\A^{\bar{r}} =  \cB{\bar{r}}{p} \cap \cB{\bar{r}}{q}$, where $\bar{r} = \frac{\din(p, q)}{2}$ by Proposition \ref{prop:delaunay-edge}. 
Given $\openY = \{y \in S \ \vert \ \din(y, p) < 2\bar{r} \textrm{ and } \din(y, q) < 2\bar{r} \}$, we equivalently have $\openY = S \cap B_{2\bar{r}}(p) \cap B_{2\bar{r}}(q)$ and $\openY = S \cap \bar{r}(\A^{\bar{r}})$, because $\varepsilon(\A^{\bar{r}})$ equals $\cB{\bar{r} + \varepsilon}{p} \cap \cB{\bar{r} + \varepsilon}{q}$ by the properties of boxes described in Appendix \ref{sec-app:delaunay-edges}.
For points in $\R^2$, these sets are illustrated in Figure \ref{fig:minibox-proof}, where $\A^{\bar{r}}$ is represented by a thickened vertical line between $p$ and $q$.
Moreover, we have $\minipq \subseteq \bar{r}(c) \subseteq \bar{r}(\A^{\bar{r}})$, where $c = \frac{p+q}{2}$, because taking $\varepsilon$-thickenings preserves inclusions, and $\minipq$ has sizes of length less than or equal to $2\bar{r}$ and center $c$.
Then, because $e$ is a non-Delaunay edge,
$\A^{\bar{r}}$ must be covered by the union of balls centered in the points of $S\setminus \{p, q\}$ by Proposition \ref{prop:delaunay-edge}.
Thus, at least one $y\in S\setminus\{p,q\}$ is such that $B_{\bar{r}}(y)$ intersects $\A^{\bar{r}}$, i.e. $\openY \neq \emptyset$.
Given $\bar{y} \in \openY$ to be a point realizing 
$$\min_{y \in \openY} \din(y, \minipq),$$
we have that $\textrm{Mini}_{p\bar{y}}$ and $\textrm{Mini}_{q\bar{y}}$ do not contain points in $S\setminus \{p,q,\bar{y}\}$, as we can show a contradiction otherwise.

Suppose there exists either $y' \in S \setminus \openY$ or $y'' \in \openY$ belonging to one of these two miniboxes.
Without loss of generality, we assume either $y' \subseteq \textrm{Mini}_{p\bar{y}}$ or $y'' \subseteq \textrm{Mini}_{p\bar{y}}$.
In the former case we have $\textrm{Mini}_{p\bar{y}} \subseteq \bar{r}(\A^{\bar{r}})$, because $p$ is on the boundary of $\bar{r}(\A^{\bar{r}})$ and $\bar{y}$ in its interior. 
So $y' \in \bar{r}(\A^{\bar{r}})$, implying that $y' \in \openY$, which is a contradiction.
In the latter case, it must be that $\din(y'', \minipq) < \din(\bar{y}, \minipq)$ by definition of $\textrm{Mini}_{p\bar{y}}$ and $\din$, which is in contradiction with $\bar{y}$ minimizing the distance to $\minipq$.

So, there exists a vertex $\bar{y} \in \openY$ of the Minibox complex connected to $p$ and $q$ by the edges $\{p, \bar{y}\}$ and $\{q, \bar{y}\}$.
These are shorter than $2\bar{r}$, by definition of $\openY$, so that $\{p, \bar{y}\}, \{q, \bar{y}\} \subseteq \minir$.
Swapping $\{p, \bar{y}\}$ and $\{q, \bar{y}\}$ for $e$ in $\gamma$, we obtain a $1$-cycle  homologous to $\gamma$. 
This procedure can be repeated only a finite number of times, as we have a finite number of non-Delaunay edges, and at each iteration an edge is replaced by edges which are strictly shorter, so that future iterations cannot reintroduce edges which were previously removed.
When the procedure cannot be repeated, we have a $1$-cycle $\gamma'$ in $\minir$ homologous to $\gamma$, containing only $\linf$-Delaunay edges.
Hence, $\gamma'$ represents a degree-one homology class in the Alpha flag complex which is mapped into $[\gamma]$ by $H_1(\flagr) \rightarrow H_1(\minir)$.
\end{proof}
\begin{corollary}
Let $S$ be a finite set of points in $(\mathbb{R}^d, \din)$.
Given a finite set of monotonically increasing real-values $\mathcal{R} = \{r_i\}_{i=1}^m$, the Alpha flag $\Fflag$ and Minibox filtrations $\Fmini$ of $S$ have the same persistence diagrams in homological degrees zero and one.
\end{corollary}
\begin{proof}
Follows from the Persistence Equivalence Theorem of \cite[Section~7.2]{edelsbrunner2010computational} as for Corollary \ref{cor:alpha-flag-and-cech-have-same-persistence}.
\end{proof}
\begin{figure}[!tp]
    \centering
    \includegraphics[width=2.8in]{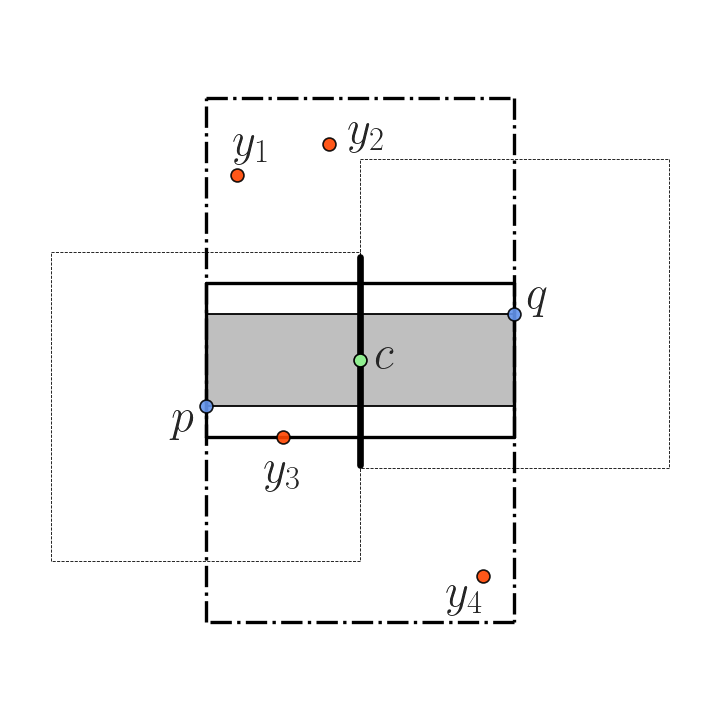}
    \caption{The pair $(p,q)$ is not a $\linf$-Delaunay edge, but it is a Minibox edge. $\minipq$ is the gray region having $p$ and $q$ as two vertices. The set $\openY$ consists of four $y_i$ points contained in the rectangle $\bar{r}(\A^{\bar{r}})$, whose boundary is represented by a dash-dot line.}
    \label{fig:minibox-proof}
\end{figure}

\paragraph{Number of Minibox edges}
We conclude this section by studying the number of edges that a Minibox complex $\minir$ can contain.
We start by noting that in the worst case a Minibox complex can contain $O(n^2)$ edges.
For example the union of 
$S_1 = \left\{p_i = \left( 0 + \frac{i}{n}, 1 - \frac{i}{n}\right) \right\}_{i=1}^n$ 
and 
$S_2 = \left\{ q_j = \left( 2 + \frac{j}{n}, 1 - \frac{j}{n}\right) \right\}_{j=1}^n$ 
is a set of $2n$ points in $\R^2$, on parallel line segments, such that the miniboxes $\textrm{Mini}_{p_iq_j}$ for $1 \leq j \leq i \leq n$ do not contain any point in $S_1 \cup S_2$.
Thus, the Minibox complex of $S_1 \cup S_2$ will contain more than $\frac{n(n-1)}{2}$ edges for a large enough radius parameter, see Figure \ref{fig:example-2n-minibox}.
For comparison, Figure \ref{fig:example-2n-delaunay} illustrates the $4n-3$ Alpha flag (i.e. $\linf$-Delaunay) edges on the same set of points.

Then, we study the expected number of Minibox edges on randomly sampled points.
Recall that a point $p$ \emph{dominates} $q$ if each of the coordinates of $p$ is greater than the corresponding coordinate of $q$.
Moreover, $p$ \emph{directly dominates} $q$ if $p$ dominates $q$ and there is no other point $y \in S$ such that $p$ dominates $y$ and $y$ dominates $q$.
It follows that if $p$ directly dominates $q$, then $\{p, q\}$ is a Minibox edge and $p$ and $q$ are not collinear.
On the other hand, if $\{p, q\}$ is a Minibox edge and $p$ and $q$ are not collinear, then $p$ and $q$ might not dominate each other. 
This is the case for $\{p, q_2\}$ and $\{p, q_4\}$ in Figure \ref{fig:domination-pairs}.
%
\begin{figure}[!tp]
    \centering
    \begin{subfigure}[b]{2in}
        \centering
        \includegraphics[width=2in]{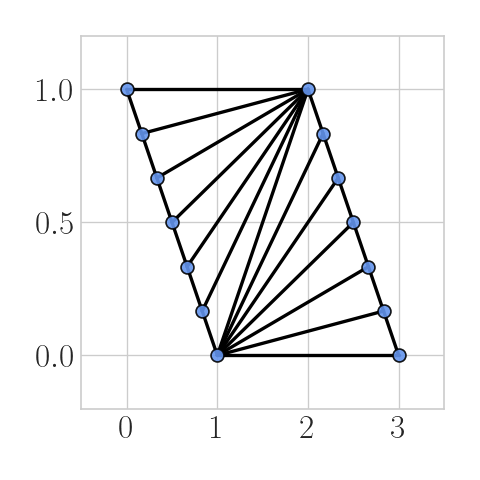}
        \caption{}
        \label{fig:example-2n-delaunay}
    \end{subfigure}
    \quad
    \begin{subfigure}[b]{2in}
        \centering
        \includegraphics[width=2in]{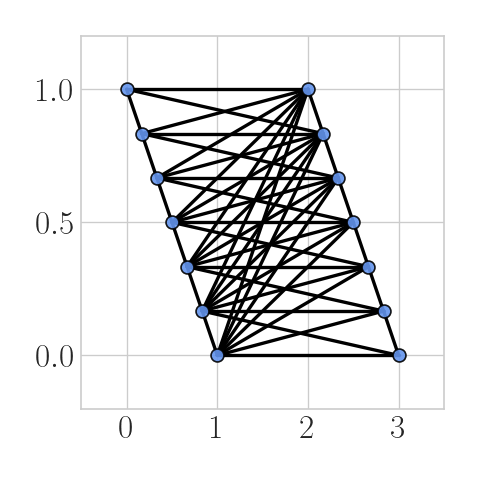}
        \caption{}
        \label{fig:example-2n-minibox}
    \end{subfigure}
    \quad
    \begin{subfigure}[b]{2in}
        \centering
        \includegraphics[width=2in]{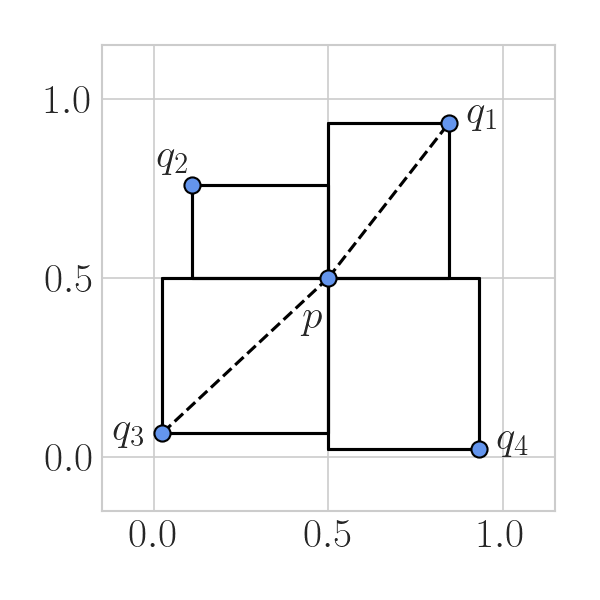}
        \caption{}
        \label{fig:domination-pairs}
    \end{subfigure}
    \caption{\textbf{(a)} $\linf$-Delaunay edges of $S_1 \cup S_2$. \textbf{(b)} Minibox edges of $S_1 \cup S_2$. \textbf{(c)} Only $\{p, q_1\}$ and $\{q_3, p\}$ are direct dominance pairs.}
    \label{fig:examples-S1-S2-types-domination}
\end{figure}
However, as long as $p, q \in \R^d$ are not collinear, there is a sequence of maximum $d-1$ reflections about the coordinate hyperplanes so that either $p$ dominates $q$ or $q$ dominates $p$.
Thus, an empty minibox $\minipq$ corresponds to a direct dominance pair $\{p, q\}$ via one of $2^{d-1}$ possible sequences of reflections.
\begin{proposition}
\label{prop:expected-edges-minibox}
Let $S$ be a finite set of $n$ uniformly distributed random points in $(\R^d, \din)$ such that there does not exist any pair of collinear points in $S$.
The expected number of Minibox edges of $S$ is $\boundMiniEdges$. 
\end{proposition}
\begin{proof}
The expected number of maximal points of $S$ (i.e. points not dominated by any other point of $S$) is $O\big(\ln^{d-1}(n)\big)$ \cite{bentley1978average-number-maxima}.
As discussed in \cite[Section 2]{dumitrescu2013maximal-empty-boxes}, it follows that the expected number of direct dominance pairs is $O\big(n\ln^{d-1}(n)\big)$, because the points directly dominated by any $p\in S$ are the maximal points in the subset of points of $S$ dominated by $p$.
Finally, the expected number of empty miniboxes, corresponding to Minibox edges, is $\boundMiniEdges$ because an empty minibox can be mapped to a direct dominance pair by one of $2^{d-1}$ possible sequences of reflections as discussed above.
\end{proof}

\section{Minibox Edge Algorithms}
\label{sec:algorithms}
We present algorithms for finding all pairs of points $\{p, q\}$ of $S \subseteq (\R^d, \din)$ such that $\minipq \cap S$ is empty.
By definition these are all the edges a Minibox complex can contain.
We study the two-dimensional, three-dimensional, and higher-dimensional cases separately. 
In every case, we assume $S$ to contain $n$ points and to be preprocessed so to eliminate collinear points.
For $d=2$, we reference algorithms for rectangular visibility and direct dominance.
For $d=3$, we recall known results on direct dominance and present two new algorithms.
These maintain dynamic tree data structures that can be used to efficiently determine whether $\minipq \cap S$ is empty or not. 
For general dimension $d$, the problem of finding all empty miniboxes  $S$ can be seen as the problem of testing offline orthogonal range emptiness. In describing the algorithms, we provide pointers to the relevant results on range queries.

We provide an implementation of our algorithms in the form of the \href{https://pypi.org/project/persty/}{persty} Python package, the source code of which is available at \href{https://github.com/gbeltramo/persty}{https://github.com/gbeltramo/persty}.

\paragraph{Preprocessing}
The algorithms we are going to present assume that there are no collinear points in $S$, i.e. $p_i \neq q_i$ for each $ 1 \leq i \leq d$ and $p, q \in S$.
This assumption can be verified in time $O(d n \log(n))$ by sorting $S$ along each of the coordinates of its points, and then iterating on each of the $d$ instances of sorted points to see if any two consecutive points share a coordinate.

If the assumption is not met for a coordinate $\hat{i}$,  we sample $n$ real values from a uniform distribution on $(-\varepsilon, \varepsilon)$, and sum these values to the $\hat{i}$-th coordinates of points in $S$.
With probability $1$, this results in the points of $S$ not being collinear in their $\hat{i}$ coordinate.
We may choose $\varepsilon$ to be an arbitrarily small real value for each such $\hat{i}$.
This way the \v{C}ech persistence diagrams of the original $S$ and of the transformed $S$ without collinear points are arbitrarily close by the Stability Theorem.
Hence, it follows that persistence diagrams computed using the Minibox complexes of the original and transformed points are arbitrarily close.

\paragraph{Points in two dimensions}
The definition of \emph{rectangular visibility} for $p$ and $q$ given in   \cite{overmars1988rectangular-visibility} is equivalent to $\minipq \cap S = \emptyset$.
Thus, the algorithm presented in \cite{overmars1988rectangular-visibility} reports the Minibox edges of $S \subseteq (\R^2, \din)$ in $O(n\log(n) +k)$ time and $O(n)$ space, where $k$ is the number of reported edges and $n$ the number of points of $S$.
Furthermore, any algorithm for finding the direct dominance pairs of $S$ can be applied twice (to $S$ and its reflection over the $x$ axis) to find the Minibox edges of $S$.
So, the same time and space complexities are obtained by using the divide-and-conquer algorithm of \cite{guting1989fast-direct-dominance} for direct dominance pairs of points in $\R^2$. 

\paragraph{Known results for points in three dimensions}
In \cite{guting1989fast-direct-dominance} it is given an algorithm for direct dominances of $S \subseteq \R^3$ taking $O\big( (n+k') \log^2(n) \big)$ time and $O\big((n+k')\log(n)\big)$ space, where $k'$ is the number of direct dominance pairs and $n$ the number of points of $S$.
Thus, the Minibox edges of $S$ can be found by applying this algorithm four times, because of the relation between Minibox edges and direct dominance pairs discussed in Section \ref{sec:minibox}.
The resulting algorithm for Minibox edges has time complexity $O(k \log^2(n) )$ and uses $O(k\log(n))$ space, where $k$ is the number of Minibox edges of $S$.

Given $n$ randomly sampled points in $\R^3$, the best, average and worst-case values of $k$ are $O(n)$, $O(n \log^2(n))$, and $O(n^2)$ respectively.
The best case follows because there exists a Minibox edge from each point to its nearest neighbour, and there are $n-1$ Minibox edges on $n$ points on any non axis-parallel line in $\R^3$.
The average case is $O(n\log^2(n))$ by Proposition  \ref{prop:expected-edges-minibox} with fixed dimension $d=3$.
The worst case is discussed with an example in Section \ref{sec:minibox}, and illustrated in Figure \ref{fig:example-2n-minibox}.
%
\begin{table}[tb]
 \caption{Complexities of Minibox edge algorithms for best, average, and worst-case $k$ of randomly sampled $S \subseteq \R^3$.}
 \label{table:complexities-algorithms-3D}
  \centering
  \begin{tabular}{|l | l | c | c | c | c |}
    \cline{3-6}
        \multicolumn{1}{ l }{} %
        & \multicolumn{1}{ l }{}
        & \multicolumn{1}{|c|}{\textbf{Generic}}  
        & \multicolumn{1}{|c|}{\textbf{Best}}
        & \multicolumn{1}{|c|}{\textbf{Average}}
        & \multicolumn{1}{|c|}{\textbf{Worst}}
    \\
    \hline
    \rule{0pt}{11pt}
        \multirow{2}*{Known} %
        & Time
        & $O(k\log^2(n))$ 
        & $O(n\log^2(n))$ 
        & $O(n\log^4(n))$ 
        & $O(n^2\log^2(n))$ 
    \\
    \cline{2-6}
    \rule{0pt}{11pt}
        & Space
        & $O(k \log(n))$
        & $O(n\log(n))$
        & $O(n\log^3(n))$
        & $O(n^2\log(n))$
    \\
    \hline
    \hline
    \rule{0pt}{11pt}
        \multirow{2}*{Algorithm 1} %
        & Time
        & $O(n^2\log(n))$
        & $O(n^2\log(n))$
        & $O(n^2\log(n))$
        & $O(n^2\log(n))$
    \\
    \cline{2-6}
    \rule{0pt}{11pt}
        & Space 
        & $O(n)$
        & $O(n)$
        & $O(n)$
        & $O(n)$
    \\
    \hline
    \hline
    \rule{0pt}{11pt}
        \multirow{2}*{Algorithm 2} %
        & Time
        & $O(k\log^2(n))$ 
        & $O(n\log^2(n))$ 
        & $O(n\log^4(n))$ 
        & $O(n^2\log^2(n))$ 
    \\
    \cline{2-6}
    \rule{0pt}{11pt}
        & Space 
        & $O(n\log^2(n))$
        & $O(n\log^2(n))$
        & $O(n\log^2(n))$
        & $O(n\log^2(n))$
    \\
    \hline
    \end{tabular}
\end{table}

In the following, we present two novel algorithms that improve the space complexity, for average and worst-case $k$,  of the Minibox edge algorithm derived from known direct dominance algorithms.

\paragraph{Algorithm for three-dimensional points using $O(n)$ space}
Let $S$ be a set of $n$ points in $\R^3$, which does not contain collinear points.
We describe a $O(n^2 \log(n))$ algorithm using $O(n)$ storage for finding the Minibox edges of $S$.
It pseudocode is given in Algorithm \ref{alg:3D-stair}.
This improves both the time and space complexities for worst-case $k$ of the algorithm derived from known direct dominance results, see Table \ref{table:complexities-algorithms-3D}. 

The idea is to sweep a plane in the $z$ direction for each $p \in S$, so to find all Minibox edges $\{p, q\}$.
In particular, we define the starting sweep-plane to be the set of points $\{ v \in \R^3 \ \vert \ v_z = p_z\}$ for each $p = (p_x, p_y, p_z) \in S$.
Moreover, we assume the sweep-plane to be centered in $(p_x, p_y)$, so that its first quadrant consists of the points with $x$ and $y$ coordinates greater than $p_x$ and $p_y$ respectively. 
As points $q \in S$ are encountered moving upward, we check whether $\minipq$ does not contain other points of $S$ using a dynamic red-black tree data structure.
In order to simplify our exposition, we only discuss the case in which the projection of $q$ belongs to the first quadrant of the sweep-plane.
This is sufficient to prove the correctness of Algorithm \ref{alg:3D-stair} by the definition of $p_{xy}$ and $q_{xy}$ on line $8$, and because distinct red-black trees are defined on line $5$ for each of the four quadrants of the sweep-plane.

We show that for any $p \in S$ Algorithm $\ref{alg:3D-stair}$ correctly identifies the pairs $\{p, q\}$ with $q \in S$ and $p_z < q_z$ such that $\minipq \cap S$ is empty.
The points of $S$ are first sorted on their $z$ coordinate.
This way each iteration of the inner loop on lines $6-22$ checks the points $q \in S$ with $p_z < q_z$ in sorted order.
As mentioned above, we only consider the case in which $(q_x, q_y)$ lies in the first quadrant of the sweep-plane with respect to $(p_x, p_y)$.
We define
$$Q' 
=
\{q' \in S 
\ \vert \  
p_x < q'_x
\textrm{ and }
p_y < q'_y
\textrm{ and }
p_z < q'_z < q_z 
\},$$
and
$$Q'_{xy} 
=
\{
q'_{xy} \in \R^2
\ \vert \ 
q'_{xy} = 
    (|q'_x - p_x|, 
     |q'_y - p_y|)
\textrm{ for each }
q' \in Q'
\},$$
where $p = (p_x, p_y, p_z)$, $q = (q_x, q_y, q_z)$, $q' = (q'_x, q'_y, q'_z) \in \R^3$.
Because the points in $Q'$ are the only points that may be contained in $\minipq$, it follows that $\minipq$ is empty if and only if the two-dimensional minibox $\minitwodim$ does not contain any point of $Q'_{xy}$.
%
\begin{figure}[!tp]
    \begin{minipage}{\textwidth}
        \begin{algorithm}[H] 
    	    \begin{algorithmic}[1] 
	        \caption{Minibox edges of a finite set of points $S$ in three-dimensions using red-black trees.}
	        \label{alg:3D-stair}
	\STATE \textbf{input:} array $points$, the finite set of points $S$ in $(\mathbb{R}^3, \din)$
	\STATE $edges \gets$ empty list of two-tuples of integers
	\STATE Sort $points$ on their $z$-coordinate
	\FOR {$i=0$ to $|S|-1$}
	    \STATE $T_1, T_2, T_3, T_4 \gets$ empty red-black trees, one per quadrant
	    \FOR {$j=i+1$ to $|S|-1$}
	        \STATE $p, q \gets points[i], points[j]$
    	    \STATE $p_{xy}, q_{xy} \gets (0, 0), (|q_x - p_x|, |q_y - p_y|)$
    	    \STATE $l \gets$ index such that $(q_x, q_y)$ is in the $l$-th quadrant of the sweep-plane centered in $(p_x, p_y)$
    	    \IF {$T_l$ is non-empty}
    	        \STATE $\hat{q}'_{xy} \gets$ first element to the left of $q_{xy}$ in $T_l$ bisecting on $|q_x - p_x|$
    	        {\IF {$q'_{xy}$ does not exist}
    	            \STATE Delete the points in $T_l$ that dominate $q_{xy}$, insert $q_{xy}$ in $T_l$ at $|q_x - p_x|$, and add $(i, j)$ in $edges$
    	        \ELSE
    	            {\IF {$\hat{q}'_{xy}$ dominates $q_{xy}$}
    	                 \STATE Delete the points in $T_l$ that dominate $q_{xy}$, insert $q_{xy}$ in $T_l$ at $|q_x - p_x|$, and add $(i, j)$ in $edges$
    	             \ENDIF}
    	        \ENDIF}
        	\ELSE 
        	    \STATE Insert $q_{xy}$ in $T_l$ at $|q_x - p_x|$, and add $(i, j)$ in $edges$
        	\ENDIF
	    \ENDFOR
	\ENDFOR
	\RETURN $edges$
        	\end{algorithmic}
        \end{algorithm}
    \end{minipage}
    \centering
    \begin{subfigure}[b]{2in}
        \centering
        \addtocounter{subfigure}{-3}
        \includegraphics[width=2in]{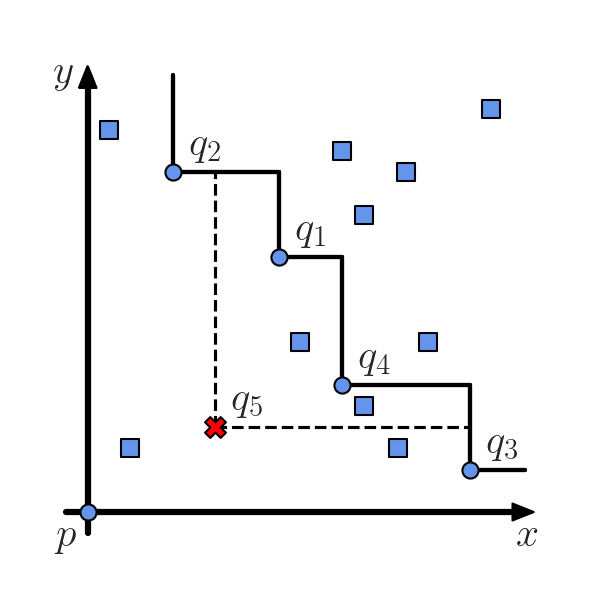}
        \caption{}
        \label{fig:staircase-1}
    \end{subfigure}
    \qquad
    \begin{subfigure}[b]{2in}
        \centering
        \includegraphics[width=2in]{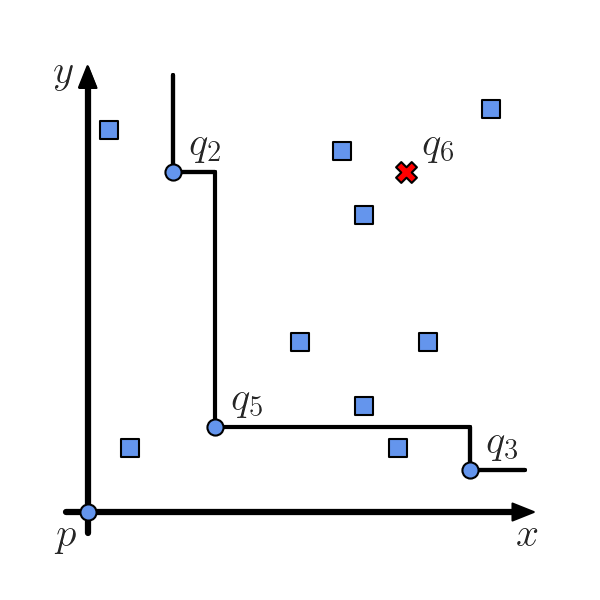}
        \subcaption{}
        \label{fig:staircase-2}
    \end{subfigure}
    \caption{Two iterations of the inner loop of Algorithm \ref{alg:3D-stair}. The points above the sweep-plane are illustrated as squares. \textbf{(a)} The point $q_5$ is reached by the sweep-plane, and $\{p, q_5\}$ is a Minibox edge. So $q_5$ is inserted in $T_1$ (initially containing $q_2$, $q_1$, $q_4$, and $q_3$) after deleting $q_1$ and $q_4$.
    \textbf{(b)} The point $q_6$ is reached by the sweep-plane, but $\{p, q_6\}$ is not a Minibox edges, because $q_6$ dominates $q_5$. So $T_1$ is not updated.}
    \label{fig:algorithm-minibox-3D-1} 
\end{figure}
Thus, a possible strategy could be to store the points of $Q'_{xy}$ in the nodes of $T_1$, and then search on these to find whether there exists $q'_{xy} \in \minitwodim$.
Note that any $q''_{xy} \in Q'_{xy}$ which dominates another $q'_{xy} \in Q'_{xy}$ is such that if $q''_{xy} \in \minitwodim$, then $q'_{xy} \in \minitwodim$ by definition of minibox and $Q'_{xy}$.
Hence, it is sufficient to store a subset of points of $Q'_{xy}$ in $T_1$, and search on those to find a point contained in $\minitwodim$.
In particular, we can restrict to $Q'_1 \subseteq Q'_{xy}$ which we define as the largest subset of $Q'_{xy}$ such that no point in $Q'_1$ dominates a point in $Q'_{xy}$.
By definition, we have that the points in $Q'_1$ correspond to a two-dimensional staircase, i.e. if the points of $Q'_1$ are sorted on their first coordinate, then their second coordinates are monotonically decreasing.
Algorithm \ref{alg:3D-stair} stores the points of $Q'_1$ in the nodes of $T_1$, sorted on their first coordinate.
This is done with the updates on lines $13$, $16$, and $20$.
Moreover, the fact that the points of $Q'_1$ form a staircase implies that $\minitwodim \cap Q'_{xy}$ is empty if and only if $q_{xy}$ does not dominate the point $\hat{q}'_{xy} \in Q'_1$ directly to its left.
Thus, we can search $T_1$ for this $\hat{q}'_{xy}$ with $O(\log(n))$ operations, which can then be used to decide whether to add $\{p, q\}$ to the list of Minibox edges or not.
Figure \ref{fig:algorithm-minibox-3D-1} illustrates two consecutive iterations of the inner loop on lines $6-22$ of Algorithm \ref{alg:3D-stair}, where on the left $Q'_1 = \{q_1, q_2, q_3, q_4\}$ and on the right $Q'_1 = \{q_2, q_3, q_5\}$.

The only data structures maintained by Algorithm \ref{alg:3D-stair} are the red-black trees $T_1$, $T_2$, $T_3$, and $T_4$, each containing at most $n-1$ points.
So, the space complexity of Algorithm \ref{alg:3D-stair} is $O(n)$.
Finally, the inner loop may require to delete and add $O(n)$ points into red-black trees, and search on these same trees $O(n)$ times.
Since either deleting, adding, or searching on a red-black tree requires $O(\log(n))$ operations, we conclude that the inner loop takes a total of $O(n \log(n))$ operations.
Hence, Algorithm \ref{alg:3D-stair} has $O(n^2 \log(n))$ time complexity.

\paragraph{Algorithm for three-dimensional points using $O(n\log^2(n))$ space}
Given $S \subseteq \R^3$ as above, we present Algorithm \ref{alg:3D-priority-range} for finding the direct dominance pairs of $S$. 
This has $O((n+k') \log^2(n))$ time and $O(n\log^2(n))$ space complexities, where $k'$ is the number of direct dominance pairs and $n$ the number of points of $S$.
The above discussion for the algorithm using known direct dominance results applies here as well. 
So, Algorithm \ref{alg:3D-priority-range} can be used to obtain the Minibox edges of $S$ in $O(k\log^2(n))$ time and $O(n\log^2(n))$ space, where $k$ is the number of Minibox edges of $S$.
This improves the space complexity for average-case $k$, see Table \ref{table:complexities-algorithms-3D}.

The idea is to use range queries taking $O(\log^2(n))$ time to find points such that $\{p, q\}$ is a direct dominance pair.
This is achieved by querying a range tree with fractional cascading \cite[Section 5.6]{de2010computational}, and updating a dynamic priority search tree with the results of these queries \cite{mccreight1985priority-search-tree}. 

We prove that for any $p \in S$, Algorithm \ref{alg:3D-priority-range} correctly identifies each direct dominance pair $\{p, q\}$ with $q \in S$.
To begin with, we build a range tree $R$ with fractional cascading on the points of $S$, which is a three-level data structure.
The first level is a binary search tree sorted on $x$ coordinates, the second level contains binary search trees sorted on $y$ coordinates, and the third level arrays of points sorted on $z$ coordinates.
This uses $O(n\log^2(n))$ space, and reports $s$ points in a three-dimensional orthogonal range $[x_1, x_2] \times [y_1, y_2] \times [z_1, z_2]$ in $O(\log^2(n) + s)$ time.
We assume $R$ to be built as in \cite[Section 5.6]{de2010computational}, so that a three-dimensional range query returns $O(\log^2(n))$ third level arrays sorted on $z$ coordinates with pointers to their first elements with $z$ coordinate greater than $z_1$.
By reporting only the $O(\log^2(n))$ points corresponding to the pointed to elements of these sorted arrays, and finding the point with minimum $z$ coordinate among these with $O(\log^2(n))$ operations, we obtain the point $q$ with minimum $z$ coordinate among those in the original three-dimensional range.
We say that $q$ is the output of a $\min$-$z$ three-dimensional query on $R$.
In particular, we define a $\min$-$z$ query on $R$ to return either the point with minimum $z$ coordinate in a three-dimensional orthogonal range or $(+\infty, +\infty, +\infty)$ if the orthogonal range is empty. 
Note that a $\min$-$z$ query on $R$ takes $O(\log^2(n))$ time.

On line $7$ of Algorithm \ref{alg:3D-priority-range} the point $q$ directly above $p$ in the $z$ direction is found, and then used to initialize a priority search tree $P$.
In particular, $P$ is defined as a dynamic red-black tree sorted on $x$ coordinates, with the property of being also a $\min$-heap on $z$ coordinates.
Each node of $P$ contains the coordinates of a point $q$ and the $x$ and $y$ ranges of the three-dimensional orthogonal range containing $q$.
The Cartesian product of these ranges is a ``vertical" rectangle in the sweep-plane used in Algorithm \ref{alg:3D-stair}.
Hence, we can think of $P$ as a data structure keeping track of these ``vertical" rectangles.

Importantly, $P$ has the following property: the union of the ``vertical'' rectangles stored at the nodes of $P$ equals to the area under the staircase described by the points of $Q'_1$ used in Algorithm \ref{alg:3D-stair}, see Figure \ref{fig:algorithm-minibox-3D-2}.
This property holds true when $P$ is initialized, and it is preserved when updating $P$ on lines $12-16$ of Algorithm \ref{alg:3D-priority-range}, because at each iteration of the inner loop on lines $9-17$ the $z$ coordinate of $q$ is monotonically increasing.
The update searches $P$ for all rectangles to the right of $(q_x, q_y)$ such that $q_y$ is within their $y$ range, and deletes them from $P$.
This step corresponds to the deletion of nodes from $T_1$ in Algorithm \ref{alg:3D-stair}.
Then, two new nodes are added for the rectangles to the left and right of $(q_x, q_y)$, which corresponds to the insertion of $q_{xy}$ in $T_1$.
In particular, if $[x_1^q, x_2^q] \times [y_1^q, y_2^q]$ is the rectangle containing $(q_x, q_y)$, and $[\hat{x}_1, \hat{x}_2] \times [\hat{y}_1, \hat{y}_2]$ is the first rectangle to the right of $(q_x, q_y)$ which is not deleted in the previous step of the algorithm, then $[x_1^q, q_x] \times [y_1^q, y_2^q]$ and $[q_x, \hat{x}_1] \times [\hat{y}_1, q_y]$ respectively are these new left and right rectangles.
Moreover, the points $q_L$ and $q_R$,  obtained with $\min$-$z$ queries on $R$, are stored in the new nodes of $P$.
Note that if a three-dimensional orthogonal range is empty, then a $\min$-$z$ query on $R$ returns $(+\infty, +\infty, +\infty)$, and we say that the node in which this point is inserted is marked.

\begin{figure}[!tp]
    \begin{minipage}{\textwidth}
        \begin{algorithm}[H] 
    	    \begin{algorithmic}[1] 
	        \caption{Direct dominance pairs of $S$ in three-dimensions using a priority search tree and range tree.}
	        \label{alg:3D-priority-range}
	\STATE \textbf{input:} array $points$, the finite set of points $S$ in $(\mathbb{R}^3, \din)$
	\STATE $pairs \gets$ empty list of two-tuples of integers
	\STATE Sort $points$ on their $z$-coordinate
	\STATE $R \gets$ range tree with fractional cascading of $points$
	\FOR {$i=0$ to $|S| - 2$}
	    \STATE $p \gets points[i]$ 
	    \STATE $q \gets$ query $R$ for the point with minimum $z$ coordinate in range $[p_x, +\infty) \times [p_y, +\infty) \times [p_z, +\infty)$
	    \STATE $P \gets$ root of $\min$-$z$ priority search tree sorted on $x$ containing $(q_x, q_y, q_z, p_x, +\infty, p_y, +\infty)$  
	    \WHILE {root of $P$ not marked}
	        \STATE $q \gets (q_x, q_y, q_z)$, i.e. first three coordinates of point in root of $P$
	        \STATE Add $(i, j)$ in pairs, where $j$ is index of $q$ in $points$
	        \STATE Delete root from $P$ 
	        \STATE Delete nodes representing ``vertical'' rectangles $[x_1, x_2] \times [y_1, y_2]$ with $q_x < x_2$ and $q_y < y_2$
	        \STATE $q_L \gets$ query $R$ for point with minimum $z$ in ``vertical'' rectangle to the left of $(q_x, q_y)$
	        \STATE $q_R \gets$ query $R$ for point with minimum $z$ in ``vertical'' rectangle to the right of $(q_x, q_y)$
	        \STATE Insert nodes in $P$ corresponding to $q_L$, $q_R$. If $q_L$ or $q_R$ is $(+\infty, +\infty, +\infty)$, then node is marked
	    \ENDWHILE
	\ENDFOR
	\RETURN $edges$
        	\end{algorithmic}
        \end{algorithm}
    \end{minipage}
     \centering
    \begin{subfigure}[b]{2in}
        \centering
        \includegraphics[width=2in]{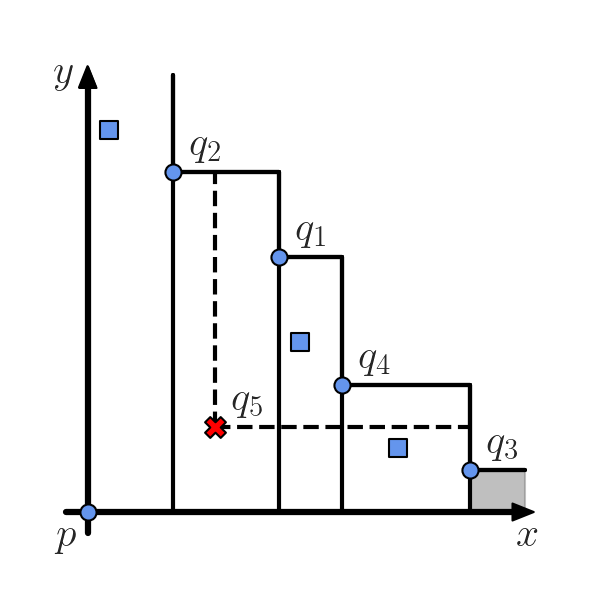}
        \caption{}
        \label{fig:staircase-3}
    \end{subfigure}
    \qquad
    \begin{subfigure}[b]{2in}
        \centering
        \includegraphics[width=2in]{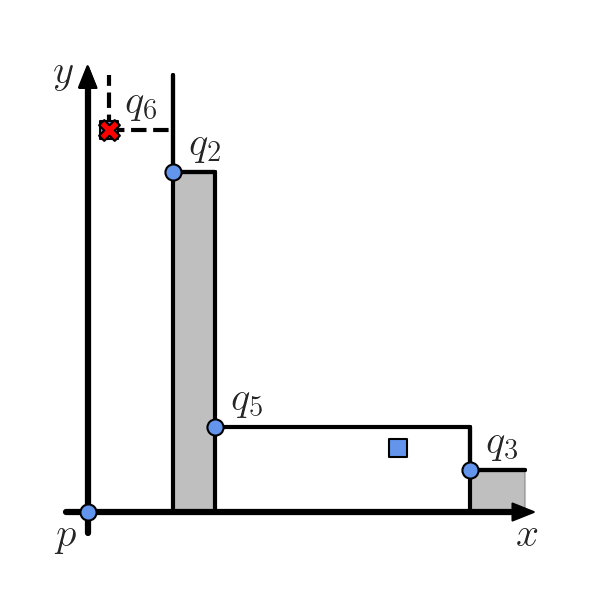}
        \caption{}
        \label{fig:staircase-4}
    \end{subfigure}
    \caption{Two iterations of the inner loop of Algorithm \ref{alg:3D-priority-range}. 
    The points stored in the first two coordinates of the nodes of the priority search tree $P$ are illustrated as squares.
    \textbf{(a)} The rectangular region containing $q_5$ corresponds to the root of $P$, so $\{p, q_5\}$ is a direct dominance pair. The nodes corresponding to the ``vertical" rectangular regions between $q_2$ and $q_1$, $q_1$ and $q_4$, $q_4$ and $q_3$ are deleted from $P$. Two new nodes are inserted in $P$, corresponding to the regions between $q_2$ and $q_5$, and $q_5$ and $q_3$. The first of these two nodes is marked, because the ``vertical" rectangle between $q_2$ and $q_5$ is empty.
    \textbf{(b)} The new root of $P$ contains $q_6$, so $\{p, q_6\}$ is a direct dominance pair.
    Two new nodes are inserted in $P$, and both are marked.
    }
    \label{fig:algorithm-minibox-3D-2}
\end{figure}
It follows that at each iteration the root of $P$ contains the point $q$ with minimum $z$ coordinate among the points whose projection lies in one of these ``vertical'' rectangles, by the $\min$-heap property of $P$.
Thus, $\{p, q\}$ is a direct dominance pair and can be reported.

To conclude, the inner loop on lines $9-17$ of Algorithm \ref{alg:3D-priority-range} finds the same direct dominance pairs of the inner loop of Algorithm \ref{alg:3D-stair}, because of the correspondence between ``vertical'' rectangles in $P$ and the area under the staircase on $Q'_1$.
Moreover, it correctly stops when all nodes are marked, and so all ``vertical" rectangles are empty, because if the root of $P$ has $q_z=+\infty$, then $q_z=+\infty$ for all the nodes in $P$ by its $\min$-heap property.


Note that for any $p \in S$ at most $2k''- 1$ nodes can be inserted in $P$ (this worst case is realised only if $Q'_{xy} = Q'_1$), where $k''$ is the number of direct dominance pairs of $S$ containing $p$.
Similarly, the number of delete operations on $P$ is $O(k'')$.
Besides, the number of $\min$-$z$ queries on $R$, of which there are two for each iteration of the inner loop, is also $O(k'')$.
Thus, Algorithm \ref{alg:3D-priority-range} has $O((n+k')\log^2(n))$ time complexity, where $k'$ is the number of direct dominance pairs of $S$, because the operations on $P$ take $O(\log(n))$ time and the $\min$-$z$ queries on $R$ take $O(\log^2(n))$ time.
Finally, the space complexity of Algorithm \ref{alg:3D-priority-range} is $O(n\log^2(n))$, because $P$ takes at most $O(n)$ space, and $R$ takes $O(n\log^2(n))$ space.

\paragraph{Points in higher dimensions}
For points in general dimension $d \geq 4$, we propose different strategies, using a decreasing amount of additional storage, to test whether $\minipq \cap S$ is empty for each pair of points in $S$.

For instance, high-dimensional range trees with fractional cascading \cite[Section 5.6]{de2010computational} can be used to answer orthogonal range emptiness queries in $O(\log^{d-1}(n))$ time, at the additional cost of $O(n \log^{d-1}(n))$ storage. 
By testing all $\frac{n(n-1)}{2}$ pairs of points in $S$, we have a $O(n^2 \log^{d-1}(n))$ algorithm.
Similarly, $kd$-trees \cite[Section 5.2]{de2010computational} can be used to answer the same query in $O(n^{1-\frac{1}{d}})$ time, only taking $O(n)$ additional storage, resulting in a $O(n^{3-\frac{1}{d}})$ algorithm for finding all the edges contained in any Minibox complex. 
Furthermore, we note that by the curse of dimensionality, if $d$ becomes too big it might be faster to test each of the $\frac{n(n-1)}{2}$ pairs of points in $S$ via a brute force strategy, searching all points in $S$ sequentially, which takes $O(n)$ time. This results in a $O(d n^3)$ total time algorithm, but does not require storing any additional data structure.
The choice among these options depends on the amount of memory that can be spared for storing additional data structures, as well as the dimension $d$.
Moreover, we note that each of the above strategies could take advantage of parallel implementations using the independence of tests on each pair of points in $S$.

Finally, we also mention that in the Word RAM model of computation the offline orthogonal range counting algorithm of \cite{chan2010counting} can be used to find all empty miniboxes on $S$ in constant dimension $d \geq 3$ in $O(n^2 \log^{d-2+\frac{1}{d}}(n))$. 
However, as remarked in \cite{chan2010counting}, for this algorithm to be applicable to floating-point numbers one needs to assume that the word size is at least as large as both $\log(n)$ and the maximum size of an input number.

\section{Computational Experiments}
\label{sec:experiments}
 Here we present various computational experiments involving $\linf$-Delaunay and Minibox edges, and the derived complexes.
First, we study the expected number of $\linf$-Delaunay and Minibox edges on randomly sampled points, as well as the size of Minibox filtrations.
We then investigate the speed up obtained by using Minibox filtrations in the calculation of \v{C}ech persistence diagrams in homological degrees zero and one.
Finally, we give examples illustrating the similarities and dissimilarities in homological degree two of persistence diagrams of Alpha flag, Minibox, and \v{C}ech filtrations.

In order to apply the algorithms presented in the previous section to a point set $S \subseteq (\R^d, \din)$, we have to assume that $S$ does not contain collinear points.
This can be enforced by applying a small perturbation to the coordinates of points of $S$, as described at the beginning of Section \ref{sec:algorithms}.
The persistence diagrams obtained from the infinitesimally perturbed points are also only infinitesimally perturbed. 
In practice, this simple strategy may not be sufficient if the number of points $n$ becomes too large, as the precision of floating-point numbers is limited. In general, the algorithm only requires that each coordinate induces a total ordering on the points, which can be done by arbitrarily breaking ties (e.g. by some  lexicographical ordering of the points). 
%
%
It should be noted that, the random perturbation method was sufficient to avoid collinear points for the experiments presented in this section, and was thus preferred to more complex approaches.
Besides, no other general position assumption needs to be imposed on the points of $S$. 
This is a consequence of the fact that the Nerve Theorem, Theorem \ref{thm:main}, and Theorem \ref{thm:alpha-flag-and-minibox} do not require $S$ to be in general position.

All computations were run on a laptop with Intel Core i7-9750H CPU with six physical cores clocked at 2.60GHz with 16GB of RAM.

\paragraph{Expected number of edges and size of filtrations}
First, we give plots of the expected numbers of $\linf$-Delaunay and Minibox edges.
Secondly, we study the expected size of Minibox filtrations versus the size of \v{C}ech filtrations.

To begin with, Figure \ref{fig:compare-number-minibox-delaunay-edges} illustrates the difference in expected numbers of $\linf$-Delaunay and Minibox edges for points in $\R^2$.
These are empirical estimates, which were obtained by averaging over the numbers of edges found on five different randomly sampled point sets.
Moreover, the function $f(n) = 2 n \ln(n)$ is plotted along with these estimates, which shows how the bound of Proposition \ref{prop:expected-edges-minibox} compares to the numbers of expected Minibox edges in practice. %
Then, Figure \ref{fig:number-minibox-edges} shows estimates of expected numbers of Minibox edges in different dimensions.

These were obtained with the same method as before.
It should be noted that, computations for $\linf$-Delaunay edges are limited to two-dimensional points, because this is the only setting where efficient algorithms are available for their computation \cite{shute1991planesweep}.

\begin{figure}[!tp]
    \centering
        \begin{subfigure}[b]{2.5in}
        \centering
        \includegraphics[width=2.5in]{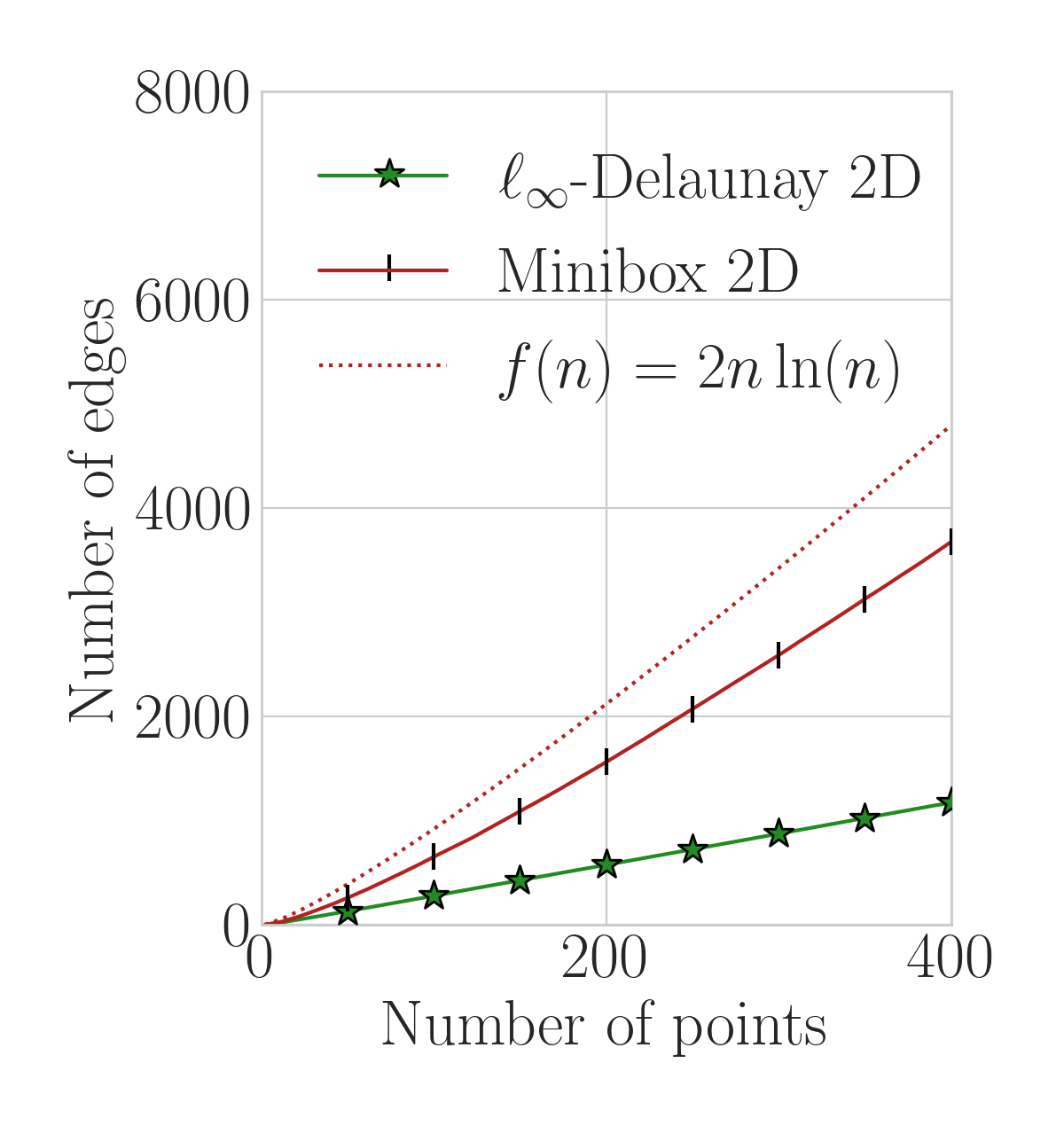}
        \caption{}
        \label{fig:compare-number-minibox-delaunay-edges}
    \end{subfigure}
    \qquad
    \begin{subfigure}[b]{2.5in}
        \centering
        \includegraphics[width=2.5in]{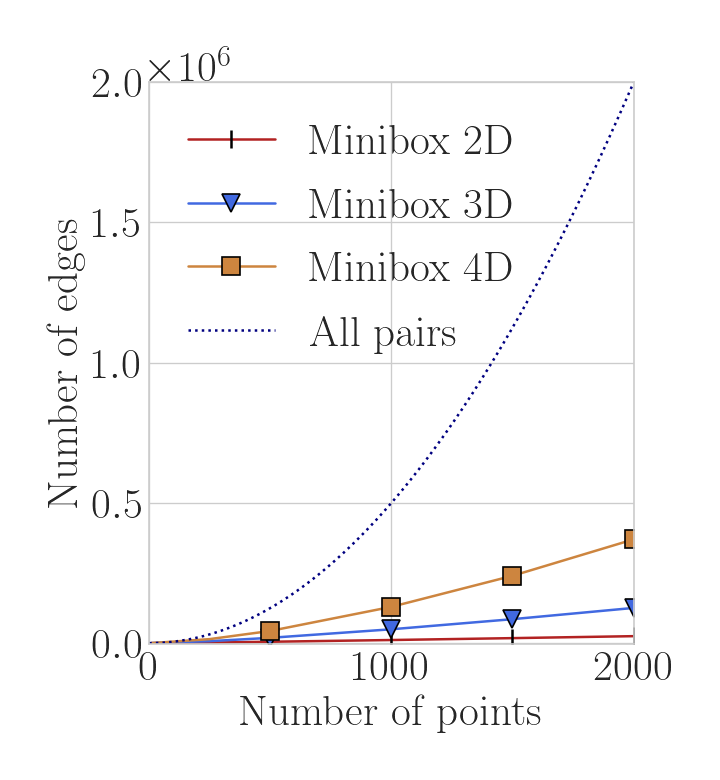}
        \caption{}
        \label{fig:number-minibox-edges}
    \end{subfigure}
    \caption{\textbf{(a)} Empirical estimates of expected numbers of $\linf$-Delaunay and Minibox edges for points in $\R^2$. These are compared to the upper bound of the expected number of Minibox edges. 
    \textbf{(b)} Empirical estimates of the expected numbers of Minibox edges for dimension up to $d=4$. These are compared to the number of \v{C}ech edges, i.e. all pairs of points.}
    \label{fig:expected-number-edges}
\end{figure}
%
\begin{table}[!tp]
 \caption{Average number of simplices contained in the Minibox and \v{C}ech filtrations for different input sizes.}
 \label{table:size-minibox-filtration}
  \centering
  \vspace{0.3cm}
  \begin{tabular}{| l | c | c | c | c |}
  \cline{2-5}
      \multicolumn{1}{ l| }{}
      & n = 500 & n = 1000 & n = 1500 & n = 2000  \\
  \hline
  \rule{0pt}{11pt}
    Minibox 2D & 
      $0.01\times 10^6$ & $0.03\times 10^6$  & $0.05\times 10^6$  & $0.07\times 10^6$   \\
  \hline
  \rule{0pt}{11pt}
    Minibox 3D  & 
      $0.17\times 10^6$ & $0.50\times 10^6$  & $0.91\times 10^6$  & $1.38\times 10^6$   \\
  \hline
  \rule{0pt}{11pt}
    Minibox 4D  & 
      $1.19\times 10^6$ & $4.50\times 10^6$  & $9.41\times 10^6$  & $15.65\times 10^6$   \\
  \hline
  \rule{0pt}{11pt}
    \v{C}ech  & 
      $20.83\times 10^6$ & $166.67\times 10^6$  & $562.50\times 10^6$  & $1333.34\times 10^6$   \\
  \hline
  \end{tabular}
\end{table}
Next, we investigate the expected number of simplices contained in Minibox filtrations.
Note that our filtrations only contain vertices, edges, and triangles, because we compute persistence diagrams in homological degrees zero and one.
Thus, \v{C}ech filtrations contain $\Theta(n^3)$ simplices.
In comparison, given the edges in the maximal Minibox complex of $S$, the clique triangles on these can be found in $O(n k^2)$ time, where $k$ is the maximum degree of any point in $S$, i.e. the maximum number of Minibox edges a point is contained in.
Moreover, $O(n k^2)$ is also an upper bound on the number of possible Minibox triangles, and by Proposition \ref{prop:expected-edges-minibox} it follows that the expected value of $k$ for a uniformly distributed finite set of random points is $O\big(2^{d-1} \ln^{d-1}(n) \big)$.
Hence, we expect the Minibox filtration of $S$ to contain less simplices compared to the \v{C}ech filtration.
We give empirical evidence of this by calculating the expected number of Minibox simplices for $500$, $1000$, $1500$, and $2000$ uniformly distributed random points, averaging over five runs.
Table \ref{table:size-minibox-filtration} presents our results for Minibox filtrations in two, three and four dimensions. The number of simplices contained in the corresponding \v{C}ech filtrations are listed for comparison.

\begin{table}[!tp]
 \caption{Timing (seconds) and memory usage (MB) with Minibox filtrations of points in $\R^2$.}
 \label{table:mini-2d-timing}
  \centering
  \begin{tabular}{| l | c | c | c | c | c | c| c |}
    \cline{2-8}
        \multicolumn{1}{ l| }{}
         & n = 500 & n = 1000 & n = 2000
        & n = 4000 & n = 8000 & n = 16000 & n = 32000
    \\
    \hline
    Edges time
        & 0.008 & 0.016 & 0.047
        & 0.117 & 0.289 & 0.891 & 2.852 \\
    \hline
    Sparse matrix time
        & 0.023 & 0.070 & 0.141
        & 0.312 & 0.734 & 1.562 & 3.406 \\
    \hline
    Dgm$_{0,1}$ time
        & 0.008 & 0.016 & 0.031
        & 0.078 & 0.172 & 0.477 & 1.148 \\
    \hline
    \hline
    Total time 
        & 0.039 & 0.102 & 0.219
        &  0.507 & 1.195 & 2.929 & 7.406 \\
    \hline
    \hline
    Peak memory usage 
        & 2.92 & 5.52 & 11.51
        & 25.15 & 53.50 & 112.07 &  246.28 \\
    \hline
    \end{tabular}
    \bigskip
    \caption{Timing (seconds) and memory usage (MB) with Minibox filtrations of points in $\R^3$.}
 \label{table:mini-3d-timing}
  \centering
  \begin{tabular}{|l | c | c | c | c | c | c| c |}
    \cline{2-8}
        \multicolumn{1}{ l| }{}
         & n = 500 & n = 1000 & n = 2000
        & n = 4000 & n = 8000 & n = 16000 & n = 32000
    \\
    \hline
    Edges time
        & 0.062 & 0.188 & 0.586
        & 2.047 & 7.500 & 27.898 & 110.641 \\
    \hline
    Sparse matrix time
        & 0.117 & 0.281 & 0.742
        & 1.836 & 4.609 & 11.289 & 26.555 \\
    \hline
    Dgm$_{0,1}$ time
        & 0.016 & 0.055 & 0.211
        & 0.547 & 1.664 & 4.516 & 12.336 \\
    \hline
    \hline
    Total time 
        & 0.195 & 0.523 & 1.539
        & 4.429 & 13.773 & 43.703 & 149.531 \\
    \hline
    \hline
    Peak memory usage 
        & 9.22 & 21.87 & 54.91
        & 137.25 & 329.25 & 770.80 & 1848.01 \\
    \hline
    \end{tabular}
    \bigskip
    \caption{Timing (seconds) and memory usage (MB) with Minibox filtrations of points in $\R^4$.}
 \label{table:mini-4d-timing}
  \centering
  \begin{tabular}{|l | c | c | c | c | c | c| c |}
    \cline{2-8}
        \multicolumn{1}{ l| }{}
         & n = 500 & n = 1000 & n = 2000
        & n = 4000 & n = 8000 & n = 16000 & n = 32000
    \\
    \hline
    Edges time
        & 0.273 & 1.648 & 9.430
        & 54.164 & 307.078 & 1657.852 & 8866.555 \\
    \hline
    Sparse matrix time
        & 0.258 & 0.727 & 2.055
        & 6.250 & 15.680 & 43.516 & 107.773 \\
    \hline
    Dgm$_{0,1}$ time
        & 0.070 & 0.227 & 0.797
        & 2.539 & 9.320 & 27.016 & 107.273 \\
    \hline
    \hline
    Total time 
        & 0.601 & 2.601 & 12.281
        & 62.953 & 332.078 & 1728.383 & 9081.601 \\
    \hline
    \hline
    Peak memory usage 
        & 19.194 & 51.18 & 155.44
        & 410.41 & 1122.84 & 2841.05 & 7960.18 \\
    \hline
    \end{tabular}
    \bigskip
     \caption{Timing (seconds) and memory usage (MB) with \v{C}ech filtrations of points in $\R^2$.}
    \label{table:cech-2d-timing}
  \centering
  \begin{tabular}{|l | c | c | c | c | c |}
    \cline{2-6}
        \multicolumn{1}{ l| }{}
         & n = 500 & n = 1000 & n = 2000
        & n = 4000 & n = 8000
    \\
    \hline
    Sparse matrix time
        & 0.656 & 2.758 & 11.047
        & 44.789 & 178.727  \\
    \hline
    Dgm$_{0,1}$ time
        & 0.133 & 0.602 & 2.958
        & 13.312 & 66.219  \\
    \hline
    \hline
    Total time 
        & 0.789 & 3.359 & 14.005
        & 58.101 & 244.945  \\
    \hline
    \hline
    Peak memory usage 
        & 42.05 & 151.14 & 614.13
        & 2532.38 & 10340.73  \\
    \hline
    \end{tabular}
\end{table}
\paragraph{Running Time and Memory Usage}
We explore the use of Minibox filtrations for the computation of \v{C}ech persistence diagrams of $S \subseteq (\R^d, \din)$ in homological degrees zero and one.
We list our results in Tables \ref{table:mini-2d-timing}, \ref{table:mini-3d-timing}, 
\ref{table:mini-4d-timing}, and \ref{table:cech-2d-timing}, where columns correspond to different sizes of the input points set $S$, and times are given in seconds.
We also report the average total peak memory use in megabytes.\footnote{In Windows this was measured using the Win32 function \texttt{GetProcessMemoryInfo()} to obtain the \texttt{PeakWorkingSetSize} memory attribute of the Python process building sparse matrices and computing persistence diagrams.}
It should be noted that these results were obtained with the implementations of the Minibox edge algorithms provided by the \href{https://pypi.org/project/persty/}{\texttt{persty}} Python packages.
For the computation of persistence diagrams we use of the \href{https://ripser.scikit-tda.org/}{Ripser.py} package, which provides a Python interface to \href{https://github.com/ripser/ripser}{Ripser} \cite{bauer2019ripser} C++ code. 
\begin{figure}[tb]
    \centering
        \begin{subfigure}[b]{2in}
            \centering
            \includegraphics[width=2in]{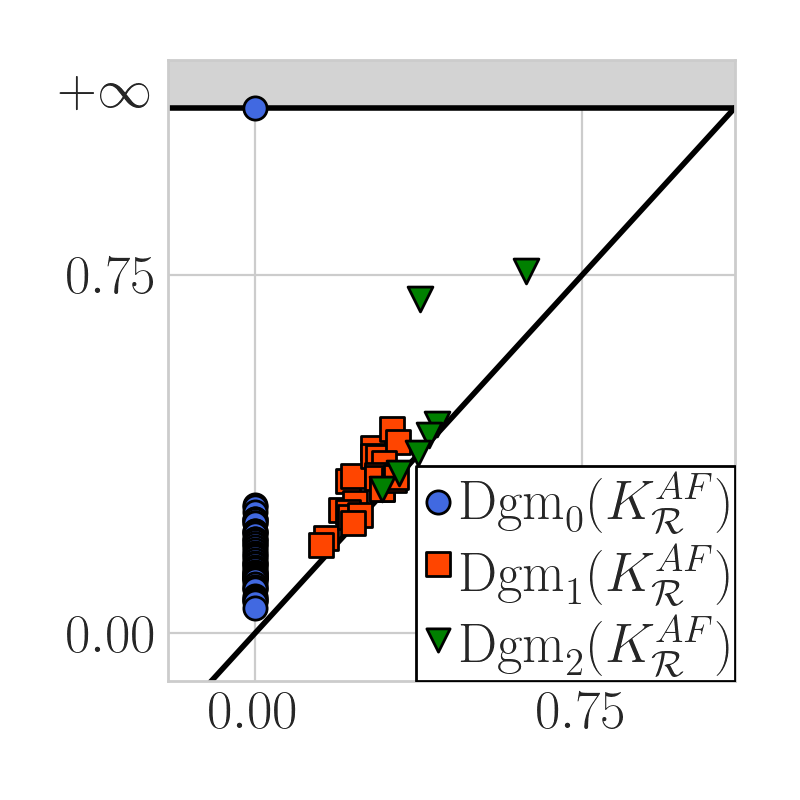}
            \caption{Alpha flag diagrams of $S_1$.}
            \label{fig:dgms-AC1}
        \end{subfigure}
        \begin{subfigure}[b]{2in}
            \centering
            \includegraphics[width=2in]{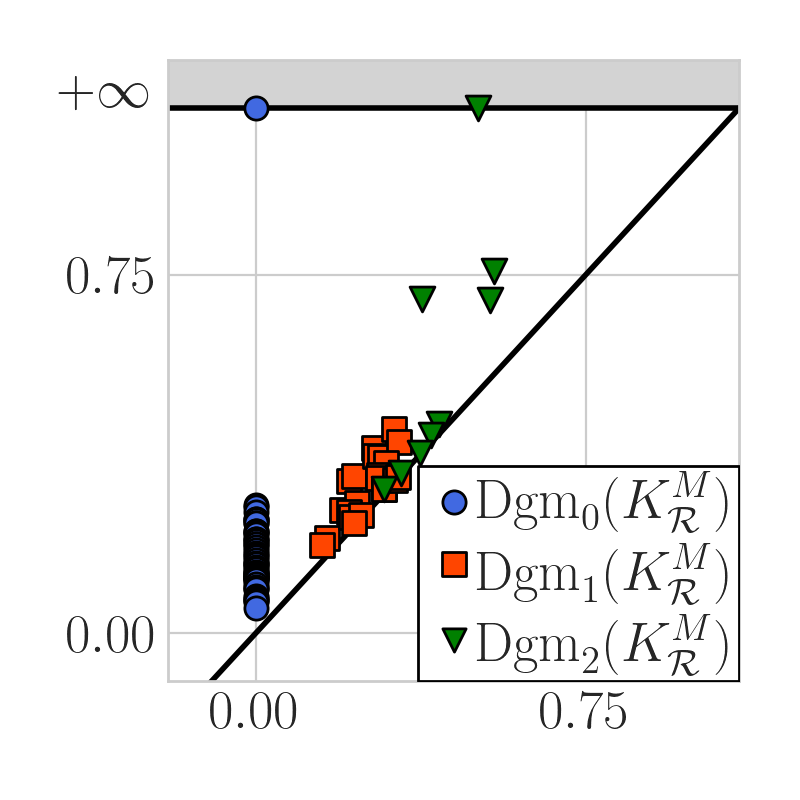}
            \caption{Minibox diagrams of $S_1$.}
            \label{fig:dgms-M1}
        \end{subfigure}
        \begin{subfigure}[b]{2in}
            \centering
            \includegraphics[width=2in]{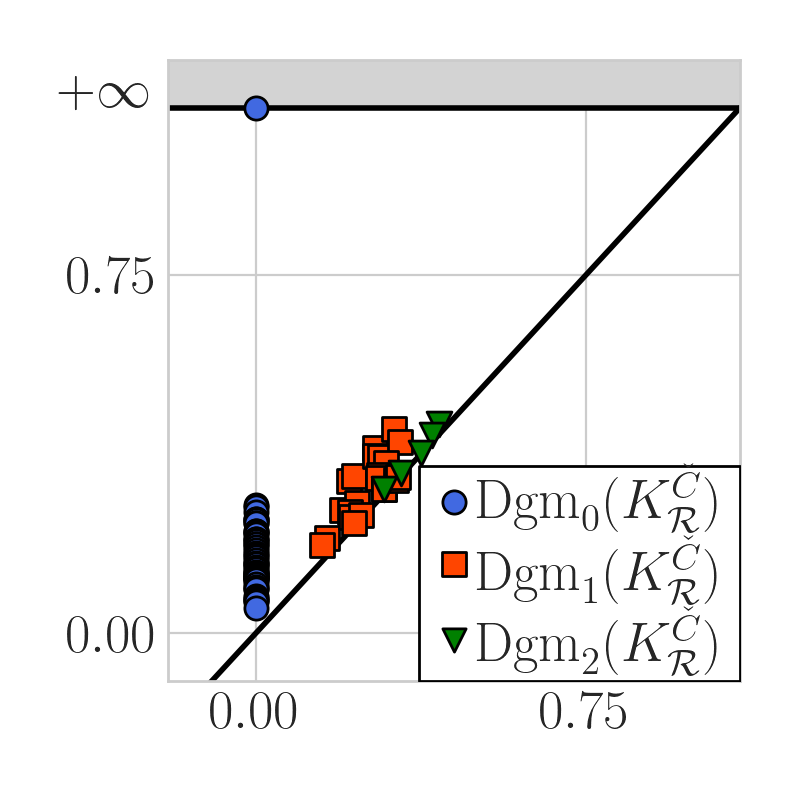}
            \caption{\v{C}ech diagrams of $S_1$.}
            \label{fig:dgms-C1}
        \end{subfigure}
        \\
        \begin{subfigure}[b]{2in}
            \centering
            \includegraphics[width=2in]{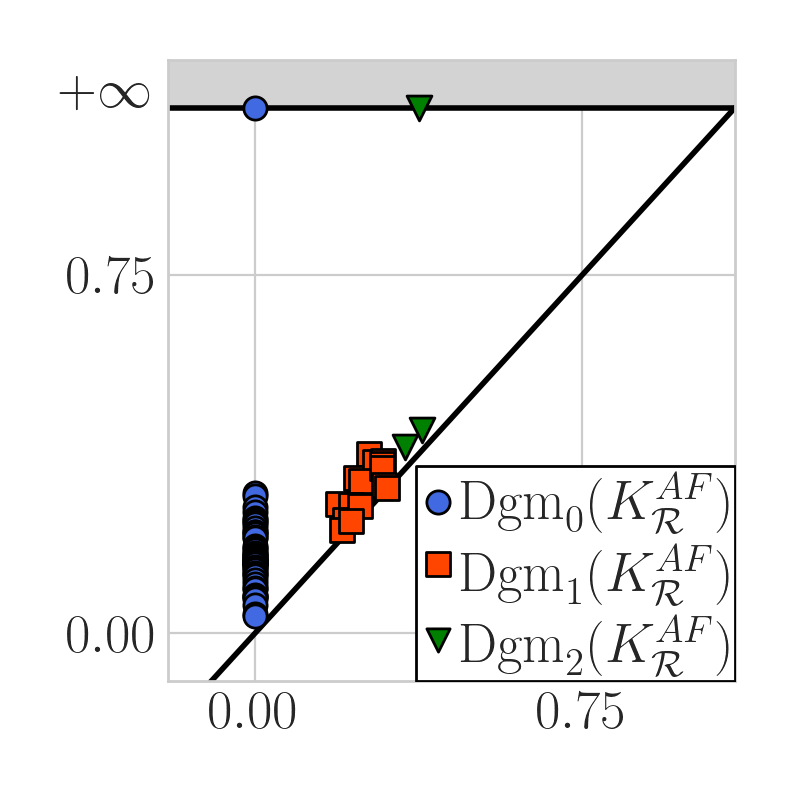}
            \caption{Alpha flag diagrams of $S_2$.}
            \label{fig:dgms-AC2}
        \end{subfigure}
        \begin{subfigure}[b]{2in}
            \centering
            \includegraphics[width=2in]{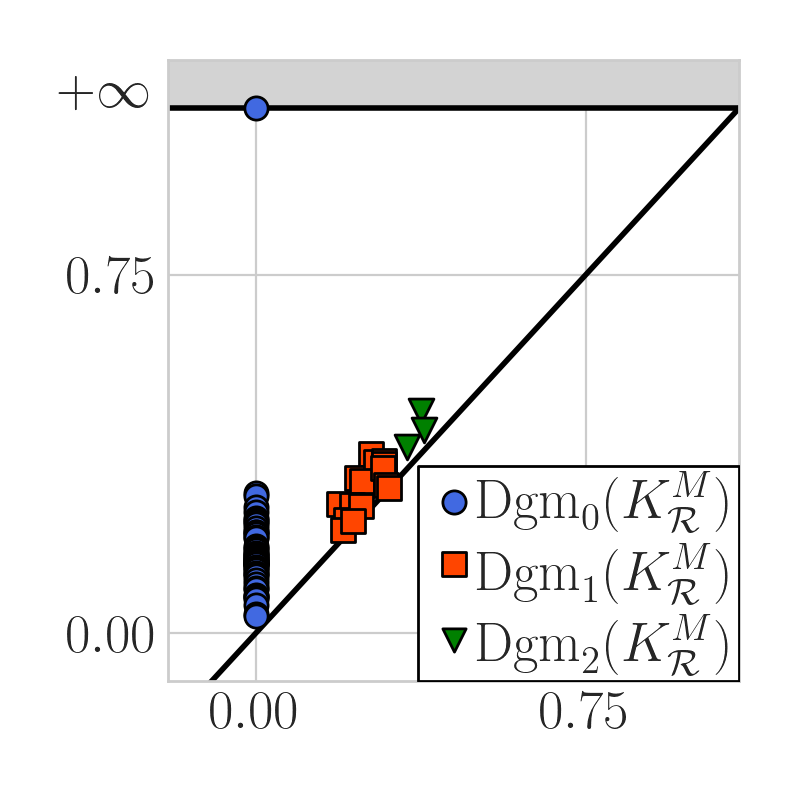}
            \caption{Minibox diagrams of $S_2$.}
            \label{fig:dgms-M2}
        \end{subfigure}
        \begin{subfigure}[b]{2in}
            \centering
            \includegraphics[width=2in]{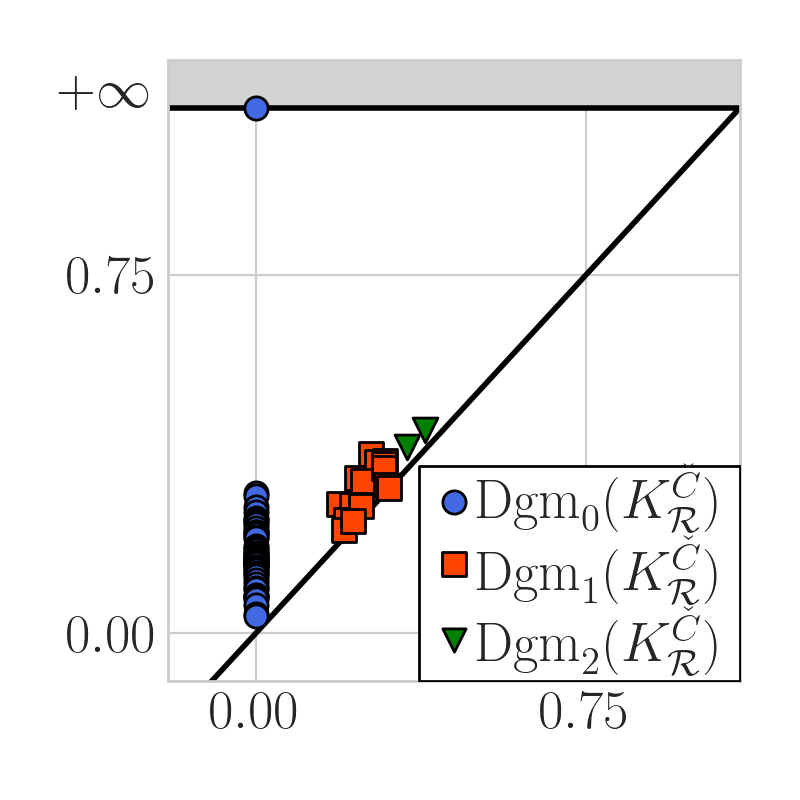}
            \caption{\v{C}ech diagrams of $S_2$.}
            \label{fig:dgms-C2}
        \end{subfigure}
    \caption{Persistence diagrams of finite sets of three-dimensional points in $\ell_{\infty}$ metric space. Each row contains the diagrams of a different finite point set. These empirically show the equality of diagrams in degrees zero and one, and illustrate the possible differences between diagrams of Alpha flag, Minibox, and \v{C}ech filtrations in homological degree two.}
    \label{fig:diagrams}
\end{figure}
In particular, we think of Minibox filtrations as of sparse filtrations, and feed into the persistent homology algorithm a precomputed sparse matrix in coordinate format.
Moreover, the same approach is used to compute \v{C}ech persistence diagrams, because \v{C}ech and Vietoris-Rips filtrations coincide by Proposition \ref{prop:equality-cech-and-rips}.
We give timing and memory usage results for points in the range $[500, 32000]$ for Minibox filtrations, averaging over five runs.
In the case of \v{C}ech filtrations, we limit our experiments to a maximum of $8000$ points, because of memory constraints.
Moreover, we consider only points in $\R^2$, as results are similar in higher dimensions.

In all the experiments, the reduced number of simplices of Minibox filtrations results in a substantial improvement in memory usage over \v{C}ech filtrations, and in a speed up in the computation of Dgm$_{0}$ and Dgm$_1$.
This allows to increase the maximum size of inputs of the persistence algorithm, given a fixed amount of available memory. The price is having to precompute Minibox edges. 
We note that this computation could also take advantage of implementations parallelizing the inner loops of the algorithms of Section \ref{sec:algorithms}.

Finally, it is worth mentioning that we would expect an even greater improvement in run time and memory usage if Alpha flag filtrations were used in place of Minibox filtrations.
We do not present such experiments here, because we only described efficient algorithms for the computation of Minibox edges in Section \ref{sec:algorithms}.

\paragraph{Differences of persistence diagrams in homological degree two}
We present two examples of Alpha flag, Minibox, and \v{C}ech persistence diagrams, obtained from distinct $S_1, S_2 \subseteq (\mathbb{R}^d, \din)$.
These finite point sets were obtained by randomly sampling fifty points in $[0, 1]^3 \subseteq \R^3$.
The persistence diagrams were calculated with \href{https://ripser.scikit-tda.org/}{Ripser.py} passing in the appropriate space matrix.
For the Alpha flag case the edges belonging to the $\linf$-Delaunay complex of $S_1$ and $S_2$ were computed with a brute force strategy using the result of Proposition \ref{prop:delaunay-edge}, i.e. checking if $\A^{\bar{r}}$ is covered by $\bigcup_{y\in S \setminus e} B_{\bar{r}}(y)$ for each pair $p, q \in S$.

The first row in Figure \ref{fig:diagrams} contains the diagrams of $S_1$. In this case $\textrm{Dgm}_2(\Fmini)$ contains a point at infinity, while $\textrm{Dgm}_2(\Fflag)$ does not. Furthermore, both contain additional off-diagonal points, which do not coincide.
In the second row of Figure \ref{fig:diagrams}, we have the diagrams of $S_2$.
In this case it is $\textrm{Dgm}_2(\Fflag)$ that contains a point at infinity, while $\textrm{Dgm}_2(\Fmini)$ only has an additional off-diagonal point.
This shows that it is possible to obtain Alpha flag and Minibox diagrams with off-diagonal points not contained in the corresponding \v{C}ech diagrams in homological degrees higher than one. 
Furthermore, $\textrm{Dgm}_2(\Fflag)$ and $\textrm{Dgm}_2(\Fmini)$ are generally different, and are not one a subset of the other.

\section{Discussion}
\label{sec:discussion}
In this paper we prove that Alpha and \v{C}ech filtrations are equivalent for point sets in $(\R^2, \din)$, and show a counterexample to this equivalence for three-dimensional point sets.
We also introduce two new types of proximity filtrations: the Alpha flag and Minibox filtrations.
We are able to prove that both of these produce the same persistence diagrams of \v{C}ech filtrations in homological degrees zero and one.
Furthermore, we describe algorithms for finding Minibox edges.
In particular, we give two new algorithms for the three-dimensional case.
These improve over known results for finding direct dominance pairs.
We present a $O(n\log^2(n))$ time and $O(n)$ space algorithm, and a $O(k\log^2(n))$ time and $O(n\log^2(n))$ space algorithm, where $k$ is the number of Minibox edges of $S$.
In two dimensions, this reduces to rectangular visibility in the plane, whereas higher dimensions require different  structures for range queries.
Then, we prove that for randomly sampled points the expected number of Minibox edges is proportional to $n \cdot \textrm{polylog}(n)$. For the Euclidean metric, the expected size of a  Delaunay triangulation of random points is known to contain a linear number of simplices (in the number of vertices) \cite{dwyer1991higher}. We believe this should also be the case for $(\R^d, \din)$, but we plan to address this in future work. 
Therefore, in many cases Minibox filtrations can be seen as a tool to drastically reduce the number of simplices to be considered in order to compute \v{C}ech persistence diagrams in homological degrees zero and one.
We also provide a number of computational experiments involving Minibox and \v{C}ech filtrations of randomly sampled points in two, three, and four-dimensional space.
These show that the reduced number of simplices contained in Minibox filtrations results in a speed up of persistent homology computations, as well as in less memory being used for the same number of points.

\begin{figure}[!tp]
    \centering
        \begin{subfigure}[b]{2in}
            \centering
            \includegraphics[width=2in]{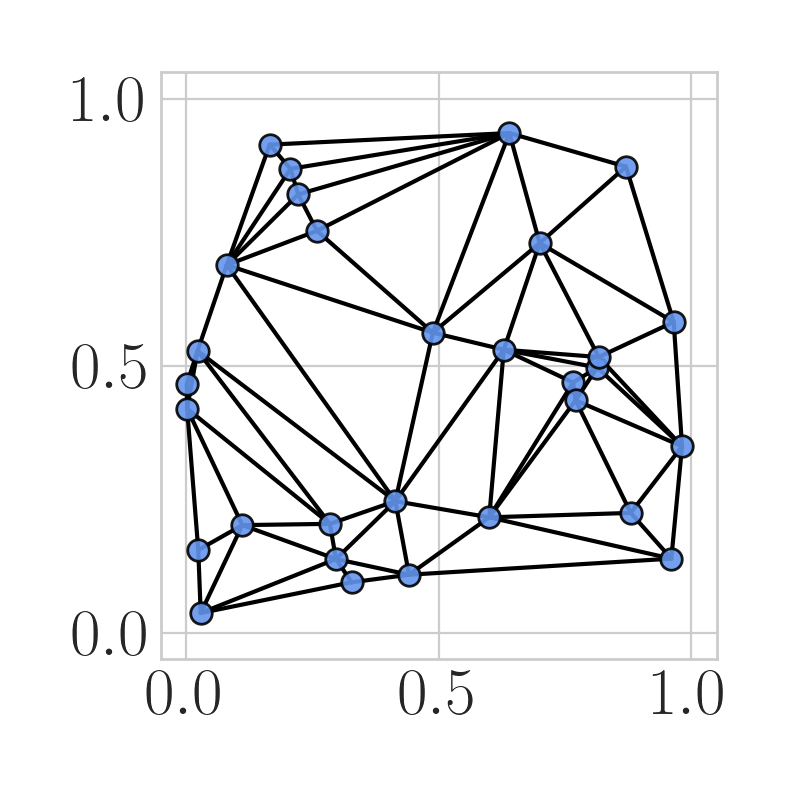}
            \caption{Alpha flag edges.}
            \label{fig:basic-delaunay-edges}
        \end{subfigure}
        \begin{subfigure}[b]{2in}
            \centering
            \includegraphics[width=2in]{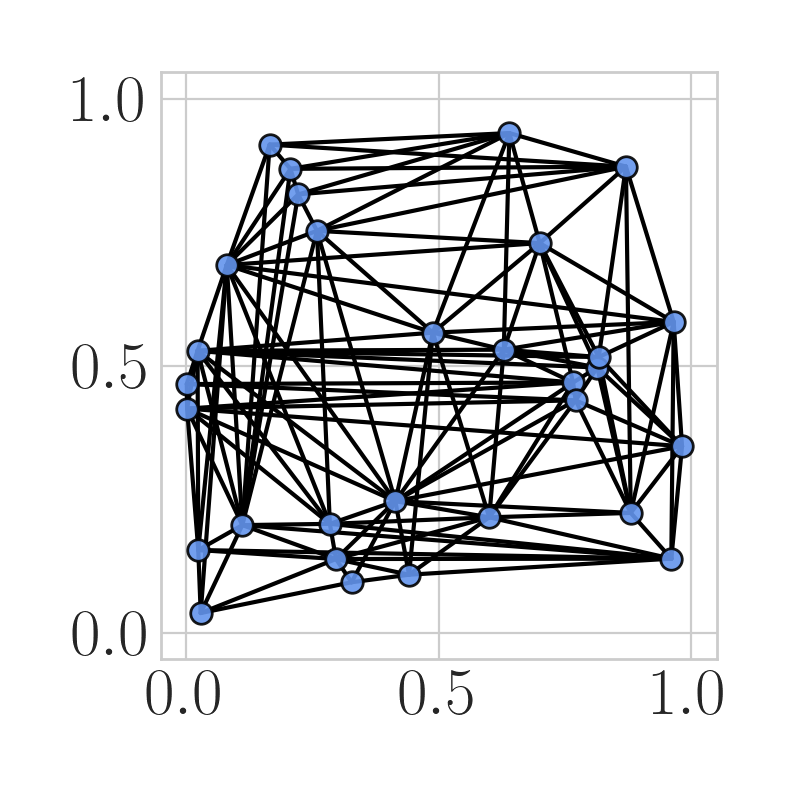}
            \caption{Minibox edges.}
            \label{fig:basic-minibox-edges}
        \end{subfigure}
        \begin{subfigure}[b]{2in}
            \centering
            \includegraphics[width=2in]{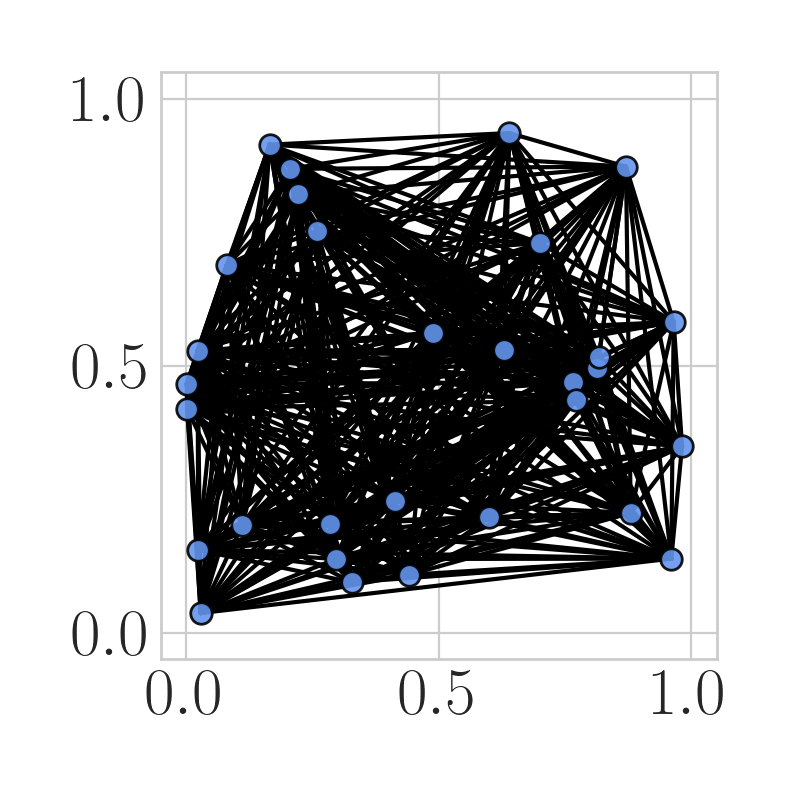}
            \caption{\v{C}ech edges. }
            \label{fig:basic-cech-edges}
        \end{subfigure}
    \caption{Comparison of Alpha flag (i.e. $\linf$-Delaunay), Minibox, and \v{C}ech edges of random points in $\R^2$.}
    \label{fig:trade-off-edges}
\end{figure}
We observe that the trade-off between Alpha flag, Minibox, and \v{C}ech complexes, for the computation of persistence diagrams in homological degrees zero and one, depends on the time complexity of algorithms for finding their edges, as well as the expected number of edges and triangles these complexes contain.
It should be noted that \v{C}ech complexes only require to list their $\binom{n}{2}$ edges and $\binom{n}{3}$ triangles, which takes $\Theta(n^3)$.
In comparison, we have that for random points the expected number of Minibox edges is $\boundMiniEdges$ for any dimension $d$ by Proposition \ref{prop:expected-edges-minibox}.
So in many settings, we expect the Minibox filtration to contain less simplices than  \v{C}ech filtrations.
Moreover,    we present Minibox edges algorithms taking $O(n\log(n) + k)$ time and $O(k\log^2(n))$ time for point sets in $\R^2$ and $\R^3$ respectively, where $k$ is the number of edges reported.
In general dimension $d$, a brute force algorithm for finding Minibox edges  has $O(d n^3)$ time complexity, but could take advantage of parallel implementations.
Additionally, Alpha flag complexes are subcomplexes on Minibox complexes by Proposition \ref{prop:minibox-property}, see Figure \ref{fig:trade-off-edges}.
It follows that Alpha flag filtrations further reduce the number of simplices to be considered for the computation of \v{C}ech persistence diagrams.
Unfortunately, an efficient algorithm for $\linf$-Delaunay edges is known only in $\R^2$ \cite{shute1991planesweep}.
For higher-dimension $d$, there exists an algorithm for constructing the $\linf$-Voronoi diagram, taking $O\big(n^{\lceil \frac{d}{2} \rceil} \log^{d-1}(n)\big)$ randomized expected time \cite{boissonnat1998voronoi}.
Hence, a direction of future work could be the study of efficient $\linf$-Delaunay edges algorithms above dimension two.

Finally, one issue with both Minibox and Alpha flag complexes is that they can only be used to compute persistence diagrams in homological degrees zero and one.
In particular, examples can be found where some of the points in $\textrm{Dgm}_2(\Fflag)$ and $\textrm{Dgm}_2(\Fmini)$ do not correspond to points in $\textrm{Dgm}_2(\Fcech)$. However, more persistent features seem to be captured by the complex, i.e. the errors consist of spurious rather than missing homological features.
Future research could focus on characterizing these and potentially introducing filtering steps, as well as investigating whether there exist alternative families of simplicial complexes for the computation of \v{C}ech persistence diagrams in homological degree two or higher.

\bibliographystyle{elsarticle-num-names}
\bibliography{references.bib}


\clearpage

\appendix
\setcounter{section}{0}
\renewcommand*{\thesection}{\Alph{section}}
   
\section{\texorpdfstring{$\linf$-Delaunay Edges}{l∞-Delaunay Edges}}
\label{sec-app:delaunay-edges}
In this section, we provide a characterization of $\linf$-Delaunay edges of a finite set of points $S \subseteq (\R^d, \din)$.

In Section \ref{sec:pre}, a box is defined as an axis-parallel hyperrectangle, i.e. the Cartesian product of $d$ intervals in $\R^d$.
Moreover, $\linf$-balls are boxes with sizes of length $2r$, and a finite set of $\linf$-balls has a non-empty intersection if and only if all pairwise intersections of $\linf$-balls are non-empty by Proposition \ref{prop:intersection-boxes} \emph{(ii)}.

We start by recalling properties of $\varepsilon$-thickenings.
%
\begin{proposition}
\label{app-prop:properties-epsilon-thickening}
\begin{itemize}
    \item[\emph{(i)}] Let $B_1, B_2 \subseteq \mathbb{R}$ be two non-empty boxes. If $B_1\cap B_2 \neq \emptyset$, then $\varepsilon(B_1 \cap B_2) = \varepsilon(B_1) \cap \varepsilon(B_2)$.
    \item[\emph{(ii)}] Taking $\varepsilon$-thickenings preserves inclusions.
    \item[\emph{(iii)}] Let $\mathcal{A} = \{ A \}_{i\in I}$ be a finite collection of sets.
    The $\varepsilon$-thickening of the union of sets in $\mathcal{A}$ is equal to the union of the $\varepsilon$-thickenings of sets in $\mathcal{A}$.
\end{itemize}
\end{proposition}
Next, we recall the definition of witness point given in Section \ref{sec:alpha}.
The idea is to restrict the bisector of a $\linf$-Delaunay simplex $\sigma$ to the points at minimal distance from the vertices of $\sigma$.
We use the properties of $\varepsilon$-thickenings and boxes to show that the set of witness points of a pair $\{p, q\} \subseteq S$ can be used to determine whether or not the pair is a $\linf$-Delaunay edge.
\begin{definition}
\label{app-def:witness-points}
Let $S$ be a finite set of points in $(\mathbb{R}^d, \din)$.
A \emph{witness} point of $\sigma \subseteq S$ is a point $z$ such that $z \in \bisector{\sigma} = \bigcap_{p \in \sigma} V_p$ and $\din(z,p) = \frac{\diamInfty{\sigma}}{2}$ for each $p \in \sigma$.
We write $\Z{\sigma}$ for the \emph{set of witness points} of $\sigma$.
\end{definition}
The statement of the following result is given as Proposition \ref{prop:delaunay-edge} in the main paper.
%
\begin{proposition}
\label{app-prop:delaunay-edge}
Let $S$ be a finite set of points in $(\R^d, \din)$.
Given a subset $e = \{p, q\} \subseteq S$, we define $\A^{\bar{r}} 
= 
\partial \cB{\bar{r}}{p} \cap \partial \cB{\bar{r}}{q}$, where $\bar{r} = \frac{\din(p,q)}{2}$.
We have that $\A^{\bar{r}} = \cB{\bar{r}}{p} \cap \cB{\bar{r}}{q}$ is a non-empty degenerate closed box.
Moreover, the set of witness points of $e$ is
$\Z{e} = \A^{\bar{r}} \setminus \big(\bigcup_{y\in S\setminus e} B_{\bar{r}}(y)\big)$, and $e$ belongs to the $\linf$-Delaunay complex of $S$ if and only if $\Z{e}$ is non-empty.
\end{proposition}
%
\begin{proof}
The intersection of boundaries $\A^{\bar{r}} = \partial \cB{\bar{r}}{p}  \cap \partial \cB{\bar{r}}{q}$ is included in $\cB{\bar{r}}{p} \cap \cB{\bar{r}}{q}$.
Moreover, any point $y \in \cB{\bar{r}}{p} \cap \cB{\bar{r}}{q}$ has to be at distance $\bar{r}$ from both $p$ and $q$.
Otherwise, if $\din(y, p) < \bar{r}$ or $\din(y, q) < \bar{r}$, then we obtain a contradiction with the triangular inequality, i.e. $\bar{r} + \bar{r} > \din(y, p) + \din(y, q) \geq \din(p, q) = 2 \bar{r}$.
Thus $\A^{\bar{r}} = \cB{\bar{r}}{p} \cap \cB{\bar{r}}{q}$, which is a non-empty degenerate box by definition of $\bar{r}$ and Proposition \ref{prop:intersection-boxes} \emph{(i)}.

Next, we show that $e$ is a $\linf$-Delaunay edge if and only if $\Z{e}$ is non-empty.
First, note that $e = \{p, q\}$ is a $\linf$-Delaunay edge if and only if there exists $z \in V_p \cap V_q$.
This $z$ does not have to be a witness point. 
In particular, $z$ has to be at the same distance $\bar{r} + \varepsilon$ from $p$ and $q$ for some real value $\varepsilon \geq 0$.
Equivalently, given $\A^{\bar{r}+\varepsilon} = \partial \cB{\bar{r} + \varepsilon}{p} \cap \partial \cB{\bar{r} + \varepsilon}{q}$, we have $z\in \A^{\bar{r}+\varepsilon} \setminus \big( \bigcup_{y\in S\setminus e} B_{\bar{r}+\varepsilon}(y) \big)$,
because by definition of $\linf$-Voronoi region $z \in V_p \cap V_q$ cannot be at distance strictly less than $\bar{r}$ from points in $S \setminus e$.
Furthermore, $z \in \A^{\bar{r}+\varepsilon} \setminus \big( \bigcup_{y\in S\setminus e} B_{\bar{r}+\varepsilon}(y) \big)$ is a witness point if and only if $\varepsilon = 0$, and $\Z{e} = \A^{\bar{r}} \setminus \big( \bigcup_{y\in S\setminus e} B_{\bar{r}}(y) \big)$.

\begin{figure}[!tp]
    \centering
    \includegraphics[width=2.5in]{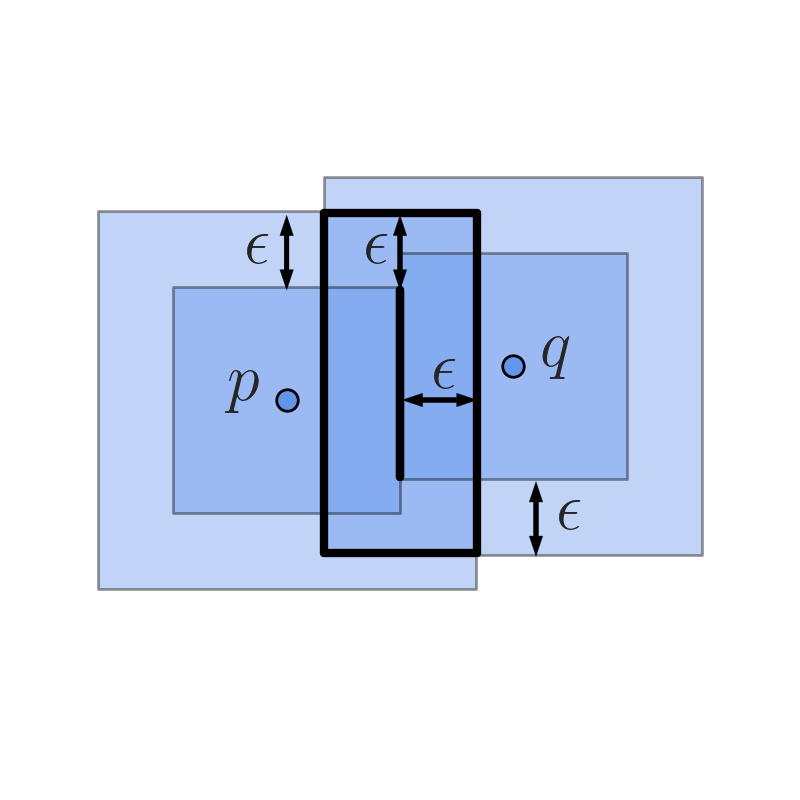}
    \caption{The $\varepsilon$-thickening of the non-empty intersection of two squares equals the intersection of the $\varepsilon$-thickenings of the squares.}
    \label{app-fig:epsilon-thickening}
\end{figure}
\begin{figure}[!tp]
    \centering
        \begin{subfigure}[b]{2in}
            \centering
            \includegraphics[width=2in]{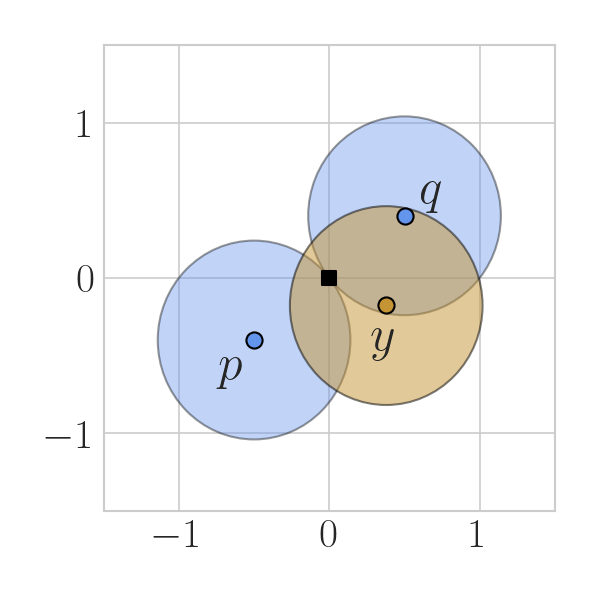}
            \caption{}
            \label{fig:witness1}
        \end{subfigure}
        \qquad
        \begin{subfigure}[b]{2in}
            \centering
            \includegraphics[width=2in]{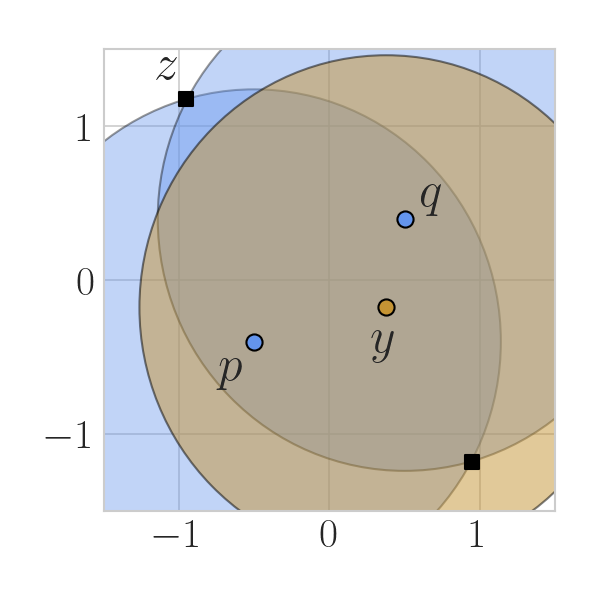}
            \caption{}
            \label{fig:witness2}
        \end{subfigure}
        \\
        \begin{subfigure}[b]{2in}
            \centering
            \includegraphics[width=2in]{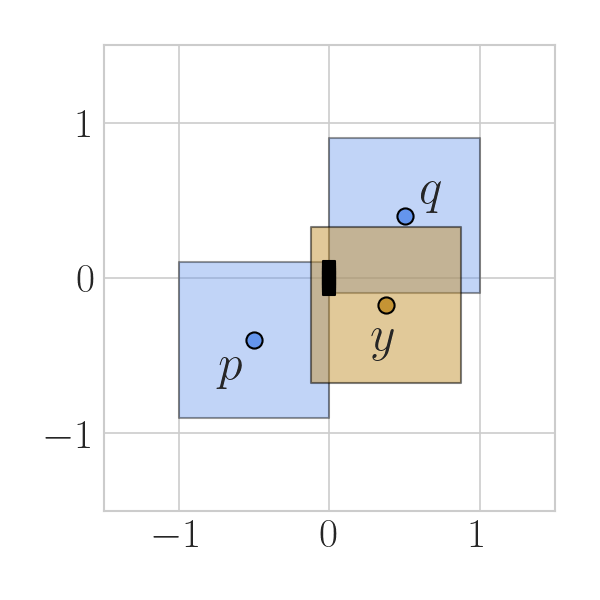}
            \caption{}
            \label{fig:witness3}
        \end{subfigure}
        \qquad
        \begin{subfigure}[b]{2in}
            \centering \includegraphics[width=2in]{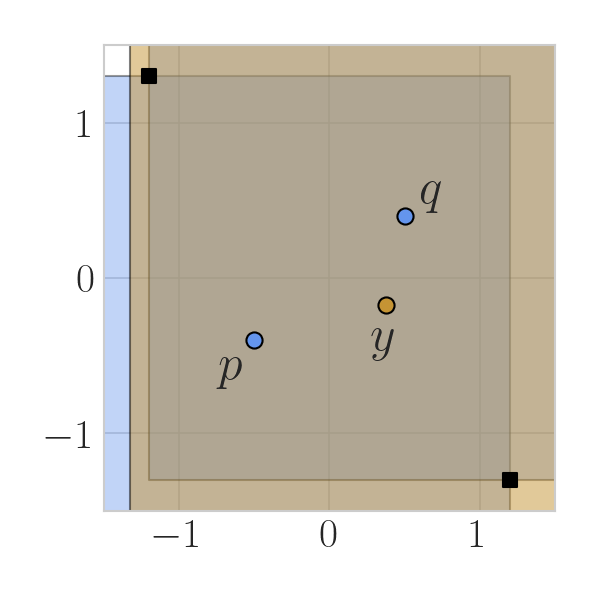}
            \caption{}
            \label{fig:witness4}
        \end{subfigure}
    \caption{In \textbf{(a)} Euclidean balls centered in $p$, $q$ intersect in a point which is covered by the ball centered in $y$.  As the radius grows in \textbf{(b)} this intersection is not covered by the ball centered in $y$, so that $z$ is a witness point of $e=\{p,q\}$.
    In \textbf{(c)} $\ell_{\infty}$-balls centered in $p$, $q$ intersect in $\A^{\bar{r}}$ which is covered by the $\linf$-ball centered in $y$. Again the radius grows in \textbf{(d)}, but in this case the $\linf$-ball centered in $y$ covers $\A^{\bar{r}+\varepsilon}$.}
    \label{fig:counterexample-epsilon-thickening}
\end{figure}

($\Rightarrow$) We prove this direction of the result by contradiction.
Let us suppose $\A^{\bar{r}}$ is covered by $\bigcup_{y\in S\setminus e} B_{\bar{r}}(y)$, i.e. $\Z{e}$ is empty.
We know that
\begin{align*}
\A^{\bar{r}+\varepsilon} 
& = 
\partial \overline{B_{\bar{r}+\varepsilon}(p)} \cap \partial \overline{B_{\bar{r}+\varepsilon}(q)} \\
& \subseteq 
\overline{B_{\bar{r}+\varepsilon}(p)} \cap \overline{B_{\bar{r}+\varepsilon}(q)} \\
& =
\varepsilon(\overline{B_{\bar{r}}(p)}) \cap \varepsilon(\overline{B_{\bar{r}}(q)})
= \varepsilon(\A^{\bar{r}}),
\end{align*}
because we can apply Proposition \ref{app-prop:properties-epsilon-thickening} \emph{(i)} to obtain $\varepsilon(\A^{\bar{r}}) = \varepsilon(\cB{\bar{r}}{p} \cap \cB{\bar{r}}{q}) = \varepsilon(\cB{\bar{r}}{p}) \cap \varepsilon(\cB{\bar{r}}{q})$.
This property of boxes is illustrated by Figure \ref{app-fig:epsilon-thickening}.
It follows that, $\A^{\bar{r}+\varepsilon} \subseteq \varepsilon(\A^{\bar{r}})$ for any $\varepsilon \geq 0$, that together with Proposition \ref{app-prop:properties-epsilon-thickening} \emph{(ii)} and \emph{(iii)} gives
\begin{equation*}
    \label{eq:epsilon-inclusion}
    \A^{\bar{r}+\varepsilon} 
    \subseteq 
    \varepsilon \bigg( \A^{\bar{r}} \bigg) 
    \subseteq 
    \varepsilon \bigg( \bigcup_{y\in S\setminus e} B_{\bar{r}}(y) \bigg)
    =
    \bigcup_{y\in S\setminus e} B_{\bar{r}+\varepsilon}(y),
\end{equation*}
for any $\varepsilon \geq 0$.
Thus,
$\A^{\bar{r}+\varepsilon} \subseteq \bigcup_{y\in S\setminus e} B_{\bar{r}+\varepsilon}(y)$, which contradicts the existence of $z \in \A^{\bar{r}+\varepsilon} \setminus \big( \bigcup_{y\in S\setminus e} B_{\bar{r}+\varepsilon}(y) \big)$,  for any $\varepsilon \geq 0$.

($\Leftarrow$) Any point in $\Z{e} \neq \emptyset$ belongs to $V_p \cap V_q$, so that $e \in K^D$.
\end{proof}
The above result is illustrated in Figure \ref{fig:counterexample-epsilon-thickening}, which shows how this characterization of $\linf$-Delaunay edges does not hold in Euclidean metric.
\section{Alpha and \v{C}ech Complexes in \texorpdfstring{$\R^2$}{R2}}
\label{sec-app:alpha-cech-R2}
We prove the equivalence of Alpha and \v{C}ech complexes for points in two-dimensions.
In this setting, $\linf$-Voronoi regions may have degenerate bisectors (containing a two-dimensional subset of $\R^2$), as explained in Section \ref{sec:pre} and illustrated in Figure \ref{fig:voronoi-delaunay-and-degenerate-case}.
To avoid such cases, we assume $S$ to be in general position. 
We recall the definition given in Section \ref{sec:alpha}.
%
\begin{definition}
\label{app-def:general-position}
Let $S$ be a finite set of points in $(\R^2, \din)$. 
We say that $S$ is in \emph{general position} if no four points lie on the boundary of a square, and no two points share a coordinate.
\end{definition}
After stating the Nerve Theorem, we prove the result given as Theorem \ref{thm:equiv-2D} in Section \ref{sec:alpha}.
\begin{theorem}[Theorem 10.7 \cite{handbook-comb}]
\label{app-thm:nerve}
Let $X$ be a triangulable space and $\{A_i\}_{i \in I}$ a locally finite family of open subsets (or a finite family of closed subsets) such that $X = \bigcup_{i \in I} A_i$.
If every non-empty intersection $A_{i_1} \cap A_{i_2} \cap \ldots \cap A_{i_t}$ is contractible, then $X$ and the nerve $\nrv(\{A_i\}_{i\in I})$ are homotopy equivalent.
\end{theorem}
\begin{theorem}
\label{app-thm:equiv-2D}
Let $S$ be a finite set of points in $(\R^2, \din)$ in general position.
The Alpha and \v{C}ech filtrations of $S$ are equivalent, i.e. produce the same persistence diagrams.
\end{theorem}
\begin{proof}
Alpha complexes $\alphar$ are nerves of collections of closed sets $\{ \cB{r}{p} \cap V_p \}_{p\in S}$ for $r \in \R$.
We show that any intersection of $k$ elements in any such collection is either empty or contractible.
\begin{itemize}
    \item $k=1$. The sets $\cB{r}{p} \cap V_p$ are star-like for any $r > 0$, because $\cB{r}{p}$ and $V_p$ are both star-like with respect to $p$. Thus, $\cB{r}{p} \cap V_p$ is contractible for any $p \in S$.
    \item $k = 2$.
    Let $p, q$ be two points of $S$, and $\bar{r} = \frac{\din(p,q)}{2}$. We show that $L = \cB{r}{p} \cap V_p \cap \cB{r}{q} \cap V_q$ is either empty or contractible.
    In $\R^2$ we have that $\A^{\bar{r}} = \cB{\bar{r}}{p} \cap \cB{\bar{r}}{q}$ is a line segment of length strictly less than $2\bar{r}$, by our general position assumption.
    If this line segment is covered by $\bigcup_{y\in S\setminus \{p, q\}} B_{\bar{r}}(y)$, then by Proposition \ref{prop:delaunay-edge} we have that the bisector $V_p \cap V_q$ is empty, so that $L$ is empty.
    Moreover, $L$ is empty if $r < \bar{r}$, because $\cB{r}{p} \cap \cB{r}{q}$ is empty.
    On the other hand, if $r \geq \bar{r}$ and $A' = \A^{\bar{r}} \setminus \bigcup_{y\in S\setminus \{p, q\}} B_{\bar{r}}(y)$ is a non-empty line segment, and we can show that $L$ is contractible.
    First, we define a deformation retraction $\phi$ of $V_p \cap V_q$ onto $A'$.
    This is obtained by taking the Euclidean projection of the points on the bisector and not in $A'$, i.e. $\left(V_p \cap V_q\right) \setminus A'$, onto $A'$.
    This can be done because $\left(V_p \cap V_q\right) \setminus A'$ contains at most two line segments, defined by the union of points in $\partial \cB{\bar{r}+\varepsilon}{p} \cap \partial \cB{\bar{r}+\varepsilon}{q}$ not contained in $\bigcap_{y\in S\setminus \{p,q\}} B_{\bar{r}+\varepsilon}(y)$ for any $\varepsilon > 0$.
    For instance, the bisector $V_p \cap V_q$ in Figure \ref{fig:non-conv} of the main paper retracts via $\phi$ to the two line segments oriented at a forty-five degree angle onto the horizontal line segment.
    Moreover, $\phi$ restricts to $L$, by the convexity of $\cB{r}{p} \cap \cB{r}{q}$, and the fact that this contains $A'$ for $r \geq \bar{r}$. 
    Hence, $L$ has the same homotopy type of $A'$, which is a line segment and so is contractible.
    \item $k = 3$. These intersections can either be empty or contain a single point, by the general position of $S$ and Corollary 3.18 of \cite{criado2019tropical}.
    \item $k > 3$. Any such intersection is empty, again by the general position of $S$.
\end{itemize}
Thus, we can apply the Nerve Theorem \ref{app-thm:nerve} obtaining that $X = \bigcup_{p\in S} \big(\cB{r}{p} \cap V_p\big)$ and $\alphar$ are homotopy equivalent for any $r \in \R$.
Besides $X = \bigcup_{p \in S} \cB{r}{p}$, and by applying the Nerve Theorem to the collection $\{\cB{r}{p} \}_{p \in S}$, we have that $X$ is homotopy equivalent to $\cechr$ as well.
So $\alphar \simeq \cechr$ for any $r \in \R$,  and the desired equivalence of Alpha and \v{C}ech filtrations follows by applying the Persistence Equivalence Theorem of \cite[Section 7.2]{edelsbrunner2010computational}.
\end{proof}
To conclude this section, we show that $\linf$-Delaunay and Alpha complexes of points in $(\R^2, \din)$ are flag complexes. This second result is stated as Proposition \ref{prop:delaunay-alpha-flag-in-2D} in the main paper.
\begin{proposition}
\label{app-prop:delaunay-alpha-flag-in-2D}
Let $S$ be a finite set of points in general position in $(\mathbb{R}^2, \din)$ and $r \geq 0$.
Both the $\linf$-Delaunay complex $K^D$ and the Alpha complex $\alphar$ of $S$ are flag complexes.
Moreover, $e = \{p, q\} \in K^D$ belongs to $\flagr$ if and only if $\frac{\din(p, q)}{2} \leq r$.
\end{proposition}
%
\begin{figure}[tb]
    \centering
        \begin{subfigure}[b]{2in}
            \centering
            \includegraphics[width=2in]{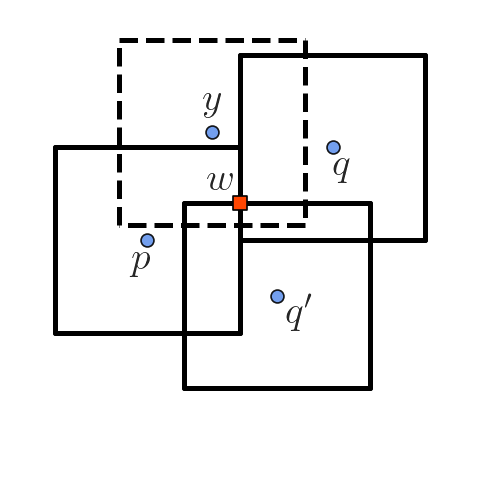}
            \caption{}
            \label{app-fig:line-above}
        \end{subfigure}
        \qquad
        \begin{subfigure}[b]{2in}
            \centering
            \includegraphics[width=2in]{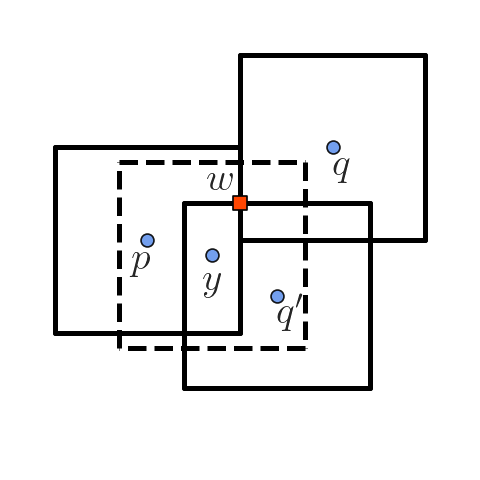}
            \caption{}
            \label{app-fig:line-below}
        \end{subfigure}
    \caption{Illustration of the proof of Proposition \ref{app-prop:delaunay-alpha-flag-in-2D}. In \textbf{(a)} the red square marker represents the point $w = (p_x + \bar{r}, q'+\bar{r})$ on $\A^{\bar{r}}$, which is covered by $B_{\bar{r}}(y)$ from above. 
    In \textbf{(b)} the same point is covered by $B_{\bar{r}}(y)$ from below. In both \textbf{(a)} and \textbf{(b)} the boundary of $B_{\bar{r}}(y)$ is drawn as a dashed line.}
    \label{app-fig:illustration-proof}
\end{figure}
\begin{proof}
We prove that all cliques on three edges belong to the $\linf$-Delaunay complex $K^D$ of $S$.
Consider three points $p, q, q' \subseteq S$, such that $\{p,  q\}$, $\{p, q'\}$ and $\{q, q'\}$ are $\linf$-Delaunay edges.
Without loss of generality, we assume $\{p, q\}$ to be the longest edge. 
We have $\A^{\bar{r}}  = \partial \cB{\bar{r}}{p} \cap \partial \cB{\bar{r}}{q} = \cB{\bar{r}}{p} \cap \cB{\bar{r}}{q}$ by Proposition \ref{app-prop:delaunay-edge}, where $\bar{r} = \frac{\din(p, q)}{2}$.
By the general position of $S$, it follows that $\A^{\bar{r}}$ is a non-empty axis-parallel line segment of length less than $2\bar{r}$.
Moreover, $\cB{\bar{r}}{p} \cap \cB{\bar{r}}{q} \cap \cB{\bar{r}}{q'}$ is non-empty by Proposition \ref{prop:intersection-boxes} \emph{(ii)} and the definition of $\bar{r}$.
So, the closed square $\cB{\bar{r}}{q'}$ intersects $\A^{\bar{r}}$, i.e. $\A^{\bar{r}} \cap \cB{\bar{r}}{q'} \neq \emptyset$. 
Then, either $\cB{\bar{r}}{q'}$ covers $\A^{\bar{r}}$ or intersects a subsegment of $\A^{\bar{r}}$ containing one of the endpoints of $\A^{\bar{r}}$.

In the former case, the interior of the square $B_{\bar{r}}(q')$ also covers $\A^{\bar{r}}$, by general position.
So, Proposition \ref{app-prop:delaunay-edge} implies that $\{p, q\}$ is not a $\linf$-Delaunay edge, which contradicts our hypothesis.
In the latter case, if we assume without loss of generality that $\A^{\bar{r}}$ is a vertical line segment and that $\cB{\bar{r}}{q'}$ intersects it from below, then the point $w = (p_x +\bar{r}, q'_x + \bar{r}) \in \A^{\bar{r}} \cap \partial \cB{\bar{r}}{q'}$, where $p = (p_x, p_y)$, $q'=(q'_x, q'_y) \in \R^2$, is the only possible witness point of the triangle $\{p, q, q'\}$ by our general position assumption. See Figure \ref{app-fig:illustration-proof}.

We suppose by contradiction that $w$ is contained in an open square $B_{\bar{r}}(y)$ with sides of length $2\bar{r}$, so that $\{p, q, q'\} \not\in K^D$, and show that in every possible case one between $\{p, q\}$, $\{p, q'\}$, and $\{q, q'\}$ cannot be a $\linf$-Delaunay edge.
Note that $B_{\bar{r}}(y)$ can cover $w$ from either above or below, see Figures \ref{app-fig:line-above} and \ref{app-fig:line-below} respectively.

In the first case, $\cB{\bar{r}}{q'} \cup \cB{\bar{r}}{y}$ covers $\A^{\bar{r}}$, so $\{p, q\}$ cannot be a $\linf$-Delaunay edge by Proposition \ref{app-prop:delaunay-edge}, which is a contradiction.

In the second case, one can check that $y$ has to belong to either $\textrm{Mini}_{pq'}$ or $\textrm{Mini}_{qq'}$.
Thus, either $\{p, q'\}$ or $\{q, q'\}$ cannot be a $\linf$-Delaunay edge by Proposition \ref{prop:minibox-property}, which is again a contradiction. 

We conclude by showing that $\alphar$ is also a flag complex.
By Proposition \ref{app-prop:delaunay-edge} any edge $e = \{p, q\}$ is added into $\alphar$ at $r = \frac{\din(p,q)}{2}$.
Moreover, when the longest edge of any $\linf$-Delaunay triangle $\tau$ is added at radius $\bar{r}$, also $\tau$ is added in $K_{\bar{r}}^A$, because from the discussion above there exist $w$ at distance $\bar{r}$ from the vertices of $\tau$, which is a witness of this triangle.
\end{proof}
\section{Supporting Lemmas for Proving Alpha Flag and \v{C}ech Equivalence}
\label{sec-app:proof-main}
In this section, we present various results used in the proof of Theorem \ref{thm:main}.
It should be noted that we do not make use of any general position assumption.
We start by recalling the definitions of single edge-length range, and edge-by-edge filtration, given in Section \ref{sec:alpha-flag} of the paper.
\begin{definition}
\label{app-def:edge-by-edge-filtration}
Let $S$ be a finite set of points in $(\R^d, \din)$.
A \emph{single edge-length range} of \v{C}ech complexes of $S$ is an open interval $(r, r + \varepsilon) \subseteq \R$ such that all the edges not in $\cechr$ and contained in $\cechre$ have the same length $2\bar{r}$.
Given a single edge-length range $(r, r+\varepsilon)$, the  \emph{\v{C}ech edge-by-edge filtration} of $S$ on this range is 
$$\cechr 
=
\KC{0}
\subseteq 
\KC{1}
\subseteq 
\ldots 
\subseteq
\KC{n_i}
=
\cechre,$$
where $\KC{i}$ contains exactly one edge not in $\KC{i-1}$, together with the cliques containing this edge, for each $1 \leq i \leq n_i$.
The corresponding \emph{Alpha flag edge-by-edge filtration} of $S$ on the same range is
$$\flagr
=
\KAF{0}
\subseteq 
\KAF{1}
\subseteq 
\ldots 
\subseteq
\KAF{n_i}
=
\flagre,$$
where $\KAF{i} = \KC{i} \cap \flagre$ for each $1 \leq i \leq n_i$.
\end{definition}
The following result corresponds to Lemma \ref{lemma:proof-main-1-add-one-edge} in the paper.
\begin{lemma}
\label{app-lemma:proof-main-1-add-one-edge}
Let $(r, r+\varepsilon)$ be a single edge-length range of \v{C}ech complexes of $S \subseteq (\R^d, \din)$, and $\{\KAF{i}\}_{i=0}^{n_i}$, $\{\KC{i}\}_{i=0}^{n_i}$ the Alpha flag and \v{C}ech edge-by-edge filtrations on this range. 
If going from $\KAF{i-1}$ to $\KAF{i}$ a $\linf$-Delaunay edge is the only simplex added in $\KAF{i}$, then this is also the only simplex added going from $\KC{i-1}$ to $\KC{i}$.
\end{lemma}
\begin{proof}
Let $e=\{p, q\}$ be the $\linf$-Delaunay edge added in $\KAF{i}$. 
We define $\bar{r} = \frac{\din(p, q)}{2}$, so that $r < \bar{r} < r + \varepsilon$, and $\openY = \{ y \in S \ \vert \  \din(y, p) < 2 \bar{r} \textrm{ and  } \din(y, q) < 2 \bar{r} \}$.

We start by proving that $\openY$ is empty.
The idea is to assume $\openY \neq \emptyset$, and show that $\KAF{i}$ has to contain a triangle of which $e$ is an edge, which is a contradiction.
To begin with, $e$ is $\linf$-Delaunay, so $\A^{\bar{r}}$ is not covered by $\bigcup_{y \in \openY} B_{\bar{r}}(y)$ by Proposition \ref{app-prop:delaunay-edge}.
Moreover, the closure $\overline{\bigcup_{y \in \openY} B_{\bar{r}}(y)}$ intersects $\A^{\bar{r}}$, because each element in $\{ \cB{\bar{r}}{y} \}_{\openY}$ intersects $\A^{\bar{r}}$ by definition of $\openY$ and Proposition \ref{prop:intersection-boxes} \emph{(ii)}.
Thus, there exists $z \in \partial \overline{\bigcup_{y \in \openY} B_{\bar{r}}(y)}$ such that $z \in \A^{\bar{r}}$, because $\A^{\bar{r}}$ is convex and closed.
Since $\partial \overline{\bigcup_{y \in \openY} B_{\bar{r}}(y)} \subseteq \bigcup_{y \in \openY} \partial \cB{\bar{r}}{y}$, it follows that $z$ is a point on a boundary $\partial \cB{\bar{r}}{\hat{y}}$ for some $\hat{y} \in \openY$.
We conclude that $z$ is a witness point of $\hat{\tau}  = \{p, q, \hat{y}\} \in K^D$, because $z \in \A^{\bar{r}} \cap \partial \cB{\bar{r}}{\hat{y}}$ and $z$ is not contained in the interior of any $\cB{\bar{r}}{y}$ for $y \in S \setminus \{p, q\}$, otherwise $z$ would not be on the boundary $\partial \overline{\bigcup_{y \in \openY} B_{\bar{r}}(y)}$.
By definition of $\openY$, both $\{p, \hat{y}\}$ and $\{q, \hat{y}\}$ are strictly shorter than $2 \bar{r}$, so they both belong to $\KAF{i-1}$.
Finally, from the above discussion we have $\{p, \hat{y}\}$, $\{q, \hat{y}\} \in \KAF{i-1}$ and $\hat{\tau} \in K^D$, so adding $e$ in the flag complex $\KAF{i}$ also adds $\hat{\tau}$ in $\KAF{i}$, which is the desired contradiction.

We can now prove that $e$ is also the only simplex added in $\KC{i}$.
Suppose there exists $\tau' = \{p, q, y'\} \in \KC{i}$.
It follows that there exists $y' \in S \setminus \{p, q\}$ such that $\A^{\bar{r}} \cap \cB{\bar{r}}{y'} = \cB{\bar{r}}{p} \cap \cB{\bar{r}}{q} \cap \cB{\bar{r}}{y'} \neq \emptyset$.
Then, $\{p, y'\}$, $\{q, y'\} \in \KC{i-1}$ with $\din(p, y') \leq 2 \bar{r}$ and $\din(q, y') \leq 2 \bar{r}$.
It cannot be that both $\din(p, y') < 2 \bar{r}$ and $\din(q, y') < 2 \bar{r}$, otherwise $\openY$ would not be empty.
So $\din(p, y') = 2 \bar{r}$ or $\din(q, y') = 2 \bar{r}$, and $\A^{\bar{r}} \cap \cB{\bar{r}}{y'} = \A^{\bar{r}} \cap \partial \cB{\bar{r}}{y'}$.
Moreover, any point $z \in \A^{\bar{r}} \cap \partial \cB{\bar{r}}{y'}$ is a witness point of $\tau' = \{p, q, y'\}$, because there does not exist any open ball $B_{\bar{r}}(y)$ containing $z$ and intersecting $\A^{\bar{r}}$ centered in $y \in S \setminus \{p, q\}$, otherwise $\openY$ would not be empty.

In conclusion, we have $\{p, q, y' \} \in K^D$, which implies $\{p, q\}$, $\{p, y'\}$, $\{q, y'\} \in K^D$, and $\{p, y'\}$, $\{q, y'\} \in \KC{i-1}$ by our hypothesis on $\tau'$.
So, $\{p, y'\}$, $\{q, y'\} \in \KAF{i-1}$ by definition of $\KAF{i-1}$, and adding $e$ in the flag complex $\KAF{i}$ also adds $\tau'$ in $\KAF{i}$, which is a contradiction.
Thus, there does not exist a triangle $\tau'$ containing $e$ in $\KC{i}$, which implies that $e$ is the only simplex added going from $\KC{i-1}$ to $\KC{i}$.
\end{proof}

%
\begin{lemma}
\label{app-lemma:projection-hyperrect}
Let $B_1$ and $B_2$ be two closed boxes in $\mathbb{R}^d$.
If $B_1 \cap B_2$ is non-empty, then the Euclidean projection 
$\uppi_{B_1}: B_1 \rightarrow B_2$,
defined by mapping each $x \in B_1$ to its closest point in Euclidean distance on $B_2$, is such that $\uppi_{B_1}(B_1) \subseteq B_1 \cap B_2$.
\end{lemma}
\begin{proof}
Let 
$B_1 = \prod_{i=1}^d [a_i^{B_1}, b_i^{B_1}]$ 
and
$B_2 = \prod_{i=1}^d [a_i^{B_2}, b_i^{B_2}]$ such that $B_1 \cap B_2 \neq \emptyset$.
Because Cartesian products and intersections of intervals commute, we have that
$[a_i^{B_1}, b_i^{B_1}] \cap [a_i^{B_2}, b_i^{B_2}]
=
[\bar{a}_i, \bar{b}_i]
\neq 
\emptyset$
for each $1 \leq i \leq d$, 
and $B_1 \cap B_2 = \prod_{i=1}^d [\bar{a}_i, \bar{b}_i]$.

Given $x \in B_1$, we suppose by contradiction that 
$y = \uppi_{B_1}(x) \in B_2$ is such that $y \not\in B_1 \cap B_2$.
Thus, $y \not \in  \prod_{i=1}^d [\bar{a}_i, \bar{b}_i]$, and there exists $1 \leq \hat{i} \leq d$ such that $y_{\hat{i}} \not\in [\bar{a}_{\hat{i}}, \bar{b}_{\hat{i}}]$.
The intervals $[a_{\hat{i}}^{B_1}, b_{\hat{i}}^{B_1}]$ and $[a_{\hat{i}}^{B_2}, b_{\hat{i}}^{B_2}]$ can intersect in four possible ways:
\begin{itemize}
    \item[\emph{(i)}] $[a_{\hat{i}}^{B_1}, b_{\hat{i}}^{B_1}]$ intersects $[a_{\hat{i}}^{B_2}, b_{\hat{i}}^{B_2}]$ on the left, i.e. 
    $a_{\hat{i}}^{B_1} 
    \leq 
    a_{\hat{i}}^{B_2} 
    \leq 
    b_{\hat{i}}^{B_1} 
    \leq 
    b_{\hat{i}}^{B_2}$. 
    Thus, $a_{\hat{i}}^{B_1} \leq x_{\hat{i}} \leq b_{\hat{i}}^{B_1} < y_{\hat{i}}$, and we define $y' = [y_1, \ldots, b_{\hat{i}}^{B_1}, \ldots, y_d]$;
    \item[\emph{(ii)}] $[a_{\hat{i}}^{B_1}, b_{\hat{i}}^{B_1}]$ intersects $[a_{\hat{i}}^{B_2}, b_{\hat{i}}^{B_2}]$ on the right, i.e. 
    $a_{\hat{i}}^{B_2} 
    \leq 
    a_{\hat{i}}^{B_1} 
    \leq 
    b_{\hat{i}}^{B_2} 
    \leq 
    b_{\hat{i}}^{B_1}$. 
    Thus, $y_{\hat{i}} < a_{\hat{i}}^{B_1} \leq x_{\hat{i}} \leq b_{\hat{i}}^{B_1}$, and we define $y'' = [y_1, \ldots, a_{\hat{i}}^{B_1}, \ldots, y_d]$;
    \item[\emph{(iii)}] $[a_{\hat{i}}^{B_1}, b_{\hat{i}}^{B_1}]$ is contained in $[a_{\hat{i}}^{B_2}, b_{\hat{i}}^{B_2}]$, i.e. 
    $a_{\hat{i}}^{B_2} 
    \leq 
    a_{\hat{i}}^{B_1} 
    \leq 
    b_{\hat{i}}^{B_1} 
    \leq 
    b_{\hat{i}}^{B_2}$. 
    Thus, $a_{\hat{i}}^{B_1} \leq x_{\hat{i}} \leq b_{\hat{i}}^{B_1} < y_{\hat{i}}$ or $y_{\hat{i}} < a_{\hat{i}}^{B_1} \leq x_{\hat{i}} \leq b_{\hat{i}}^{B_1}$, and in the first case we define $y' = [y_1, \ldots, b_{\hat{i}}^{B_1}, \ldots, y_d]$ and in the second $y'' = [y_1, \ldots, a_{\hat{i}}^{B_1}, \ldots, y_d]$;
    \item[\emph{(iv)}] $[a_{\hat{i}}^{B_1}, b_{\hat{i}}^{B_1}]$ contains $[a_{\hat{i}}^{B_2}, b_{\hat{i}}^{B_2}]$, i.e. 
    $a_{\hat{i}}^{B_1} 
    \leq 
    a_{\hat{i}}^{B_2} 
    \leq 
    b_{\hat{i}}^{B_2} 
    \leq 
    b_{\hat{i}}^{B_1}$.
\end{itemize}
In case \emph{(iv)} we have a contradiction as $$y_{\hat{i}} \in [a_{\hat{i}}^{B_2}, b_{\hat{i}}^{B_2}] = [\bar{a}_{\hat{i}}, \bar{b}_{\hat{i}}] \not \ni y_{\hat{i}}.$$
In the other three cases, taken either $y'$ or $y''$ we have
\begin{align}
    \label{eq:diseg1}
    d_2(x, y') 
    = 
    \sqrt{(x_{\hat{i}} - b_{\hat{i}}^{B_1})^2 + \sum_{i=1, i\neq \hat{i}}^d (x_i - y_i)^2}
    <
    \sqrt{\sum_{i=1}^d (x_i - y_i)^2}
    = 
    d_2(x, y),\\
    \label{eq:diseg2}
    d_2(x, y'') 
    = 
    \sqrt{(x_{\hat{i}} - a_{\hat{i}}^{B_1})^2 + \sum_{i=1, i\neq \hat{i}}^d (x_i - y_i)^2}
    <
    \sqrt{\sum_{i=1}^d (x_i - y_i)^2}
    = 
    d_2(x, y).
\end{align}
because $(x_{\hat{i}} - b_{\hat{i}}^{B_1})^2 < (x_{\hat{i}} - y_{\hat{i}})^2$ in Equation \eqref{eq:diseg1}, and $(x_{\hat{i}} - a_{\hat{i}}^{B_1})^2 < (x_{\hat{i}} - y_{\hat{i}})^2$ in Equation \eqref{eq:diseg2}.
The proof follows because this contradicts $y$ being the closest point in Euclidean distance to $x$ in $B_2$.
\end{proof}
%
%
%
%
\begin{lemma}
\label{app-lemma:retraction}
Let $p$, $q \in (\R^d, \din)$ be such that $\din(p, q) = 2\bar{r}$, and $\A^{\bar{r}} = \cB{\bar{r}}{p} \cap \cB{\bar{r}}{q}$.
Given a finite set of points $\mathcal{Y} \subseteq (\R^d, \din)$ such that $\A^{\bar{r}}$ is covered by $\bigcup_{y \in \mathcal{Y}} B_{\bar{r}}(y)$, then $\nrv\left( \{\cB{\bar{r}}{y}\}_{y \in \mathcal{Y}} \right)$ has the homotopy type of $\A^{\bar{r}}$.
\end{lemma}
\begin{proof}
From the Nerve Theorem \ref{app-thm:nerve}, it follows that $\nrv\left( \{\cB{\bar{r}}{y}\}_{y \in \mathcal{Y}} \right)$ and $\bigcup_{y \in \mathcal{Y}} \cB{\bar{r}}{y}$ are homotopy equivalent, because convex sets and their intersections are contractible.
We show how to define a deformation retraction
\begin{equation*}
    \phi: \bigg( \bigcup_{y\in \mathcal{Y}} 
          \cB{\bar{r}}{y} \bigg)
          \times 
          [0,1]
          \rightarrow
          \A^{\bar{r}},
\end{equation*}
which implies that $\bigcup_{y\in \mathcal{Y}}
\cB{\bar{r}}{y}$ and $\A^{\bar{r}}$ have the same homotopy type.

To obtain $\phi$, we first define $\phi_y: \cB{\bar{r}}{y} \times [0, 1] \rightarrow \A^{\bar{r}}$
for each $y \in \mathcal{Y}$.
Given the Euclidean projection $\uppi_{\cB{\bar{r}}{y}}: \cB{\bar{r}}{y} \rightarrow \A^{\bar{r}}$,
we set 
\begin{equation*}
    \phi_y(x, t) = (1-t) \cdot x + t \cdot \uppi_{\cB{\bar{r}}{y}}(x),
\end{equation*}
for every $x \in \cB{\bar{r}}{y}$ and $t\in [0,1]$. 
It should be noted that $\A^{\bar{r}}$ is a $(d-1)$-dimensional closed hyperrectangle by Proposition \ref{app-prop:delaunay-edge}.
So, we have $\uppi_{\cB{\bar{r}}{y})}(x) \in \cB{\bar{r}}{y} \cap \A^{\bar{r}}$, by Lemma \ref{app-lemma:projection-hyperrect}.
Moreover, the straight line segment from $x$ to $\uppi_{B_{\bar{r}}(y)}(x)$ is fully contained in $B_{\bar{r}}(y)$, by the convexity of this set.
Thus, $\phi_y$ is well-defined and continuous by the continuity of $\uppi_{B_{\bar{r}}(y)}$. 
We set
\begin{equation*}
    \phi(x, t) = \phi_{\hat{y}}(x,t),
\end{equation*}
for every $x\in \bigcup_{y\in \mathcal{Y}} B_{\bar{r}}(y)$ and $t\in [0,1]$, with $\hat{y} \in \mathcal{Y}$ such that $x \in B_{\bar{r}}(\hat{y})$.
This might seem not well-defined, because for a given $x$ all the $\phi_{\hat{y}}$ corresponding to a point in $\hat{\mathcal{Y}} = \{\hat{y} \in \mathcal{Y} \ \vert \ x \in B_{\bar{r}}(\hat{y}) \}$ can be used to define $\phi(x, t)$ for any $t \in [0,1]$.
Luckily, given
$R = \bigcap_{\hat{y} \in \hat{\mathcal{Y}}} B_{\bar{r}}(\hat{y})$, which is a box containing $x$, Proposition \ref{app-lemma:projection-hyperrect} guarantees that 
$\uppi_R: R \rightarrow \A^{\bar{r}}$ is such that $\uppi_R(R) \subseteq R\cap \A^{\bar{r}}$.
Thus, $\phi$ is well-defined because the straight line segment defined by
$(1-t)\cdot x + t\cdot \uppi_{R}(x)$ for $t\in [0,1]$ is contained within $R$, again by convexity.
Furthermore, $\phi$ is continuous by the continuity of the Euclidean projections $\uppi_{B_{\bar{r}}(y)}$, and is a deformation retraction onto $\A^{\bar{r}}$ because $\A^{\bar{r}} \subseteq \bigcup_{y\in \mathcal{Y}} B_{\bar{r}}(y)$ by hypothesis.
\end{proof}
%
%
%
We conclude this section by presenting the proof of the result stated as Lemma \ref{lemma:removing-non-delaunay-edge} in the paper.
\begin{lemma}
\label{app-lemma:removing-non-delaunay-edge-full}
Let $(r, r+\varepsilon)$ be a single edge length range of \v{C}ech complexes of $S \subseteq (\R^d, \din)$, and $\{\KC{i}\}_{i=0}^{n_i}$ the \v{C}ech edge-by-edge filtration on this range. 
If the edge $e = \{p, q\}$ added going from $\KC{i-1}$ to $\KC{i}$ is non-Delaunay for $1 \leq i \leq n_i$, then $H_k(\KC{i} \setminus \st(e)) = H_k(\KC{i-1})$ and $H_k(\KC{i})$ are isomorphic for $k=0,1$.
\end{lemma}
\begin{proof}
We can apply the reduced Mayer-Vietoris sequence, as given in \cite[Section~4.6]{spanier2012algebraic}, with $A = \cl(\st(e)) \subseteq \cechr$ and $B = \KC{i} \setminus \st(e)$, so that $A \cap B = \cl(\st(e)) \setminus \st(e)$.
We obtain
\begin{align*}
	\cdots \rightarrow 
	\tilde{H}_k(A\cap B) 
	\rightarrow 
	\tilde{H}_k(A) \oplus \tilde{H}_k(B) &
	\rightarrow
	\tilde{H}_k(A \cup B) 
	\rightarrow
	\tilde{H}_{k-1}(A\cap B)
	\rightarrow \cdots
	\\
	& \Downarrow
	\\
	\cdots \rightarrow 
	\tilde{H}_k(\cl(\st(e)) \setminus \st(e)) 
	\rightarrow 
	\tilde{H}_k(\KC{i} \setminus \st(e)) & 
	\rightarrow
	\tilde{H}_k(\KC{i}) 
	\rightarrow
	\tilde{H}_{k-1}(\cl(\st(e)) \setminus \st(e))
	\rightarrow \cdots
\end{align*}
where $\tilde{H}_k(A)$ cancels out, because it is trivial by definition of $A$.
Thus, showing that $\tilde{H}_k(\cl(\st(e)) \setminus \st(e))$ is trivial in homological degrees $k$ and $k-1$, implies that %
$\tilde{H}_k(\KC{i} \setminus \st(e))
\rightarrow 
\tilde{H}_k(\cechr)$ 
is an isomorphism, from the exactness of the Mayer-Vietoris sequence above.
Note that $\KC{i} \setminus \st(e) = \KC{i-1}$ by definition of \v{C}ech edge-by-edge filtration.
\begin{figure}[tb]
    \centering
    \begin{subfigure}[b]{2in}
        \centering
        \includegraphics[width=2in]{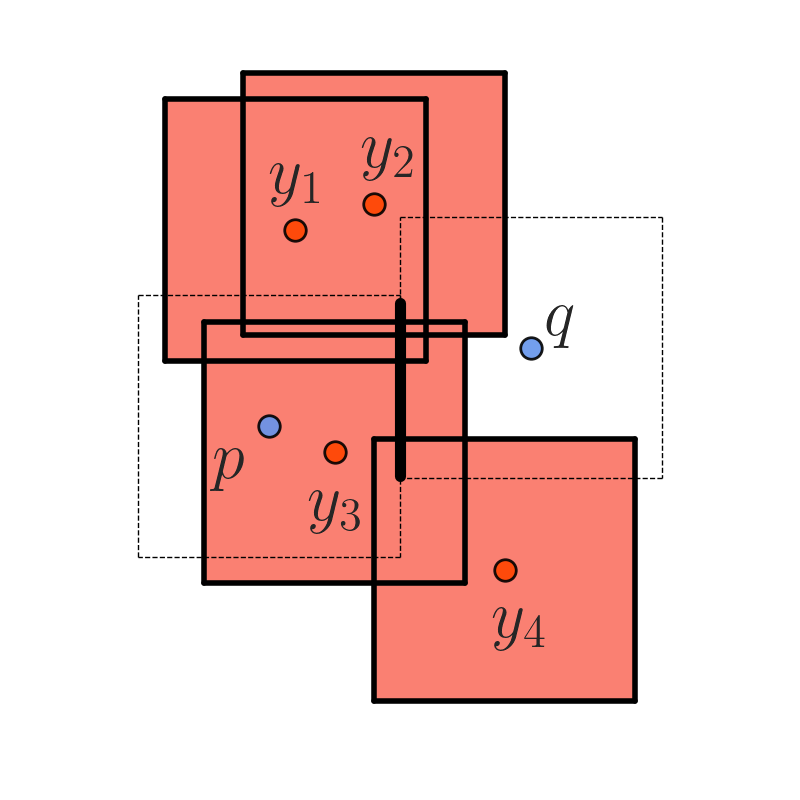}
        \caption{}
        \label{app-fig:example-main1}
    \end{subfigure}
    \begin{subfigure}[b]{2in}
        \centering
        \includegraphics[width=2in]{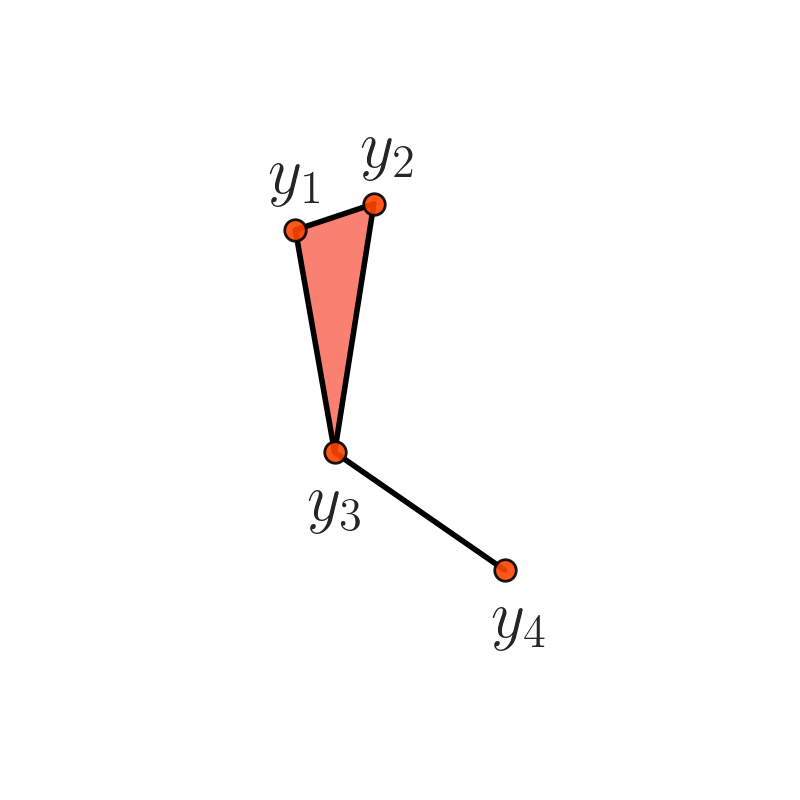}
        \caption{}
        \label{app-fig:example-main2}
    \end{subfigure}
    \begin{subfigure}[b]{2in}
        \centering
        \includegraphics[width=2in]{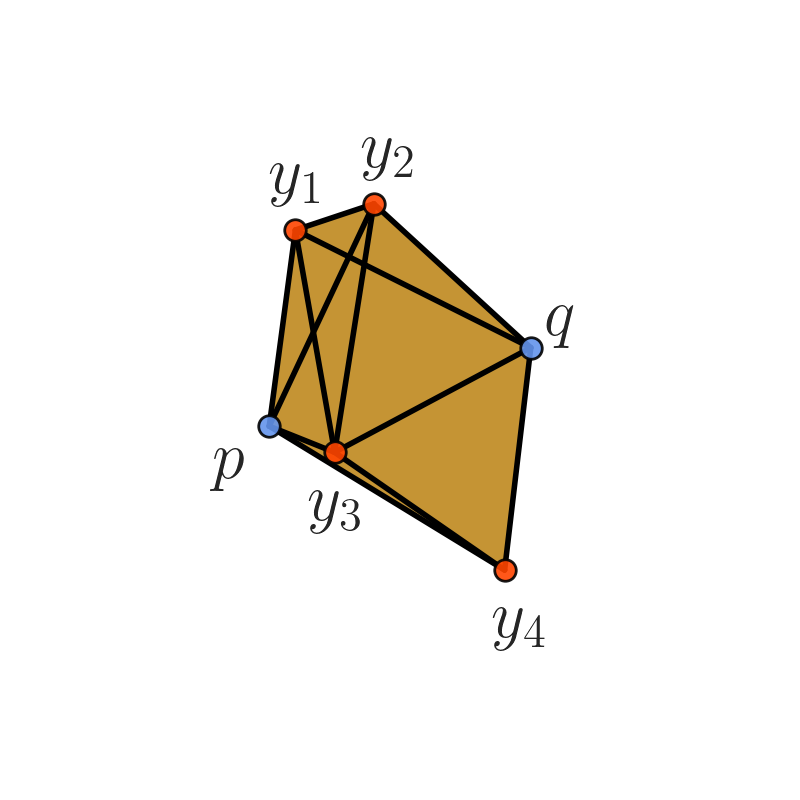}
        \caption{}
        \label{app-fig:example-main3}
    \end{subfigure}
    \caption{\textbf{(a)} $\linf$-Balls centered in the points of $\openY = \{y_1, y_2, y_3, y_4\}$ covering $\A^{\bar{r}}$. 
    \textbf{(b)} $K_0$ is complex on $\openY$ with the structure of the nerve $\text{Nrv}\left(\{\cB{\bar{r}}{x}\}_{x\in \openX}\right)$.
    \textbf{(c)} $K_1$, the union of the cones from $K_0$ to $p$ and $q$.}
    \label{app-fig:example-main}
\end{figure}

We define $\openY = 
\{ y \in S 
\ \vert \ 
\din(y, p) < 2 \bar{r} 
\textrm{ and }
\din(y, q) < 2 \bar{r}
\}
$, where $\bar{r} = \frac{\din(p,q)}{2}$.
This is the set of points in $S$ such that $B_{\bar{r}}(y) \cap \A^{\bar{r}} \neq \emptyset$, where $\A^{\bar{r}} = \cB{\bar{r}}{p} \cap \cB{\bar{r}}{q}$.
Furthermore, $\bigcup_{y \in \openY} B_{\bar{r}}(y)$ covers $\A^{\bar{r}}$ by Proposition \ref{app-prop:delaunay-edge}.
We define 
$$\linkY
=
\{y \in S
\ \vert \ 
\{p, y\}
\textrm{ and }
\{q, y\} 
\textrm{ are edges of }
\KC{i-1}
\},$$
which is the set of vertices of $\cl(\st(e))$.
It follows $\openY \subseteq \linkY$, because $\KC{i-1}$ contains all the edges strictly shorter than $2\bar{r}$.
In particular, $p$ and $q$ are not in $\openY$.
The idea is to build a complex $K_0$ with trivial homology on the vertices of $\openY$, prove that $\finalcomplex$ contains it, and finally that the additional simplices in $\finalcomplex$ do not alter the homology of $K_0$ in degrees zero and one.

%
A possible candidate for $K_0$ is the subcomplex $\nrvOpenY$, with $\openY$ as set of vertices.
However, this might contain edges of length $2\bar{r}$, i.e. $\{y', y''\} \in \nrvOpenY$ such that $\din(y', y'') = 2\bar{r}$ and $y', y'' \in \openY$.
Thus, some edge $\{y', y''\} \in \nrvOpenY$ might not be in $\KC{i-1}$, which implies that $\{y', y''\}$ is also not in $\cl(\st(e))$.
To solve this issue, we map the points $\openY$ into a set $\openX$ such that $\din(x', x'') \neq 2 \bar{r}$.
In particular, we define $\openX$ as a small perturbation of $\openY$ (i.e. each $x \in \openX$ corresponds to a $y \in \openY$ and the coordinates of $x$ and $y$ are arbitrarily close) such that:
\begin{itemize}
    \item[\emph{(i)}] edges on $\openX$ have length different from $2\bar{r}$, i.e. $\din(x', x'') \neq 2\bar{r}$ for each $x', x'' \in \openX$;
    \item[\emph{(ii)}] the union of open balls on $\openX$ covers $\A^{\bar{r}}$, i.e. $\A^{\bar{r}} \subseteq \bigcup_{x \in \openX} B_{\bar{r}}(x)$;
    \item[\emph{(iii)}] the pattern of intersection of open balls on $\openY$ and $\openX$ is the same.
\end{itemize}
Because the elements in the finite family $\{B_{\bar{r}}(y) \}_{y \in \openY}$ are open sets, it follows that there exists $\openX$ with properties \emph{(i)}, \emph{(ii)}, and \emph{(iii)} above.
Importantly, the nerve $\nrv\left( \{ \cB{\bar{r}}{x} \}_{x \in \openX} \right)$ of the family of closed balls centered in points of $\openX$ only containing edges strictly shorter than $2\bar{r}$.
We define $K_0$ as the complex on $\openY$ with the combinatorial structure of the nerve $\nrv\left( \{ \cB{\bar{r}}{x} \}_{x \in \openX} \right)$, see Figure \ref{app-fig:example-main2}.
So $K_0 \subseteq \cl(\st(e))$, because vertices and edges of $K_0$ are a subset of those of $\cl(\st(e))$ and this is a subcomplex of a flag complex.
Moreover, $K_0$ does not contain $e$, and so does not contain any simplex of $\st(e)$.
It follows that
\begin{align*}
    K_0 \setminus \st(e) 
    &
    \subseteq 
    \cl(\st(e)) \setminus \st(e)
    \\
    & \Downarrow
    \\
    K_0
    & 
    \subseteq 
    \cl(\st(e)) \setminus \st(e)
\end{align*}
We can show the existence of a filtration
$$
K_0 \subseteq
K_1 \subseteq
\ldots 
\subseteq
K_j
\subseteq
\ldots 
\subseteq
K_{n_j}
=
\cl(\st(e)) \setminus \st(e),
$$
such that if $\tilde{H}_k(K_{j-1})$ is trivial for $k=0, 1$, then $\tilde{H}_k(K_{j})$ is also trivial for $k=0,1$ for each $1 \leq j \leq n_j$.

By Lemma \ref{app-lemma:retraction}, $K_0$ has the same homotopy type of $\A^{\bar{r}}$, which is convex and so contractible, see Figure \ref{app-fig:example-main1}.
Hence, $\tilde{H}_k\left(K_0\right)$ is trivial for any $k \geq 0$.

Next, we define $K_1$ as the union of the cone from $K_0$ to $p$ and the cone from $K_0$ to $q$, see Figure \ref{app-fig:example-main3}.
Note that these cones are in $\cl(\st(e)) \setminus \st(e)$ because all the edges from $p$ and $q$ to $K_0$ are strictly shorter than $2\bar{r}$ by definition of $\openY$.
Importantly, the complex $K_1$ has trivial reduced homology because it collapses on $K_0$.

Then, we define each step $K_{j-1} \subseteq K_{j}$ for $2 \leq j \leq |\linkY| - 1$ by adding one of the vertices of $\linkY$ not in $K_{j-1}$. 
In particular, we add each of these vertices together with three edges and two triangles which all belong to $\cl(\st(e)) \setminus \st(e)$.
Let $y' \in \linkY$ be the vertex to be added in $K_{j}$.
We have $\din(p, q) = 2 \bar{r}$, and $\din(y', p) \leq 2 \bar{r}$, $\din(y', q) \leq 2 \bar{r}$ by definition of $\linkY$.
So, $\cB{\bar{r}}{y'} \cap \A^{\bar{r}} = \cB{\bar{r}}{y'} \cap \cB{\bar{r}}{p} \cap \cB{\bar{r}}{q} \neq \emptyset$ by Proposition \ref{prop:intersection-boxes} \emph{(ii)}.
It follows that there exists $z \in \cB{\bar{r}}{y'} \cap \A^{\bar{r}}$, and because $\A^{\bar{r}}$ is covered by $\bigcup_{y \in \openY} B_{\bar{r}}(y)$ it must be that $z \in B_{\bar{r}}(y'')$ for some $y'' \in \openY$.
We have $\din(y', y'') < 2 \bar{r}$, because $\cB{\bar{r}}{y'} \cap B_{\bar{r}}(y'') \neq \emptyset$, and so $\{y', y''\} \in \cl(\st(e)) \setminus \st(e)$.
Thus, the simplices $\{p, y'\}$, $\{q, y'\}$, $\{y', y''\}$, $\{p, y', y''\}$, and $\{q, y', y''\}$ are all in $\cl(\st(e)) \setminus \st(e)$, because $y' \in \linkY$ (i.e. $\{p, y'\}$, $\{q, y'\} \in \KC{i-1}$) and $\{p, y''\}$, $\{q, y''\}$, $\{y', y''\}$ are strictly shorter than $2\bar{r}$.
To conclude, we define 
$$K_{j}
=
K_{j-1} \cup \{y'\}
\cup 
\{p, y'\}
\cup 
\{q, y'\}
\cup
\{y', y''\}
\cup 
\{p, y', y''\}
\cup 
\{q, y', y''\}.
$$

At each step $K_{j-1} \subseteq K_j$ for $2 \leq j \leq |\linkY| - 1$, the edges $\{p, y'\}$ and $\{q, y'\}$ are free faces of $K_j$, which collapses on $K_{j-1}$.
Thus, after adding the set $\linkY$ of vertices of $\cl(\st(e)) \setminus \st(e)$ in $K_1$, we obtain a complex $K_{|\linkY|-1}$ which still has trivial reduced homology.

Then, we define the steps $K_{j-1} \subseteq K_{j}$ for $|\linkY| \leq j \leq n_j - 1$ by adding a single edge $\{y' ,y''\}$ among those in $\cl(\st(e)) \setminus \st(e)$ but not yet in $K_{j-1}$.
In particular, we set
$$
K_j = K_{j-1} \cup \{y', y''\} \cup \{p, y', y''\},
$$
where $\{p, y', y''\}$ can be added because both $\{p, y'\}$ and $\{p, y''\}$ were added in previous steps.
So $\{y', y''\}$ is a free face of $K_j$, which collapses on $K_{j-1}$ for each $|\linkY| \leq j \leq n_j -1$, and we have that $K_{n_j-1}$ has trivial reduced homology.

In the final step $K_{n_j-1} \subseteq K_{n_j}$, we add all the simplices in $K_{n_j} = \cl(\st(e)) \setminus \st(e)$ which are not in $K_{n_j-1}$.
As $K_{n_j-1}$ contains all the vertices and edges of $K_{n_j}$ by definition, in the final step we only add triangles and higher-dimensional simplices in $K_{n_j}$.
These new simplices cannot affect degree-zero homology. 
On the other hand, they could affect degree-one homology by deleting classes in $\tilde{H}_1(K_{n_j-1})$.
But this cannot happen because $\tilde{H}_1(K_{n_j-1})$ is already trivial.  
So, we have that the reduced homology in degrees zero and one of $K_{n_j}$ is trivial.

The proof follows from the exactness of the reduced Mayer-Vietoris sequence as mentioned above, and the fact that isomorphisms in reduced homology translate into isomorphisms in non-reduced homology.
\end{proof}

\clearpage
\section{Notation}
\label{app:notation}

\begin{itemize}
    \item[] \textbf{Preliminaries}
    \item $\st(\tau)$ star of $\tau \subseteq K$
    \item $B_r(p)$ open ball, $\overline{B_r(p)}$ closed ball, $\partial \overline{B_r(p)}$ boundary of closed ball.
    \item $V_p$ $\linf$-Voronoi region, $K^D$ $\linf$-Delaunay complex.
    \item[] \textbf{Persistent Homology}
    \item $K_{\mathcal{R}}$ filtration of $K$ parameterized by $\mathcal{R} = \{r_i\}_{i=1}^m$.
    \item $\modulek$ is the $k$-th persistence module of $K_{\mathcal{R}}$.
    \item $\Frips$ Vietoris-Rips filtration, $\Fcech$ \v{C}ech filtration, $\Falpha$ Alpha filtration.
    \item[] \textbf{$\linf$-Delaunay Edges}
    \item Given $\sigma \subseteq S$, $\bar{r} = \frac{\diamInfty{\sigma}}{2}$.
    \item $z$ is witness of $\sigma$ if $z \in \bigcap_{p\in \sigma} V_p$ and $\din(z, p) = \bar{r}$ for each $p \in \sigma$.
    \item $\Z{\sigma}$ is the set of witness points of $\sigma$
    \item $A_{\sigma}^{\bar{r}+\varepsilon} 
           = 
           \bigcap_{p \in \sigma}
           \partial \overline{B_{\bar{r}+\varepsilon}(p)}$
           for $\varepsilon \geq 0$ and $\sigma \subseteq S$.
    \item Given an edge $e = \{p, q\}$, then $\A^{\bar{r}} = 
    \partial \overline{B_{\bar{r}}(p)} 
    \cap 
    \partial \overline{B_{\bar{r}}(q)}
    =
    \overline{B_{\bar{r}}(p)} 
    \cap 
           \overline{B_{\bar{r}}(q)}$
           is a non-empty box.
           If a witness of $e=\{p,q\}$ exists, then it must be in $\A^{\bar{r}}$ and not in $\bigcup_{y \in S\setminus \{p, q\}} B_{\bar{r}}(y)$.
    \item[] \textbf{Alpha Flag and Minibox Complexes}
    \item $\minipq = 
                   \prod_{i=1}^d 
                   \big( \min\{p_i, q_i\}, \max\{p_i, q_i\} \big)$.
    \item $\Fflag$ Alpha flag filtration, $\Fmini$ Minibox filtration.
\end{itemize}

\end{document}